\documentclass[11pt]{amsart}
\usepackage[foot]{amsaddr}
\usepackage{ifxetex}
\ifxetex
\usepackage[no-math]{fontspec}
\else
\fi
\usepackage{amsmath}
\usepackage{amsfonts}
\usepackage{amssymb}
\usepackage{amsthm}
\usepackage{fullpage}
\usepackage[lining,semibold,type1]{libertine} 
\usepackage[T1]{fontenc} 
\usepackage{textcomp} 
\usepackage[varqu,varl]{inconsolata}
\usepackage[libertine,vvarbb]{newtxmath}
\usepackage[scr=rsfso]{mathalfa}
\usepackage{bm}

\usepackage{listings}
\usepackage{hyperref, color}
\hypersetup{colorlinks=true,citecolor=blue, linkcolor=blue, urlcolor=blue}
\usepackage[linesnumbered,boxed,ruled,vlined]{algorithm2e}
\usepackage{bm}
\usepackage{bbm}
\usepackage[numbers]{natbib}
\usepackage{xcolor}
\usepackage{enumerate} 
\usepackage{enumitem}
\usepackage{tabularx}
\usepackage{array}
\usepackage{cleveref}
\usepackage{mathrsfs}
\usepackage{tcolorbox}

\newcolumntype{L}[1]{>{\raggedright\arraybackslash}p{#1}}
\newcolumntype{C}[1]{>{\centering\arraybackslash}m{#1}}
\newcolumntype{R}[1]{>{\raggedleft\arraybackslash}p{#1}}
  
\usepackage{makecell}

\usepackage{footnote}
\makesavenoteenv{tabular}

\newcommand{\DTV}[2]{d_{\mathrm{TV}}\left({#1},{#2}\right)}

\renewcommand{\epsilon}{\varepsilon}

\newcommand{\tmix}{T_{\mathsf{mix}}}

\newcommand{\ExtVec}[2]{{#1}^{#2}}
\newcommand{\0}{-1}
\newcommand{\1}{+1}

\newtheorem{theorem}{Theorem}[section]
\newtheorem{observation}[theorem]{Observation}
\newtheorem{claim}[theorem]{Claim}
\newtheorem*{claim*}{Claim}

\newtheorem{example}[theorem]{Example}
\newtheorem{fact}[theorem]{Fact}
\newtheorem{lemma}[theorem]{Lemma}
\newtheorem{proposition}[theorem]{Proposition}
\newtheorem{corollary}[theorem]{Corollary}
\theoremstyle{definition}

\newtheorem{definition}[theorem]{Definition}
\newtheorem{remark}[theorem]{Remark}
\newtheorem*{remark*}{Remark}

\newtheorem{question}{Question}
\newtheorem{case}{Case}
\newtheorem*{case*}{Case}

\def\Pr{\mathop{\mathbf{Pr}}\nolimits}

\newcommand{\norm}[1]{\left\Vert#1\right\Vert}

\newcommand{\tuple}[1]{\left(#1\right)} 
\newcommand{\inner}[2]{\left\langle #1,#2\right\rangle}
\newcommand{\tp}{\tuple}

\newcommand{\abs}[1]{\left\vert#1\right\vert}

\newcommand{\spgap}[2]{\lambda_{\mathsf{#1}}^{\mathsf{#2}}}
\newcommand{\GD}{\mathsf{GD}}
\newcommand{\FD}{\mathsf{FD}}
\newcommand{\sgap}{min\text{-}gap}

\newcommand{\Mag}[2]{{#1}^{({#2})}}
\newcommand{\INF}[4]{\mathcal{I}_{#1}^{#2}({#3},{#4})}
\newcommand{\interior}[1]{%
  {\kern0pt#1}^{\mathrm{o}}%
}

\def\*#1{\boldsymbol{#1}} 
\def\+#1{\mathcal{#1}} 
\def\-#1{\mathrm{#1}} 
\def\^#1{\mathscr{#1}} 

\usepackage{xifthen}

\newcommand{\todo}[1]{\typeout{TODO: \the\inputlineno: #1}\textbf{{\color{red}[[[ #1 ]]]}}}

\renewcommand{\Pr}[2][]{ \ifthenelse{\isempty{#1}}
  {{\mathbf{Pr}}\left[#2\right]} {{\mathbf{Pr}}_{#1}\left[#2\right]} } 
\newcommand{\E}[2][]{ \ifthenelse{\isempty{#1}}
  {{\mathbf{\mathbf{E}}}\left[#2\right]}
  {{\mathbf{\mathbf{E}}}_{#1}\left[#2\right]} }
\newcommand{\Var}[2][]{ \ifthenelse{\isempty{#1}}
  {\mathbf{\mathbf{Var}}\left[#2\right]}
  {\mathbf{\mathbf{Var}}_{#1}\left[#2\right]} }

\newcommand{\OPr}[2][]{ \ifthenelse{\isempty{#1}}
  {\mathop{\mathbf{Pr}}\left[#2\right]} {\mathop{\mathbf{Pr}}_{#1}\left[#2\right]} } 
\newcommand{\OEp}[2][]{ \ifthenelse{\isempty{#1}}
  {\mathop{\mathbf{\mathbf{E}}}\left[#2\right]}
  {\mathop{\mathbf{\mathbf{E}}}_{#1}\left[#2\right]} }

\usepackage{xargs} 

\newcommand{\one}[1]{\textbf{1}\left[#1\right]}

\newcommand{\Rd}{\mathsf{Trans}}
\newcommand{\relaxT}[2][]{
  \ifthenelse{\isempty{#2}}
  {t_{\mathrm{rel}}^{\mathrm{#1}}}
  {t_{\mathrm{rel}}^{\mathrm{#1}}(#2)}
}
\newcommand{\Plim}[1][]{
  \ifthenelse{\isempty{#1}}
  {P^{\mathsf{FD}}_{\theta}}
  {P^{\mathsf{FD}}_{#1}}
}

\newcommandx{\Instance}[4][1= , 2=, 3= , 4= ]{
  \ifthenelse{\isempty{#1}}
  {\tp{V, E,
      \ifthenelse{\isempty{#3}}
      {\tp{\lambda_v}_{v\in V}}
      {\*\lambda},
      \ifthenelse{\isempty{#4}}
      {\tp{A_e}_{e \in E}}
      {\*A}
    }
  }
  {
    \ifthenelse{\isempty{#2}}
    {
      \tp{V#1, E#1,
        \ifthenelse{\isempty{#3}}
        {\tp{\lambda#1_{v}}_{v \in V#1}}
        {\*\lambda#1},
        \ifthenelse{\isempty{#4}}
        {\tp{A#1_{e}}_{e\in E#1}}
        {\*A#1}
      }
    }
    {
      \tp{V#1, E#1,
        \ifthenelse{\isempty{#3}}
        {\tp{\lambda#1_{v#2}}_{v#2 \in V#1}}
        {\*\lambda#1},
        \ifthenelse{\isempty{#4}}
        {\tp{A#1_{e#2}}_{e#2\in E#1}}
        {\*A#1}
      }
    }
  }
}

\newcommandx{\Inst}[2][1=, 2= ]{\Instance[][][#1][#2]}

\newcommand{\Pproj}[1][]{P^{\mathrm{proj}}_{#1}}
\newcommandx{\HyperGeo}[3][1 = \ell, 2=V, 3=k]{\Pi_{#2, #3, #1}}

\def\p:#1{#1^\star} 
\def\b:#1{#1^\circ} 
\def\r:#1{#1^\ast} 

\newcommand{\TSAW}{T_{\mathrm{SAW}}}
\newcommand{\ALO}[2]{I^{#2}_{#1}}

\title{Rapid mixing of Glauber dynamics via spectral independence for all degrees}
\date{}
\author{Xiaoyu Chen\textsuperscript{1}}
\author{Weiming Feng\textsuperscript{1 2}}
\author{Yitong Yin\textsuperscript{1}}
\author{Xinyuan Zhang\textsuperscript{1}}

\address{1: State Key Laboratory for Novel Software Technology, Nanjing University, 163 Xianlin Avenue, Nanjing, Jiangsu Province, China.}
\address{2: School of Informatics, University of Edinburgh, Informatics Forum, Edinburgh, United Kingdom.}
\address{\textnormal{E-mails: \url{chenxiaoyu233@smail.nju.edu.cn}, \url{wfeng@ed.ac.uk}, \url{yinyt@nju.edu.cn}, \url{zhangxy@smail.nju.edu.cn}}}

\thanks{
  This research was supported by the National Key R\&D Program of China 2018YFB1003202.
  Weiming Feng is supported by funding from the European Research Council (ERC) under the European Union's Horizon 2020 research and innovation programme (grant agreement No.~947778).
}

\begin{document}
\begin{abstract}
We prove an optimal $\Omega\tp{n^{-1}}$ lower bound on spectral gap of the Glauber dynamics for anti-ferromagnetic two-spin systems with $n$ vertices in the tree uniqueness regime.
This spectral gap holds for any, including unbounded, maximum degree $\Delta$.
%
Consequently, we have the following mixing time bounds for the models satisfying the uniqueness condition with a slack $\delta\in(0,1)$: 
\begin{itemize}
\item $C(\delta) n^2\log n$ mixing time for the hardcore model with fugacity $\lambda\le (1-\delta)\lambda_c(\Delta)= (1-\delta)\frac{(\Delta - 1)^{\Delta - 1}}{(\Delta - 2)^\Delta}$;
\item $C(\delta) n^2$ mixing time for the Ising model with edge activity $\beta\in\left[\frac{\Delta-2+\delta}{\Delta-\delta},\frac{\Delta-\delta}{\Delta-2+\delta}\right]$;
\end{itemize}
where the maximum degree $\Delta$ may depend on the number of vertices $n$, and $C(\delta)$ depends only on $\delta$.
%

Our proof is built upon the recently developed connections between the Glauber dynamics for spin systems and the high-dimensional expander walks. 
In particular,
we prove a stronger notion of spectral independence, called the \emph{complete spectral independence}, and use a novel Markov chain called the \emph{field dynamics} to connect this stronger spectral independence to the rapid mixing of Glauber dynamics for all degrees.
\end{abstract}

\maketitle



\setcounter{tocdepth}{1}
\tableofcontents
\setcounter{page}{0} \thispagestyle{empty} 
\newpage

\section{Introduction}
Spin systems 
are basic graphical models for high-dimensional joint distributions expressed by pairwise interactions,
and have been extensively studied in theoretical computer science, probability theory, and statistical physics.
%
A \emph{two-spin system} is specified on an undirected graph $G=(V,E)$ by three  real parameters $\beta,\gamma,\lambda\ge 0$.
Without loss of generality, one can assume that $0\le \beta\le \gamma$, $\gamma>0$ and $\lambda>0$.
%
Each \emph{configuration} $\sigma\in\{\0,\1\}^V$ assigns every vertex $v\in V$ one of the two \emph{spin states} from $\{\0,\1\}$. 
A probability distribution $\mu$ over all configurations $\sigma\in\{\0,\1\}^V$, called \emph{Gibbs distribution}, is defined as:
\[
\mu(\sigma)=
\frac{1}{Z}\beta^{m_{\1}(\sigma)} \gamma^{m_{\0}(\sigma)}\lambda^{n_{\1}(\sigma)},
\]
where $m_{i}(\sigma) \triangleq \abs{\{(u,v) \in E\mid \sigma_u = \sigma_v = i\}}$ denotes the number of $i$-monochromatic edges for $i\in\{\0,\1\}$ and $n_{\1}(\sigma) \triangleq \abs{\{v \in V \mid \sigma_v = \1\}}$ denotes the number of $\1$-spin vertices, and the normalizing factor 
\[
Z=\sum_{\sigma\in\{\0,\1\}^V}\beta^{m_{\1}(\sigma)} \gamma^{m_{\0}(\sigma)}\lambda^{n_{\1}(\sigma)}
\]
gives the \emph{partition function}.
%
%
The two-spin system is called \emph{ferromagnetic} if $\beta\gamma>1$ and is called \emph{anti-ferromagnetic} if $\beta\gamma<1$.
%
%
In particular, two extensively studied classes of two-spin systems are:
\begin{itemize}
\item 
The \emph{hardcore model} with \emph{fugacity} $\lambda$, which corresponds to two-spin systems with $\beta=0$ and $\gamma=1$. 
Every configuration $\sigma$ with $\mu(\sigma)>0$ corresponds to an independent set $I_{\sigma}$ in $G$, 
and 
$\mu(\sigma)\propto\lambda^{|I_\sigma|}$.
\item 
The \emph{Ising model} with \emph{edge activity} $\beta$ and \emph{external field} $\lambda$, which corresponds to two-spin systems with $\beta=\gamma$. 
The Gibbs distribution becomes $\mu(\sigma)\propto\beta^{m(\sigma)}\lambda^{n_{\1}(\sigma)}$, where $m(\sigma)=m_{\0}(\sigma)+m_{\1}(\sigma)$ gives the number of monochromatic edges.
\end{itemize}

A phenomenon  that has drawn considerable attention of two-spin systems is the computational phase transition for sampling.
In a seminal work~\cite{Wei06}, 
by exploiting a phase transition property based on decay of correlation, known as the \emph{spatial mixing},
Weitz showed that sampling from the Gibbs distribution $\mu$ of the hardcore model with fugacity $\lambda < \lambda_c(\Delta)$ on any $n$-vertex graph of bounded maximum degree $\Delta$ is tractable in time $n^{O(\log\Delta)}$.
%
Here the critical threshold $\lambda_c(\Delta) \triangleq \frac{(\Delta - 1)^{\Delta - 1}}{(\Delta - 2)^\Delta}$ is a famous threshold for the uniqueness/non-uniqueness phase transition for the hardcore model on the infinite $\Delta$-regular tree~\cite{kelly1985stochastic},
beyond which, i.e.~when $\lambda>\lambda_c(\Delta)$, 
the infinite-volume Gibbs measure on the $\Delta$-regular tree is not uniquely defined,
and approximately sampling from the hardcore model on graphs with bounded maximum degree $\Delta$ becomes computationally intractable~\cite{sly2010computational,sly2012computational,GSV16}.
%

This sharp computational phase transition was extended to all anti-ferromagnetic two-spin systems of maximum degree $\Delta$ specified by $(\beta,\gamma,\lambda)$ that is \emph{up-to-$\Delta$ unique}~\cite{lly12, LLY13, SST14},
which corresponds to the uniqueness condition for infinite regular trees up to degree $\Delta$.
And if $(\beta,\gamma,\lambda)$ lies outside this regime, the sampling problem becomes computationally intractable~\cite{sly2012computational,galanis2015inapproximability}.
Similar bounds were also achieved by another family of critical algorithms based on the polynomial interpolation approach for approximating non-vanishing polynomials~\cite{peters2019conjecture,liu2019fisher,SS20}.
%
%
%
%
A glaring issue with both these families of critical algorithms is the high time cost, usually in a form of $n^{O(\log\Delta)}$, 
which is due to enumerating $O(\log n)$-sized local structures.
Such time complexity grows super-polynomially in the size $n$ of the model when the maximum degree $\Delta$ is unbounded.

A major open question is whether sampling from spin systems is always tractable for the class of instances 
within the uniqueness regime, 
which does not by any means restrict to the graphical models with universally bounded max-degrees.
%
We wonder whether such \emph{fixed-parameter tractable} algorithms exist:
\begin{question}\label{question-FPT}
Let $\delta\in(0,1)$ be an arbitrary gap.
Can we approximately sample from the hardcore model on any $n$-vertex graph $G$ of maximum degree $\Delta_G$ with fugacity $\lambda\le(1-\delta)\lambda_c(\Delta_G)$ in time $f(\delta)\cdot\mathrm{poly}(n)$?
\end{question}
%
For anti-ferromagnetic two-spin systems, the condition $\lambda\le(1-\delta)\lambda_c(\Delta)$ is generalized by the regime of $(\beta,\gamma,\lambda)$ that is \emph{up-to-$\Delta$ unique with gap $\delta$},
which corresponds to the interior (determined by the gap $\delta$) of the regime for the uniqueness condition on infinite regular trees up to degree $\Delta$.

There is a canonical Markov chain based algorithm for sampling from Gibbs distributions known as the \emph{Glauber dynamics} ({a.k.a}~\emph{heat bath}, \emph{Gibbs sampling}).
The Glauber dynamics for a joint distribution $\mu$ of variables from $V$ is a Markov chain $\tp{X_t}_{t \geq 0}$ on space $\Omega(\mu)$, where $\Omega(\mu)$ denotes the the support of $\mu$.
At the $t$-th step, the rule for updating $X_t$ is:
\begin{itemize}
\item pick a $v\in V$ uniformly at random;
\item update $X_t(v)$ according to $\mu$ projected onto $v$ given the boundary condition ${X_t(V\setminus\{v\})}$.
\end{itemize}
The chain is stationary and reversible at $\mu$~\cite{levin2017markov}.
The rate of convergence is given by the \emph{mixing time}:
\begin{align*}
  \forall 0 < \epsilon < 1, \quad T_{\mathrm{mix}}(\epsilon) \triangleq \max_{X_0 \in \Omega(\mu) } \min \left\{t \mid \DTV{X_t}{\mu} \leq \epsilon\right\},
\end{align*}
where $\DTV{X_t}{\mu}$ denotes the {total variation distance} between the distribution of $X_t$ and $\mu$.

Let $P$ denote the transition matrix of the Glauber dynamics, which is positive semidefinite~\cite{DGU14}.
Then $P$ has non-negative real eigenvalues $1 = \lambda_1 \geq \lambda_2 \geq \cdots \geq \lambda_{|\Omega(\mu)|} \geq 0$.
The \emph{spectral gap}  is defined by 
$\spgap{gap}{}(P) \triangleq 1 - \lambda_2$.
Given the spectral gap, the mixing time is bounded as:
\begin{align}
T_{\mathrm{mix}}(\epsilon) \leq \frac{1}{\spgap{gap}{}(P)} \log \tp{\frac{1}{\epsilon\cdot\mu_{\min}}}, \quad\text{where } \mu_{\min}\triangleq\min_{\sigma \in \Omega(\mu)}\mu(\sigma).	\label{eq:mixing-time-spectral-gap}
\end{align}

It was widely speculated that for a wide range of spin systems, the computational phase transition for sampling is captured by the rapid mixing of Glauber dynamics 

Using coupling based techniques, one can obtain $O(n\log n)$ mixing time bounds. 
However, previous critical results using coupling methods held either by assuming girth lower bounds~\cite{hayes2006coupling,efthymiou2019convergence}
or by assuming ferromagnetism and bounded maximum degree~\cite{mossel2013exact}.

In a series of breakthrough works~\cite{ALOV19,alev2020improved, CGM19}, the Glauber dynamics was interpreted as a higher order random walk on simplicial complexes, and techniques for high-dimensional expander walks were applied to analyze its mixing.
In particular, Alev and Lau~\cite{alev2020improved} gave a sharp ``local-to-global'' argument for lifting the spectral expansions from a local down-up walk to the high-dimensional expander walk.
In a seminal work~\cite{anari2020spectral}, Anari, Liu and Oveis Gharan formulated a key concept, called the \emph{spectral independence},
which is measured by the spectral radius of the influence matrix,
in which each entry $(u,v)\in V^2$ gives the influence of $u$'s spin on the marginal probability at $v$ in the Gibbs distribution.
%
%
For product distributions,
this value is 0.
Intuitively, it measures how variables are independent of each other in a joint distribution.
The concept intrinsically connects the spatial mixing properties with the local spectral expansions of high-dimensional walks.
Then by utilizing the ``local-to-global''  result of \cite{alev2020improved}, Anari, Liu and Oveis Gharan~\cite{anari2020spectral} proved an $n^{\exp(O(1/\delta))}$ mixing time bound for the hardcore model with fugacity $\lambda \leq (1-\delta)\lambda_c(\Delta)$.
In \cite{chen2020rapid}, Chen, Liu and Vigoda proved tight bounds on spectral independence from spatial mixing properties, and consequently gave an improved $n^{O(1/\delta)}$ mixing time bound for all anti-ferromagnetic two-spin systems that are up-to-$\Delta$ unique with gap $\delta$, that is, when the uniqueness condition is satisfied with a slack $\delta$.

This series of breakthroughs culminated in a $C(\delta,\Delta)n\log n$ mixing time bound for all anti-ferromagnetic two-spin systems satisfying the up-to-$\Delta$ uniqueness with gap $\delta$ by Chen, Liu and Vigoda~\cite{chen2020optimal},
which was obtained by an ingenious  ``local-to-global'' argument for the relative entropy decay (modified log-Sobolev constant) in  a multi-level down-up walk that corresponds to heat-bath block dynamics.
%
%
It gave an optimal $O(n\log n)$ mixing time when the max-degree $\Delta=O(1)$, although in general the bound may grow like $\Delta^{O(\Delta^2/\delta)}n\log n$.
As observed in a very recent work~\cite{jain2021spectral}, a variance-decay variant of~\cite{chen2020optimal} could in fact imply a $\tilde{O}(\Delta^{O(1/\delta)}n^{2})$ mixing time.
All in all, the previous best bound on the mixing time for classes of near-critical instances without degree or girth restriction is $n^{O(1/\delta)}$.

Now look back to the problems like
\Cref{question-FPT}.
The major unresolved instances are the ones that have very large degree (e.g.~$\Delta=n^{\Omega(1)}$) and contain many small cycles.
%
Indeed for such instances, it is unknown whether local Markov chains as Glauber dynamics should mix rapidly, 
or  whether there are hard instances that prevent efficient sampling
while
only the sub-criticality of the system is fixed.
\subsection{Results for two-spin systems}\label{section-results-2spin}

We show that large degree and small cycles do not slow down mixing.
In particular, we give an optimal lower bound for the spectral gap of the Glauber dynamics for all anti-ferromagnetic two-spin systems within the uniqueness regime. 
\begin{theorem}
\label{theorem-2pin-gap}
For all $\delta\in(0,1)$, there exists a finite $C(\delta)>0$ such that for every anti-ferromagnetic two-spin system on an $n$-vertex graph $G=(V,E)$ with maximum degree $\Delta=\Delta_G\ge 3$ specified by $(\beta,\gamma,\lambda)$ that is up-to-$\Delta$ unique with gap $\delta$, 
the spectral gap of the Glauber dynamics is 
\[
\spgap{gap}{} \geq 	\frac{1}{C(\delta)n}.
\]
\end{theorem}
\begin{remark}
The constant $C(\delta)$ in the theorem depends only on the gap $\delta$, and is independent of the maximum degree $\Delta$ and the parameters $(\beta,\gamma,\lambda)$.
This gives an optimal spectral gap $\Omega(n^{-1})$ for all anti-ferromagnetic two-spin systems satisfying the uniqueness condition with a constant gap, regardless of the maximum degree $\Delta$.
More precisely, $C(\delta)$ is bounded as $C(\delta)=(\frac{1}{\delta})^{O(1/\delta)}$, and for the hardcore or Ising models, $C(\delta)$ can be further improved to $C(\delta)=\exp(O(\frac{1}{\delta}))$.
\end{remark}
Due to the well known relation between the spectral gap and the mixing time in~\eqref{eq:mixing-time-spectral-gap},
Theorem~\ref{theorem-2pin-gap} implies the mixing time bound for the same class of spin systems:
\begin{align}
\tmix(\epsilon) \leq C(\delta)n \tp{n\log\tp{\lambda+\frac{1}{\lambda}}+n\Delta\log \alpha + \log \frac{1}{\epsilon}},
\label{eq:2pin-mixing-time}
\end{align}
where $C(\delta)= (\frac{1}{\delta})^{O(1/\delta)}$ is the same factor as in Theorem~\ref{theorem-2pin-gap}, $\alpha = (\frac{1}{\beta}+2)$ when $\beta>0$ and
$\alpha = (\gamma+\frac{1}{\gamma}+2)$ when $\beta=0$. 
Therefore, the mixing time of the Glauber dynamics for anti-ferromagnetic 2-spin systems within the uniqueness regime is always bounded by $n^{3+o(1)}$, even when the maximum degree $\Delta$ or the parameters $(\beta,\gamma,\lambda)$ may depend on $n$, as long as the positive $\beta,\gamma,\lambda\in\exp\tp{n^{o(1)}}\cap\exp\tp{-n^{o(1)}}$.

For the hardcore model and the Ising model, we have even better bounds.
\begin{theorem}
\label{theorem-hardcore}
For all $\delta \in(0,1)$, there exists a $C(\delta)=\exp(O(\frac{1}{\delta}))$ such that 
for every hardcore model on an $n$-vertex graph $G=(V,E)$ with maximum degree $\Delta=\Delta_G\ge 3$  and with fugacity $\lambda\le(1-\delta)\lambda_c(\Delta)$,
the spectral gap and the mixing time of the Glauber dynamics are respectively bounded as
\[
\spgap{gap}{} \geq 	\frac{1}{C(\delta)n}
\quad\text{ and }\quad
\tmix(\epsilon) \leq C(\delta)n\left( n \log \Delta + \log \frac{1}{\epsilon}\right).
\] 
\end{theorem}

As far as we know, this is the first fixed-parameter-tractable (FPT) result for  the hardcore model within the uniqueness regime, where the gap $\delta$ to the critical threshold is the only fixed parameter.
Previously, for the same hardcore uniqueness regime, \cite{chen2020rapid} gave an $n^{O(1/\delta)}$  and~\cite{chen2020optimal} gave a $\Delta^{O(\Delta^2/\delta)}n\log n$ 
upper bounds on the mixing time.
Here we give a $\tilde{O}(n^2)$ mixing time bound, which is always bounded by a polynomial of absolute constant degree, no matter how large  the graph maximum degree $\Delta$ is.
And unlike the mixing results in~\cite{hayes2006coupling,efthymiou2019convergence,jain2021spectral}, our result needs not to assume a girth lower bound.

\begin{theorem}
\label{theorem-Ising}
For all $\delta \in(0,1)$, there exists a $C(\delta)=\exp(O(\frac{1}{\delta}))$ such that for every Ising model on an $n$-vertex graph $G=(V,E)$ with maximum degree $\Delta=\Delta_G\ge 3$  and with edge activity  $\beta\in\left[\frac{\Delta-2+\delta}{\Delta-\delta},\frac{\Delta-\delta}{\Delta-2+\delta}\right]$ and external field $\lambda>0$, 
the spectral gap and the mixing time of the Glauber dynamics are respectively bounded as
\[
\spgap{gap}{} \geq 	\frac{1}{C(\delta)n}
\quad\text{ and }\quad
\tmix(\epsilon) \leq C(\delta)n\left( n + \log \frac{1}{\epsilon}\right).
\] 
\end{theorem}
Compared to~\cite{chen2020optimal}, which gave $\Delta^{O({1}/{\delta})}n\log n$ mixing time bound for the same Ising uniqueness regime, and \cite{mossel2013exact}, which gave $\exp(\Delta^{O({1}/{\delta})})n\log n$ mixing time bound for the ferromagnetic half of this regime, 
our mixing time is bounded by ${O}(n^2)$ for all, including the unbounded, maximum degrees $\Delta$, while previously the best known mixing time upper bound when $\Delta$ is unbounded was $n^{O(1/\delta)}$~\cite{chen2020rapid, chen2020optimal}.

\subsection{Results for spectrally independent joint distributions}


%
%
Let $V$ be a set of Boolean random variables, and $\mu$ a  distribution over $\{\0,\1\}^V$.
We use $\Omega(\mu)$ to denote the support of $\mu$. 
A configuration $\tau \in \{\0,\1\}^V$ is \emph{feasible} if $\tau \in \Omega(\mu)$. 
For any subset $\Lambda \subseteq V$, let $\mu_{\Lambda}$ denote the marginal distribution on $\Lambda$ projected from $\mu$. 
A partial configuration $\tau \in \{\0,\1\}^\Lambda$, where $\Lambda\subseteq V$, is {feasible} if $\tau \in \Omega(\mu_{\Lambda})$.
For $\sigma_\Lambda \in \Omega(\mu_{\Lambda})$,  we use $\mu^{\sigma_{\Lambda}}$ to denote the distribution over $\{\0,\1\}^V$  induced by $\mu$ conditioned on the configuration on $\Lambda$ being fixed as $\sigma_\Lambda$.
Formally:
\begin{align*}
\forall \tau \in \{\0,\1\}^V,\quad \mu^{\sigma_\Lambda}_{}(\tau) = \Pr[\*X \sim \mu]{\*X = \tau \mid \*X_\Lambda = \sigma_\Lambda }.	
\end{align*}
For $S \subseteq V$, we use $\mu^{\sigma_\Lambda}_{S}$ to denote the marginal distribution on $S$ projected from $\mu^{\sigma_\Lambda}$,
and write $\mu^{\sigma_\Lambda}_{v}=\mu^{\sigma_\Lambda}_{\{v\}}$.

\begin{definition}[spectral independence]\label{definition-weight-tot-inf}
Let $\mu$ be a distribution over $\{\0,\1\}^V$.
For any $\Lambda \subseteq V$, any $\sigma_\Lambda \in \Omega(\mu_{\Lambda})$, 
the \emph{influence matrix} $\Psi^{\sigma_{\Lambda}}_{\mu} \in \mathbb{R}^{(V \setminus \Lambda) \times (V\setminus \Lambda)}_{\geq 0}$ is defined as: 
$\Psi^{\sigma_{\Lambda}}_{\mu}(u,u) = 0$ for all $u \in V \setminus \Lambda$; and for all distinct $u, v \in V \setminus \Lambda$,
\begin{align*}
\Psi^{\sigma_{\Lambda}}_{\mu}(u,v) \triangleq \max_{c,c' \in \Omega\tp{\mu^{\sigma_\Lambda}_u}} \DTV{\mu_v^{\sigma_\Lambda, u\gets c}}{\mu_v^{\sigma_\Lambda, u \gets c'}},
\end{align*}
where 
$\mu_v^{\sigma_\Lambda, u\gets c}$ denotes the marginal distribution $\mu_v^{\sigma_\Lambda}$ further conditioned on the value of $v$ being fixed as $c$,
and $\DTV{\cdot}{\cdot}$ denotes the total variation distance.

Let $\eta > 0$.
The distribution $\mu$ is said to be \emph{$\eta$-spectrally independent} if for any $\Lambda \subseteq V$, any $\sigma_\Lambda \in \Omega(\mu_{\Lambda})$, 
the spectral radius of the  influence matrix $\Psi^{\sigma_\Lambda}_{\mu}$ has
\begin{align}
\label{eq-sp-independent}
 \rho \tp{\Psi^{\sigma_\Lambda}_{\mu}} \leq \eta.
\end{align}
\end{definition}

\begin{remark}
The above definition is an alternative to the original definition of the spectral independence due to Anari, Liu and Oveis Gharan~\cite{anari2020spectral}, who defined the notion using the {signed} influence matrix $I^{\sigma_{\Lambda}}_{\mu}(u,v)\triangleq \mu^{\sigma_\Lambda, u \gets \1}_v(\1) - \mu^{\sigma_\Lambda, u \gets \0}_v(\1)$.
Here for some technical reasons (explained in~\Cref{section-mixing-field-dynamics-outline}), we adopt the notion of spectral independence proposed in~\cite{feng2021rapid} defined using the {absolute} influence matrix.
The two definitions are morally equivalent, and in spin systems they can both be implied by the spatial mixing property~\cite{anari2020spectral, chen2020rapid, chen2021rapid, feng2021rapid}.
\end{remark}

In a seminal work~\cite{anari2020spectral}, Anari, Liu and Oveis Gharan introduced the notion of spectral independence and proved a $n^{O(\eta)}$ mixing time bound for $\eta$-spectrally independent joint distributions of Boolean variables.
%
For Gibbs distributions specified by spin systems, in a recent major breakthrough~\cite{chen2020optimal} Chen, Liu and Vigoda proved a $C(\delta,\Delta)n \log n$ upper bound for the mixing time of Glauber dynamics
after establishing a tight $O({1}/{\delta})$-spectral independence~\cite{chen2020rapid}.  This mixing time bound is remarkably optimal when the max-degree $\Delta$ is bounded by a constant, however in general the $C(\delta,\Delta)$ factor can grow as fast as $\Delta^{O(\Delta^2/\delta)}$.

In this paper, we give a novel connection between the spectral independence and    rapid mixing of the Glauber dynamics, 
which can give us optimal spectral gap bounds without restricting the maximum degree of the graphical model.
To formally state our result, we  introduce the following notions.

\begin{definition}[magnetizing a joint distribution with local fields]
\label{defintion-magnetizing}
Let $\mu$ be a distribution over $\{\0,\1\}^V$. 
Let $\*\phi =(\phi_v)_{v \in V}$, where each $\phi_v >0$ specifies a \emph{local field} at $v$.
Denote by $\Mag{\mu}{\*\phi}$ the distribution 
obtained from imposing the local fields $\*\phi$ onto $\mu$.
Formally, $\pi=\Mag{\mu}{\*\phi}$ is a distribution over $\{\0,\1\}^V$ such that:
\begin{align*}
\forall \sigma \in \{\0,\1\}^V, \quad	
\pi(\sigma) 
\propto \mu(\sigma)\prod_{v \in V:\sigma_v = \1}\phi_v.
\end{align*}
In particular, if $\*\phi$ is a constant vector with $\phi_v = \theta$ for all $v \in V$ for some scalar $\theta>0$, we write  $\Mag{\mu}{\theta}=\Mag{\mu}{\*\phi}$.
\end{definition}

After magnetizing with local fields $\*\phi$, each variable $v\in V$ is locally biased towards $\1$ if $\phi_v > 1$ or towards $\0$ if $\phi_v<1$.
Without loss of generality, we consider only the case with $\phi_v\le 1$, which can cover the $\phi_v> 1$ case by flipping the roles of $\0$ and $\1$ for variable $v$.


\begin{definition}[complete spectral independence]\label{definition-complete-SI}
Let $\eta > 0$.
A distribution $\mu$ over $\{\0,\1\}^V$ is said to be \emph{completely $\eta$-spectrally independent} if $\Mag{\mu}{\*\phi}$ is $\eta$-spectrally independent for all $\*\phi \in (0,1]^V$.
\end{definition}
The notion captures a desirable situation:
when establishing the spectral independence, one usually proves stronger results, 
so that the spectral independence remains to hold for all smaller local fields.
%
For instance, in the hardcore model, decreasing $\lambda$ would only make the model more spectrally independent.
For general anti-ferromagnetic 2-spin systems, this becomes more complicated, 
nevertheless, due to~\Cref{lemma-good-direction}, the same holds after properly flipping the roles of \0 and \1 for variables.

%


%
Let $\mu$ be a distribution over $\{\0,\1\}^V$. 
We use $\spgap{gap}{\GD}(\mu)$ to denote the spectral gap of the Glauber dynamics for $\mu$.
We further consider the spectral gap {up to worst-case pinning}, which is defined as:
%
\begin{align}
\label{eq-def-worst-gap}
\spgap{\sgap}{\GD}(\mu) \triangleq \min_{\Lambda \subseteq V, \sigma_\Lambda \in \Omega(\mu_\Lambda)} \spgap{gap}{\GD}\tp{\mu^{\sigma_\Lambda}},
\end{align}
where $\spgap{gap}{\GD}\tp{\mu^{\sigma_\Lambda}}$ gives the spectral gap of Glauber dynamics for the joint distribution $\mu^{\sigma_\Lambda}$ over $\{\0,\1\}^V$ conditioned on $\sigma_\Lambda$, with convention that $\spgap{gap}{\GD}\tp{\mu^{\sigma_\Lambda}}=1$ if $\mu^{\sigma_\Lambda}$ has trivial support with $|\Omega(\mu^{\sigma_\Lambda})| = 1$.



%

\begin{theorem}[main technical theorem]
\label{theorem:main}
Let $\mu$ be a distribution over $\{\0,\1\}^V$, and $\eta > 0$.
If $\mu$ is completely $\eta$-spectrally independent, then for all $\theta\in(0,1)$, 
\begin{align*}
\spgap{gap}{\GD}(\mu) \geq \tp{\frac{\theta}{2}}^{2\eta + 7} \cdot \spgap{\sgap}{\GD}\left(\Mag{\mu}{\theta}\right). 
\end{align*}
\end{theorem}


The above theorem can be thought as a boosting theorem.
After magnetizing with external field $\theta$, the original near-critical $\mu$ is transformed to a much easier $\Mag{\mu}{\theta}$ whose spectral gap 
is either known or easy to bound using standard approaches.
Then \Cref{theorem:main} effectively boosts the mixing result for the easier distribution $\Mag{\mu}{\theta}$ up to the near-critical regime, with a $\theta^{O(\eta)}$ overhead which is determined by both the degree of spectral independence and the distance between the critical threshold and the easier regime.

\begin{example}[hardcore model]\label{example-hardcore-GD-mixing}
A generalization of the analyses in~\cite{chen2020rapid} shows that the Gibbs distribution $\mu$ of the hardcore model with fugacity $\lambda\le(1-\delta)\lambda_c(\Delta)$ is completely $O(\frac{1}{\delta})$-spectrally independent (formally stated in \Cref{lemma-good-direction}).
%
%
We then choose $\theta=\frac{1}{25}$, so that $\Mag{\mu}{\theta}$ corresponds to the hardcore model with a smaller fugacity 
\[
\theta\lambda=\frac{\lambda}{25}<\frac{(\Delta-1)^{\Delta-1}}{25(\Delta-2)^\Delta}\le\frac{1}{2\Delta},
\]
where the last inequality holds for all $\Delta\ge 3$.
In that regime, 
there is a coupling of the Glauber dynamics that decays step-wise, 
and consequently (see~\cite{MFC98} and~\cite[ch.~13.1]{levin2017markov}), the spectral gap $\spgap{\sgap}{\GD}\left(\Mag{\mu}{\theta}\right)\ge\frac{1}{2n}$.
Then by \Cref{theorem:main}, the Gibbs distribution $\mu$ of the hardcore model with fugacity $\lambda\le(1-\delta)\lambda_c(\Delta)$ has
\[
\spgap{gap}{\GD}(\mu) \geq \theta^{O(\frac{1}{\delta})} \cdot \spgap{\sgap}{\GD}\left(\Mag{\mu}{\theta}\right)
\ge \frac{1}{\exp(O(\frac{1}{\delta})) \cdot n}.
\]
Assuming $\frac{1}{2\Delta}\le\lambda<\lambda_c(\Delta)$ (the $\lambda<\frac{1}{2\Delta}$ case is dealt by coupling), it holds that $\mu_{\min}\ge(\frac{1}{8\Delta})^n$. Thus by~\eqref{eq:mixing-time-spectral-gap}, 
\[
\tmix(\epsilon) \leq 
\frac{1}{\spgap{gap}{\GD}(\mu)}\log\left(\frac{1}{\epsilon\mu_{\min}}\right)
\le\exp\left(O\left({1}/{\delta}\right)\right)\cdot n\left( n \log \Delta + \log \frac{1}{\epsilon}\right).
\]
\end{example}

\begin{remark}\label{remark-general-result-by-flipping}
For general anti-ferromagnetic two-spin systems, similar arguments apply, except that naively decreasing the external field $\lambda$ might not work.
Nevertheless, we observe that for every variable there always exists a good direction pointing to either $\0$ or $\1$ such that biasing the local field of the variable towards the good direction may only make the model easier.
This is formally stated in~\Cref{lemma-good-direction}.
Therefore, we can flip the roles of \0 and \1 for those variables whose good direction is pointing to \1,
which results in a new Gibbs distribution $\nu$ that is isomorphic to $\mu$, 
such that the spectral independence for anti-ferromagnetic two-spin systems proved in~\cite{chen2020rapid} in fact means the complete spectral independence of $\nu$.
Then the same argument as the hardcore model can apply. 
%
%
The details are given in~\Cref{section-spectral-gap-2-spin}.
\end{remark}


\section{The Field Dynamics and Proofs Outline}
\subsection{The field dynamics}
We introduce a novel Markov chain called \emph{the field dynamics},
which plays a key role in the proof of our main theorem (\Cref{theorem:main}) and may be of independent interests.

Given a distribution $\mu$ over $\{\0,\1\}^V$ and a scalar $\theta\in(0,1)$, 
the field dynamics for $\mu$ with parameter $\theta$, 
denoted by $\Plim$, is a Markov chain on space $\Omega(\mu)$ defined as follows.

\begin{center}
  \begin{tcolorbox}[=sharpish corners, colback=white, width=1\linewidth]
  	\begin{center}
	\textbf{\emph{The Field Dynamics}}
  	\end{center}
  	\vspace{6pt}
    The initial state is an arbitrary feasible configuration $\*\sigma \in \Omega(\mu)$;\\
    The rule for updating a configuration $\*\sigma\in\Omega(\mu)$ is:
    \begin{enumerate}
    \item generate a random $S\subseteq V$
    by independently selecting each $v \in V$ into $S$ with probability
    \begin{align*}
    p_v \triangleq \begin{cases}
 	1 &\text{if } \sigma_v = \0,\\
 	\theta & \text{if } \sigma_v = \1;
 	\end{cases}
	\end{align*}
	\item replace $\sigma_S$ by a random partial configuration sampled according to $\pi^{\sigma_{V \setminus S}}_S$, where $\pi = \Mag{\mu}{\theta}$.
    \end{enumerate}
  \end{tcolorbox} 
\end{center}
The field dynamics can be thought as an adaptive variant of the heat-bath block dynamics.
In each step, a block $S$ is randomly generated adaptively to the current configuration $\sigma$.
And the configuration on $S$ is resampled given the boundary condition $\sigma_{V\setminus S}$ according to a properly biased distribution $\pi=\Mag{\mu}{\theta}$ for canceling with the bias introduced in the adaptive construction of block $S$.

For instance, on the hardcore model, the field dynamics gives us the following novel Markov chain:
\begin{example}[the field dynamics for the hardcore model]
Let $\mu$ be the Gibbs distribution of the hardcore model on graph $G=(V,E)$ with fugacity $\lambda$.
The rule for updating the current $\sigma\in\Omega(\mu)$:
\begin{itemize}
\item 
each occupied vertex $v$ (with $\sigma_v=\1$) is picked independently with probability $1-\theta$; all picked vertices and their neighbors are removed from $V$ and let $\interior{S}$ denote the set of remaining vertices;
\item 
resample $\sigma_{\interior{S}}$ according to the Gibbs distribution $\pi$ of the hardcore model on $G[\interior{S}]$ with fugacity $\theta\lambda$.
\end{itemize}

\end{example}



The next proposition shows that the field dynamics $\Plim$ for $\mu$ has the stationary distribution $\mu$, and is always ergodic and reversible (hence having real eigenvalues).
The proposition is proved in \Cref{section-EF-stationary}.
\begin{proposition}
\label{theorem-field-dynamics-basic}
For every distribution $\mu$ over $\{\0,\1\}^V$, and for all $\theta \in (0,1)$, 
the field dynamics $\Plim$ for $\mu$ is irreducible, aperiodic and reversible with respect to $\mu$.
\end{proposition}


Let $\spgap{gap}{\FD}(\mu,\theta)$ denote the spectral gap of the field dynamics $\Plim$ for $\mu$. 
We state two key lemmas for the spectral gap of the field dynamics.
%
The first is a \emph{mixing lemma} that guarantees the fast mixing of the field dynamics  for $\mu$ assuming the complete spectral independence of $\mu$.
\begin{lemma}[field dynamics - mixing lemma]
\label{lemma-field-dynamics-mixing}
Let $\mu$ be a distribution over $\{\0,\1\}^V$ and $\eta > 0$.
If $\mu$ is completely $\eta$-spectrally independent, then for all $\theta\in(0,1)$, 
\begin{align*}
\spgap{gap}{\FD}(\mu,\theta) \geq \tp{\frac{\theta}{2}}^{2 \eta + 7}.	
\end{align*}
\end{lemma}

On the other hand, the following \emph{comparison lemma} relates the spectral gap of the field dynamics to that of the Glauber dynamics, which holds universally for all joint distributions $\mu$ with Boolean domain.

\begin{lemma}[field dynamics - comparison lemma]
\label{lemma-compare-field-Glauber}
Let $\mu$ be a distribution over $\{\0,\1\}^V$.
For all $\theta\in(0,1)$, 
\begin{align*}
\spgap{gap}{\GD}(\mu) \geq \spgap{gap}{\FD}(\mu,\theta) \cdot	\spgap{\sgap}{\GD}\left(\Mag{\mu}{\theta}\right).
\end{align*}
\end{lemma}

Our main technical result (\Cref{theorem:main}) follows immediately from \Cref{lemma-field-dynamics-mixing} and \Cref{lemma-compare-field-Glauber}.
In the next, we give outlines of the proofs of these two key lemmas.

\subsection{Mixing of field dynamics via spectral independence}
\label{section-mixing-field-dynamics-outline}
We now give an outline of our proof of the mixing lemma for the field dynamics (\Cref{lemma-field-dynamics-mixing}).
A key observation is that the field dynamics is in fact a limiting instance for the uniform block dynamics.

Let $\mu$ be a distribution over $\{\0,\1\}^V$. 
The \emph{uniform $\ell$-block dynamics} for $\mu$ is basically the heat-bath block dynamics for $\mu$ where in each step a subset of vertices of size $\ell$ is chosen uniformly at random and gets updated.
Specifically, the Markov chain is on space $\Omega(\mu)$ and in each step a subset $S\subseteq V$ of $\ell$ variable is chosen uniformly at random and the configuration on $S$ is updated according to $\mu$ conditional on the current configuration on $V\setminus S$.
In~\cite{chen2020optimal}, Chen Liu and Vigoda proved the following theorem for the rapid mixing of the uniform block dynamics from spectral independence.
\begin{theorem}[\cite{chen2020optimal}] \label{theorem:spectral-independent-imply-spectral-gap}
Let $\mu$ be a distribution over $\{\0, \1\}^V$, $n=|V|$ and $\eta > 0$.
If $\mu$ is $\eta$-spectrally independent, then for any   $2\lceil \eta \rceil \leq \ell \leq n$, the uniform $\ell$-block dynamics $P_\ell$ for $\mu$ has 
$\spgap{gap}{}\tp{P_\ell} \geq  \tp{\frac{\ell}{2n}}^{2\lceil \eta \rceil + 1}$.
\end{theorem}
This is a general result that holds for all joint distributions, and imposes no restriction on the maximum degree of the underlying model.
In fact, the entropy (modified log-Sobolev) variant of \Cref{theorem:spectral-independent-imply-spectral-gap} proved in~\cite{chen2020optimal} required small  degrees. 
The statement in \Cref{theorem:spectral-independent-imply-spectral-gap} is a consequence to the ``local-to-global'' argument for the variance contraction~\cite[Theorem A.9]{chen2020optimal}.
For completeness, we provide a proof of \Cref{theorem:spectral-independent-imply-spectral-gap} in \Cref{appendix-block-mixing} using our notion of spectral independence with absolute influence matrix.

Previously, rapid mixing was established for the Glauber dynamics by directly comparing it with the uniform block dynamics~\cite{chen2020optimal}, which resulted in a super-polynomial reliance on the max-degree.

Here, we describe a novel way to utilize the rapid mixing of uniform block dynamics.




\begin{definition}[$k$-transformation]
\label{definition-k-transformation}
Let $\mu$ be a distribution over $\{\0, \1\}^V$ and $k \geq 1$ an integer.
The \emph{$k$-transformation} of $\mu$, denoted by $\mu_k = \Rd(\mu, k)$, is a distribution over $\{\0,\1\}^{V\times[k]}$ defined as follows.

Let $\*X\sim\mu$. Then $\mu_k = \Rd(\mu, k)$ is the distribution of $\*Y\in\{\0,\1\}^{V\times [k]}$ constructed as follows:
\begin{itemize}
  \item if $X_v = \0$, then $Y_{(v,i)} = \0$ for all $i\in[k]$;
  \item if $X_v = \1$, then $Y_{(v,i^*)} = \1$ and $Y_{(v,i)} = \0$ for all $i\in[k]\setminus \{i^*\}$, where $i^*$ is chosen from $[k]$ uniformly and independently at random.
\end{itemize}
\end{definition}


We use $P_{k,\ell}$ to denote the uniform $\ell$-block dynamics for $\mu_k$, and  $\spgap{gap}{}(P_{k,\ell})$ its spectral gap.
%
Recall that $\spgap{gap}{\FD}(\mu,\theta)$ denotes the spectral gap of the field dynamics $\Plim$ for $\mu$ with parameter $\theta\in(0,1)$.
 
One of our key discoveries is that the field dynamics $\Plim$ for $\mu$ is a limiting instance for the uniform $\lceil \theta kn \rceil$-block dynamics running on the $k$-transformations of $\mu$. 
\begin{lemma}\label{lemma:block-dynamics-spectral-gap-implies-limit-chain-spectral-gap}
For all $\theta \in (0,1)$, it holds that
\begin{align*}
\spgap{gap}{\FD}(\mu,\theta)\ge \limsup_{k \to \infty}\spgap{gap}{}\tp{P_{k, \lceil \theta k n \rceil }}.
\end{align*}
\end{lemma}
The lemma is proved in Sections~\ref{section-ex-field-block-dynamics} and~\ref{section-analysis-projected-block-dynamics}. 
In a high level, it is proved as follows:
Each feasible $\sigma\in\Omega(\mu_k)$ in the $k$-transformed distribution $\mu_k$ can be naturally projected back to a feasible $\sigma^{\star}\in\Omega(\mu)$ in the original $\mu$, where $\sigma_v^{\star}$ indicates whether $\sigma_{(v,i)}=\1$ for some $i\in[k]$.
This also naturally projects the block dynamics $P_{k, \lceil \theta k n \rceil }$ to a chain on $\Omega(\mu)$ that essentially preserves the spectral gap. 
And more crucially, the projected chain 
gives an entry-wise approximation of the field dynamics $\Plim$.

Next, we prove a  Chen-Liu-Vigoda theorem (\Cref{theorem:spectral-independent-imply-spectral-gap}) for the $k$-transformed distributions.
\begin{lemma} \label{lemma:block-dynamics-spectral-gap}
Let $\mu$ be a distribution over $\{\0, \1\}^V$, $n=|V|$ and $\eta > 0$.
If $\mu$ is completely $\eta$-spectrally independent, 
then for any integers $k \geq 1$ and $2(\eta + 3) \leq \ell \leq kn$, it holds that 
    $\spgap{gap}{}\tp{P_{k, \ell}} \geq \tp{\frac{\ell}{2kn}}^{2\eta  + 7}$.
\end{lemma}
The lemma is proved in \Cref{sec:block-mixing}.
Note that \Cref{lemma:block-dynamics-spectral-gap} 
assumes only the spectral independence for $\mu$, 
while the block dynamics is for  $\mu_k$.
We show that assuming the \emph{complete} spectral independence of $\mu$, the $k$-transformation may only cause a constant additive overhead to the spectral independence.
This is the part where we are technically more convenient to work with absolute influence matrices rather than the signed ones, 
because of the monotonicity of spectral radius that holds for nonnegative matrices.
%
%

\Cref{lemma:block-dynamics-spectral-gap-implies-limit-chain-spectral-gap} and \Cref{lemma:block-dynamics-spectral-gap} together suffice to prove the mixing lemma for the field dynamics (\Cref{lemma-field-dynamics-mixing}).

\begin{proof}[Proof of \Cref{lemma-field-dynamics-mixing}]
Let $\theta \in (0, 1)$.
Let $(A_k)_{k \in \mathbb{N}}$ and $(B_k)_{k \in \mathbb{N}}$ be constructed as: 
  \begin{align*}
    A_k = \spgap{gap}{}\tp{P_{k, \lceil \theta kn \rceil}} \quad \text{ and }\quad 
    B_k = \begin{cases}
 	\tp{\frac{\lceil \theta kn \rceil}{2kn}}^{2\eta  + 7} &\text{if } k \geq \frac{2(\eta+3)}{\theta n}, \\
 	0 & \text{otherwise}. 
 \end{cases}
\end{align*}
By \Cref{lemma:block-dynamics-spectral-gap}, $2\geq A_k \geq B_k \geq 0$ for all $k\ge 1$.
We have that
  $\limsup_{k \to \infty} A_k \geq \limsup_{k \to \infty} B_k$.

Since $\lim_{k \to \infty} B_k = \tp{\frac{\theta}{2}}^{2 \eta + 7}$, by \Cref{lemma:block-dynamics-spectral-gap-implies-limit-chain-spectral-gap}, we have
  \begin{align*}
   \spgap{gap}{\FD}(\mu,\theta) \geq \limsup_{k \to \infty} \spgap{gap}{}\tp{P_{k, \lceil \theta kn \rceil}} &\geq \tp{\frac{\theta}{2}}^{2 \eta + 7}. \qedhere
\end{align*}
\end{proof}

\subsection{Comparing Glauber dynamics with field dynamics}
With the rapid mixing result we just proved for the field dynamics, we already have an efficient sampling algorithm for distribution $\mu$.
Suppose that we try to implement the field dynamics.
The nontrivial operation in each step of the field dynamics, namely the step of resampling from the marginal distribution induced by $\pi=\Mag{\mu}{\theta}$ on a randomly generated block, can be simulated by running a Glauber dynamics for $\pi$ on the block with boundary condition, which is known to be rapidly mixing because $\pi=\Mag{\mu}{\theta}$ lies in a much easier regime!

\begin{example}[a hardcore sampler]
For the hardcore model on graph $G$ with fugacity $\lambda\le(1-\delta)\lambda_c(\Delta_G)$,
as discussed in \Cref{example-hardcore-GD-mixing},
the Gibbs distribution $\mu$ is completely $O(\frac{1}{\delta})$-spectrally independent. 
We choose $\theta=\frac{1}{25}$. 
Due to \Cref{lemma-field-dynamics-mixing}, the  field dynamics $\Plim$ for $\mu$ has a spectral gap $1/C_{\delta}$ where $C_\delta=\exp(O(\frac{1}{\delta}))$, and thus mixes in $\tilde{O}(C_\delta\cdot{n})$ steps.
Meanwhile, each transition step of $\Plim$ is simulated as a subroutine by a Glauber dynamics for $\pi=\Mag{\mu}{\theta}$ that corresponds to a hardcore model with fugacity $\theta\lambda=\lambda/25<\frac{1}{2\Delta_G}$, which mixes in $\tilde{O}(n)$ steps due to path coupling.
Together this gives us a sampling algorithm for the hardcore model that runs in $\tilde{O}(C_\delta\cdot{n^2})$ steps of single-site updates regardless of the maximum degree as long as $\lambda\le(1-\delta)\lambda_c(\Delta_G)$.
\end{example}
For general anti-ferromagnetic two-spin systems, similar results hold as discussed in~\Cref{remark-general-result-by-flipping}.

This sampling procedure is different from the Glauber dynamics for $\mu$. 
Nevertheless, the rapid mixing of this new procedure might serve as a ``proof of concept'' for the rapid mixing of the Glauber dynamics.
Indeed, compared to the Glauber dynamics, in this new procedure, the $(\1)$-variables are less favored than $(\0)$-variables when being chosen to be resampled; but when being resampled, variables have higher chances to be updated to $\0$ than in the Glauber dynamics for $\mu$, because they are now being resampled according to $\pi=\Mag{\mu}{\theta}$ for $\theta\in(0,1)$. 
These two types of biases might be canceling each other. 
Conceptually, a rapidly mixing Glauber dynamics is somehow conceived in this thought experiment.


Such intuition is formally justified by a comparison of variance decays of the chains.
Given any function $f\in\mathbb{R}^{\Omega(\mu)}$, let $\Var[\mu]{f}$ and $\+E_P(f,f)$ respectively denote its variance and Dirichlet form (defined in \Cref{section-prelim-mixing-time}). It is well known that $\spgap{gap}{}(P)=\inf_{f}\frac{\+E_{P}(f,f)}{\Var[\mu]{f}}$ for any reversible chain $P$.

By comparing the Dirichlet forms of the field dynamics $\Plim$ and the Glauber dynamics $P_\mu$, we establish  
\begin{align}
\+E_{\Plim}(f,f) 
\le 
\frac{1}{\spgap{\sgap}{\GD}\left(\Mag{\mu}{\theta}\right)}\+E_{P_\mu}(f,f).\label{eq:comparison-Dirichlet-forms}
\end{align}
This is proved in \Cref{section-FD-comparison-lemma}, although there it is expressed as a tensorization of variance for $\mu$, which is equivalent to the comparison of Dirichlet forms.
The comparison lemma for the field dynamics (\Cref{lemma-compare-field-Glauber}) follows directly from \eqref{eq:comparison-Dirichlet-forms} and characterization of spectral gap $\spgap{gap}{}(P)=\inf_{f}\frac{\+E_{P}(f,f)}{\Var[\mu]{f}}$.

\subsection{Wrapping up}
\label{section-applications}
We now describe how to prove  the main theorem for anti-ferromagnetic two-spin systems (\Cref{theorem-2pin-gap}) by the main technical theorem (\Cref{theorem:main}).

In general, magnetizing a distribution $\mu$ over $\{\0,\1\}^V$ to $\Mag{\mu}{\theta}$ for a $\theta\in(0,1)$ does not necessarily make it easier for sampling. 
Nevertheless, for anti-ferromagnetic two-spin systems, this issue can be circumvented by flipping roles of $\1$ and $\0$ for certain variables. Given a \emph{direction vector} $\*\chi \in \{\0, \1\}^V$, the distribution $\nu = \mathsf{flip}(\mu,\*\chi)$ obtained from \emph{flipping} $\mu$ according to $\*\chi$ is defined as:
\begin{align*}
  \forall \sigma \in \{\0, \1\}^V, \quad \nu(\sigma) \triangleq \mu(\sigma \odot \*\chi),
\end{align*}	
where $\sigma \odot \*\chi \in \{\0,\1\}^V$ is defined as that $(\sigma \odot \*\chi)_v = \sigma_v \chi_v$ for all $v \in V$.

Now consider an anti-ferromagnetic two-spin system on graph $G = (V, E)$, specified by parameters $(\beta,\gamma,\lambda)$, where $0\le \beta\le \gamma$,  $\beta\gamma<1$, and $\gamma,\lambda>0$.
%
The good direction $\*\chi\in\{\0,\1\}^V$ is constructed as:
\begin{align} \label{eq:good-direction}
  \forall v\in V,\quad
  \chi_v 
  &\triangleq \begin{cases}
    \1 & \lambda \leq \tp{\frac{\gamma}{\beta}}^{\Delta_v / 2} \\
    \0 & \text{otherwise,}
  \end{cases}
\end{align}
where $\Delta_v$ denotes the degree of vertex $v$.

Assuming that the anti-ferromagnetic two-spin system is up-to-$\Delta$ unique with gap $\delta$ (\Cref{definition:up-to-Delta-unique}), 
the following theorem guarantees that after flipping according to the good direction $\*\chi$ in~\eqref{eq:good-direction},
the Gibbs distribution $\mu$ becomes completely $O(\frac{1}{\delta})$-spectrally independent, and if further every local field is biased by a $\Theta(\delta^2)$-factor, the Glauber dynamics is known to be rapidly mixing with spectral gap $\Omega(\frac{\delta}{n})$.

\begin{theorem}\label{lemma-good-direction}
For all $\delta\in(0,1)$, 
for every anti-ferromagnetic two-spin system on an $n$-vertex graph $G=(V,E)$ with maximum degree $\Delta=\Delta_G\ge 3$ 
that is up-to-$\Delta$ unique with gap $\delta$, 
the distribution $\nu = \mathsf{flip}(\mu,\*\chi)$ obtained from flipping the Gibbs distribution $\mu$ according to the direction $\*\chi$ defined in~\eqref{eq:good-direction}, 
satisfies:
\begin{itemize}
\item $\nu$ is completely $\frac{144}{\delta}$-spectrally independent;
\item for $\theta=\frac{\delta^2}{64}$ and $C=\frac{8}{\delta}$,  it holds that $\spgap{\sgap}{\GD}(\pi) \ge \frac{1}{Cn}$, where $\pi=\nu^{\tp{\theta}}$.
\end{itemize}
\end{theorem}

\Cref{theorem-2pin-gap} is an easy consequence of \Cref{lemma-good-direction}. 
\begin{proof}[Proof of \Cref{theorem-2pin-gap}]
By~\Cref{lemma-good-direction}, we can apply \Cref{theorem:main} to distribution $\nu= \mathsf{flip}(\mu,\*\chi)$, so that
  \begin{align*}
    \spgap{gap}{GD}(\nu) &\geq \tp{\frac{\delta^2}{128} }^{\frac{288}{\delta}+7} \cdot \spgap{\sgap}{\GD}\tp{\Mag{\nu}{\theta}} \geq \tp{\frac{\delta}{8\sqrt{2}} }^{\frac{590}{\delta}} \cdot \tp{\frac{\delta}{8n}}  = \frac{1}{C(\delta)\cdot n},\quad \text{where } C(\delta) = \tp{\frac{1}{\delta}}^{O\tp{\frac{1}{\delta}}}.
  \end{align*}
Note that distribution $\nu$ is isomorphic to the original Gibbs distribution $\mu$ 
and the transition matrices of the Glauber dynamics for $\mu$ and $\nu$ are equivalent up to a bijection, and thus have the same set of eigenvalues. Hence, 
\begin{align*}
 \spgap{gap}{GD}(\mu) = 	 \spgap{gap}{GD}(\nu) \geq \frac{1}{C(\delta) \cdot n}, \quad \text{where }C(\delta) = \tp{\frac{1}{\delta}}^{O\tp{\frac{1}{\delta}}}. &\qedhere
\end{align*}
\end{proof}

\Cref{lemma-good-direction} is proved in \Cref{section-spectral-gap-2-spin} as follows:
\begin{itemize}
\item
The first part regarding the complete spectral independence for the flipped distribution $\nu$ is proved by observing that the analyses in~\cite{LLY13} and~\cite{chen2020rapid} remain to hold when each local field $\lambda_v$ is further biased towards the direction indicated by the $\chi_v$ in~\eqref{eq:good-direction}.
\item
The second part regarding the spectral gap for the easier distribution $\pi=\nu^{\tp{\theta}}$ is due to the decay of path coupling in that easier regime and its implication to the spectral gap~\cite[ch.~13.1]{levin2017markov}.
\end{itemize}
More specifically, there is a good direction for the local fields such that: 
(1) the uniqueness, spatial mixing, and spectral independence properties are all closed in that direction; and (2) a natural relaxation of the uniqueness condition arises from standard path coupling once the system moves towards the good direction by a constant factor.
These discoveries may suggest that the notion of ``easier regime'' is perhaps more naturally described  in our fashion, as biasing the local fields towards its easier direction.




The mixing time bound in \eqref{eq:2pin-mixing-time} follows from~\eqref{eq:mixing-time-spectral-gap} and the following straightforward analysis of $\mu_{\min}$.
\begin{proof}[Proof of \eqref{eq:2pin-mixing-time}]
Define the marginal lower bound by
\begin{align*}
b \triangleq \min_{\Lambda \subseteq V, v \in V \setminus \Lambda} \min_{\sigma \in \Omega(\mu_\Lambda)} \min_{c \in \Omega(\mu_v^\sigma)}\mu_v^\sigma(c).
\end{align*}
Obviously $\mu_{\min}\ge b^n$ due to chain rule, hence $\log \frac{1}{\mu_{\min}}\le n\log \frac{1}{b}$.

When $\beta = 0$,  the marginal bound has $b \ge \min_{1 \le d \le \Delta} \left\{\frac{\gamma^d}{\gamma^d+\lambda}, \frac{\lambda}{\gamma^d+\lambda}\right\}$,
which implies $\frac{1}{b}\le (\lambda + \frac{1}{\lambda}) \tp{\frac{1}{\gamma}+\gamma+2}^{\Delta}$.

When $\beta > 0$, by considering the worst configuration of the neighborhood, we have 
\begin{align*}
b \geq \min_{1 \leq d \leq \Delta}\min_{0\leq s \leq d}	\min\left\{ \frac{\gamma^s}{\gamma^s + \lambda \beta^{d-s}}, \frac{\lambda \beta^{d-s}}{\gamma^s + \lambda \beta^{d-s}} \right\},
\end{align*}
which implies that
\begin{align*}
\frac{1}{b} \leq 
\left\{
\begin{array}{lll}
2 + \lambda + \frac{1}{\lambda} \beta^{-\Delta} 
&\leq (\frac{1}{\lambda}+\lambda)\tp{\frac{1}{\beta}+2}^{\Delta} 
&\text{if } \gamma \geq 1,\\
2 + \lambda  \gamma^{-\Delta} + \frac{1}{\lambda}\beta^{-\Delta} 
&\leq \tp{\frac{1}{\lambda}+\lambda}\tp{\frac{1}{\beta}+2}^\Delta 
&\text{if } \gamma < 1.
\end{array}
\right.
\end{align*}

Overall,
$\log\frac{1}{\mu_{\min}} \le n\log\tp{\lambda+\frac{1}{\lambda}} +n\Delta\log {\alpha}$, where 
\[
\alpha = 
\begin{cases}
\gamma+\frac{1}{\gamma}+2 & \text{if }\beta=0,\\
\frac{1}{\beta}+2 & \text{if }\beta>0.
\end{cases}
\]
The mixing time bound in \eqref{eq:2pin-mixing-time} follows from~\eqref{eq:mixing-time-spectral-gap}.
\end{proof}

For the hardcore mode and the Ising model, similar results hold with better choice of $\theta$ and tighter bounds of $\mu_{\min}$, which are proved in \Cref{section-hardcore-ising-mixing}. Hence we have \Cref{theorem-hardcore} and \Cref{theorem-Ising}. 

\subsection{Related work and open problems}
Complexity classification of two-spin systems is a fundamental problem \cite{jerrum1993polynomial,DFJ02,GJP03,Wei06,mossel2009hardness,sly2010computational,lly12,LLY13,mossel2013exact,galanis2014improved,liu2014complexity,SS14,SST14,GSV16,guo2018uniqueness,guo2020zeros,SS20}. 
Previously, the primary technique for relating the spatial mixing properties to the rapid mixing of  Glauber dynamics for spin systems is   coupling \cite{dyer2004mixing,goldberg2005strong,hayes2006coupling,mossel2013exact,efthymiou2019convergence}. 
In \cite{anari2020spectral},  techniques based on high-dimensional expanders \cite{kaufman2016high,oppenheim2018local,kaufman2018high,alev2020improved} were applied to bound the mixing time of Glauber dynamics for 2-spin systems through the notion of spectral independence. 
The results were substantially improved and strengthened to an optimal $O(n\log n)$ mixing time bound for $\Delta=O(1)$ in~\cite{chen2020rapid,chen2020optimal} through entropy factorization \cite{CMT15,caputo2020block}.
These were extended to general spin models beyond Boolean domain and/or more general class of Markov chains~\cite{chen2021rapid,feng2021rapid,chen2020optimal,blanca2021mixing,liu2021coupling}.

In this paper, we prove an optimal $\Omega(n^{-1})$ lower bound on the spectral gap of the Glauber dynamics for anti-ferromagnetic two-spin systems satisfying the uniqueness condition with a constant slack.
It leaves several open directions.
First, a major open problem is to prove an optimal $O(n \log n)$ mixing time bound for the same regime with no degree restriction.
A powerful tool for this goal is a ``local-to-global'' argument for relative entropy decays (modified log-Sobolev constants) as given in~\cite{chen2020optimal,guo2020local}.
However, a major obstacle for this approach 
is that the current analyses of the uniform block dynamics based on entropy decay
result in bounds that grow exponentially in the max-degree $\Delta$.
And even if this is resolved, another difficulty is to establish modified log-Sobolev inequalities in the ``easier'' regime.

Our techniques crucially rely on variables with Boolean domain.
It is important to extend our approach to general distributions with variables beyond Boolean domains, e.g.~proper $q$-colorings.

Finally, an open aspect is to optimize the reliance on $\delta$.
In fact, it should be restated as to improve the reliance of the mixing time on the spectral independence.
To see this is an important question, consider graph matchings, whose spectral independence is bounded by $O(\sqrt{\Delta})$, but so far in this entire line of research the dependency of mixing time on spectral independence is at least exponential.


\subsection{Organization of the paper}
The preliminaries are given in \Cref{section-prelim}.
The comparison lemma for the field dynamics (\Cref{lemma-compare-field-Glauber}) is proved in \Cref{section-FD-comparison-lemma}.
And the mixing lemma for the field dynamics (\Cref{lemma-field-dynamics-mixing}) is proved in  \Cref{section-anslysis-EF-dynamics} (for preparation), \Cref{section-analysis-projected-block-dynamics} (for the proof of  \Cref{lemma:block-dynamics-spectral-gap-implies-limit-chain-spectral-gap}) and \Cref{sec:block-mixing} (for the proof of \Cref{lemma:block-dynamics-spectral-gap}).
Finally, the mixing results for the two-spin systems are proved in \Cref{section-spectral-gap-2-spin}.
Additionally, In \Cref{appendix-block-mixing}, we provide a proof of \Cref{theorem:spectral-independent-imply-spectral-gap} for the mixing of uniform block dynamics assuming our notion of spectral independence with absolute influence matrix;
and more proofs for the uniqueness and spectral independence of two-spin systems with local fields are provided in \Cref{section-uniqueness-SI}.

\section{Preliminaries}\label{section-prelim}
\subsection{Notation}
Throughout the paper, we use $\log$ to denote the natural logarithm with base $\mathrm{e}$.

Let $V$ be a ground set.
For any configuration $\sigma \in \{\0,\1\}^V$, we use $\sigma^{-1}(c) \triangleq \{v \in V \mid \sigma_v = c\}$ to denote the pre-image of $c\in\{\0,\1\}$ under $\sigma$. Let $\|\sigma\|_+\triangleq|\sigma^{-1}(\1)|$ (and $\|\sigma\|_-\triangleq|\sigma^{-1}(\0)|$) denote the number of $\1$'s (and $\0$'s) in $\sigma$. For $\Lambda \subseteq V$, let $\mathbf{1}_\Lambda \in \{+1\}^\Lambda$ denote the all-$(\1)$ configuration on $\Lambda$.

For a probability distribution $\mu$, we use $\Omega(\mu)$ to denote the support of $\mu$.

Let $G=(V,E)$ be an undirected graph. For every vertex $v\in V$, we use $\Delta_v$ to denote the degree of $v$ in $G$ and denote by $\Delta_G\triangleq\max_{v\in V}\Delta_v$ the maximum degree of $G$.

\subsection{The uniqueness condition}
Let $\beta, \gamma, \lambda$ be real numbers satisfying 
\begin{align}\label{eq:anti-ferro-beta-gamma-lambda}
0 \leq \beta\leq \gamma, \gamma>0, \lambda > 0 \text{ and }\beta\gamma < 1,
\end{align}
that is, $(\beta, \gamma, \lambda)$ gives parameters for an anti-ferromagnetic two-spin system.

Given any integer $d \geq 1$, the univariate tree recursion $F_d: \mathbb{R} \to \mathbb{R}$ is defined by
\begin{align} \label{eq:tree-recursion}
  F_d(x) &\triangleq \lambda \tp{\frac{\beta x + 1}{x + \gamma}}^d,
\end{align}
and let $\hat{x}_d$ denote the unique positive fixed point of $F_d$, i.e. $\hat{x}_d = F_d(\hat{x}_d)$.

\begin{definition} [$d$-uniqueness \cite{LLY13}]\label{definition:d-unique}
Let $\delta \in (0,1)$, and $d \geq 1$ be an integer. 
  A $(\beta,\gamma,\lambda)$, where $\beta, \gamma, \lambda$ satisfy~\eqref{eq:anti-ferro-beta-gamma-lambda}, is said to be \emph{$d$-unique with gap $\delta$} if
  \begin{align}
  \label{eq-def-fd}
   f_d(\hat{x}_d) &\triangleq  \abs{F'_d(\hat{x}_d)} = \frac{d(1-\beta\gamma)\hat{x}_d}{(\beta\hat{x}_d + 1)(\hat{x}_d + \gamma)} \leq 1 - \delta.	
  \end{align}
\end{definition}

\begin{definition} [up-to-$\Delta$ uniqueness \cite{LLY13}] \label{definition:up-to-Delta-unique}
  Let $\delta \in (0,1)$ and $\Delta \in[3, +\infty]$.
  A $(\beta,\gamma,\lambda)$, where $\beta, \gamma, \lambda$ satisfy \eqref{eq:anti-ferro-beta-gamma-lambda}, is said to be \emph{up-to-$\Delta$ unique with gap $\delta$} if it is $d$-unique with gap $\delta$ for all integers $1 \leq d < \Delta$.
\end{definition}
\noindent 
Note that the definition includes that $\Delta=+\infty$, in which case the $d$-uniqueness should hold for all $d\ge1$.

Let $\+I = (V, E, \beta, \gamma, \lambda)$ be an anti-ferromagnetic two-spin system specified by parameters $(\beta, \gamma, \lambda)$ satisfying \eqref{eq:anti-ferro-beta-gamma-lambda}, on graph $G = (V, E)$ with maximum degree $\Delta=\Delta_G \geq 3$.
If $(\beta,\gamma,\lambda)$ is up-to-$\Delta$ unique with gap $\delta$, we simply say that $\+I$ is up-to-$\Delta$ unique with gap $\delta$.

\subsection{Markov chain, spectral gap and coupling}\label{section-prelim-mixing-time}
\subsubsection{Basic definitions}
Let $\Omega$ be a finite state space. Let  $(X_t)_{t \geq 0}$ be  a Markov chain over $\Omega$ with \emph{transition matrix} $P\in\mathbb{R}_{\geq 0}^{\Omega \times \Omega}$.
We use matrix $P$ to refer to the corresponding Markov chain if this is clear in the context.
The Markov chain is \emph{irreducible} if for any $X,Y \in \Omega$, there is an integer $t$ such that $P^t(X,Y) > 0$.
The Markov chain is \emph{aperiodic} if for any $X\in \Omega$, $\gcd\{t > 0 \mid P^t(X,X) > 0\} = 1$.
A distribution $\mu$ is called a stationary distribution of $P$ if $\mu = \mu P$.
If a Markov chain is both irreducible and aperiodic, then it has a unique stationary distribution.
The Markov chain $P$ is \emph{reversible} with respect to a distribution $\mu$ if the following \emph{detailed balance equation} is satisfied 
\begin{align*}
\forall X,Y \in \Omega, \quad \mu(X)P(X,Y) = \mu(Y)P(Y,X),	
\end{align*}
which implies $\mu$ is a stationary distribution of $P$.

Let $\mu$ be a distribution with support $\Omega$.
Let $P$ be a Markov chain over $\Omega$ with the unique stationary distribution $\mu$.
The \emph{mixing time} of $P$ is defined by
\begin{align*}
 \forall 0 < \epsilon < 1, \quad T_{\mathrm{mix}}(\epsilon) \triangleq \max_{X \in \Omega } \min \left\{t \mid \DTV{P^t(X,\cdot)}{\mu} \leq \epsilon\right\},
\end{align*}
where $P^t(X,\cdot)$ is the distribution generated by the Markov chain after $t$ transition steps when starting from $X$, and $\DTV{P^t(X,\cdot)}{\mu}$ denotes the \emph{total variation distance} between $P^t(X,\cdot)$ and $\mu$, formally,
\begin{align*}
\DTV{P^t(X,\cdot)}{\mu} \triangleq \frac{1}{2} \sum_{Y \in \Omega} \abs{P^t(X,Y)- \mu(Y)}.
\end{align*}
\subsubsection{Spectral gap of reversible Markov chains}
For reversible Markov chains, the mixing time is closely related to the \emph{spectral gap}.
Let $\mu$ be a distribution with support $\Omega(\mu)$.
For any $f,g \in \mathbb{R}^{\Omega(\mu)}$, define their \emph{inner product} with respect to $\mu$ by
\begin{align*}
  \inner{f}{g}_\mu \triangleq \sum_{\sigma \in \Omega(\mu)} \mu(\sigma) f(\sigma) g(\sigma).
\end{align*}
%

Let $P$ be a Markov chain over $\Omega = \Omega(\mu)$ that is reversible with respect to $\mu$.
It is well known that the transition matrix $P$ is a self-adjoint matrix with respect to inner product $\inner{\cdot}{\cdot}_\mu$, formally,
\begin{align*}
\forall f,g \in \mathbb{R}^{\Omega(\mu)}, \quad 	\inner{Pf}{g}_\mu = \inner{f}{Pg}_\mu.
\end{align*}
By standard linear algebra results~\cite[Lemma~12.2]{levin2017markov},  $P$ has $\abs{\Omega}$ real eigenvalues $1 = \lambda_1 \geq \lambda_2\geq \ldots \geq \lambda_{\abs{\Omega}} \geq -1$; each $\lambda_i$ corresponds to a real eigenvector $f_i \in \Omega^{\mathbb{R}}$ such that $Pf_i = \lambda_i f_i$, where $f_1 = \mathbf{1}$ and $f_1,f_2,\ldots,f_{\abs{\Omega}}$ form orthonormal bases of inner product space $(\mathbb{R}^\Omega, \inner{\cdot}{\cdot}_{\mu})$.
Assume $\abs{\Omega} \geq 2$.
The \emph{absolute spectral gap} of $P$ is defined by $1 - \lambda_{\star} \triangleq 1 - \max\{\abs{\lambda_i} \mid  2 \leq i \leq \abs{\Omega}\}$.
The \emph{spectral gap} $\spgap{gap}{}(P)$ of $P$ is defined by $1 - \lambda_2$. 
We simply assume $\lambda_{\star} = \lambda_2 = 0$ if $\abs{\Omega} = 1$.
%
%
The following relation between mixing time and absolute spectral gap is well-known~\cite[Theorem 12.4]{levin2017markov}:
\begin{align}
\label{eq-mixing-relaxation}
T_{\mathrm{mix}}(\epsilon) \leq 	\frac{1}{1-\lambda_\star}\log\tp{\frac{1}{\epsilon \mu_{\min}}}, \quad \text{where } \mu_{\min} = \min_{\sigma \in \Omega(\mu)}\mu(\sigma).
\end{align}

To analyze the spectral gap of a reversible Markov chain $P$, we introduce the following standard notations.
Let $f \in \mathbb{R}^{\Omega(\mu)}$. 
The $2$-norm of function $f \in \mathbb{R}^{\Omega(\mu)}$ with respect to $\mu$ is defined by
\begin{align*}
  \norm{f}_{2,\mu} \triangleq \sqrt{\inner{f}{f}_\mu},
\end{align*}
the \emph{expectation} of $f$ with respect to $\mu$ is defined by
\begin{align}
\label{eq-def-Expectation}
  \E[\mu]{f} \triangleq \sum_{\sigma \in \Omega} \mu(\sigma) f(\sigma),
\end{align}
the \emph{variance} with respect to $\mu$ is defined by
\begin{align}
\label{eq-def-Variance}
\Var[\mu]{f} \triangleq \E[\mu]{f^2} - (\E[\mu]{f})^2 = \frac{1}{2} 	\sum_{\sigma, \tau \in \Omega} \mu(\sigma) \mu(\tau) \tp{f(\sigma) - f(\tau)}^2,
\end{align}
and the \textit{Dirichlet form} with respect to $P$ (where $P$ is reversible with respect to $\mu$) is defined by
\begin{align}
\label{eq-def-dirichlet}
\+E_{P}(f,f) \triangleq \inner{f}{(I - P)f}_\mu = \frac{1}{2} \sum_{\sigma, \tau \in \Omega} \mu(\sigma) P(\sigma, \tau) \tp{f(\sigma) - f(\tau)}^2.
\end{align}
We can slightly extend above definitions to allow functions $f \in \mathbb{R}^S$ for some $S \supseteq \Omega$, and for these functions $f$, the expectation in~\eqref{eq-def-Expectation}, the variance in~\eqref{eq-def-Variance} and the Dirichlet form in~\eqref{eq-def-dirichlet} are all well-defined. 

For any $f,g \in \mathbb{R}^\Omega$, we write $f \perp_{\mu} g$ if $\inner{f}{g}_{\mu} = 0$.
By Courant-Fischer theorem~\cite[Lemma 13.7]{levin2017markov}, 
\begin{align*}
\spgap{gap}{}(P) = \inf \left\{ {\+E_P(f,f)} \mid f \in \mathbb{R}^{\Omega(\mu)}, f \perp_\mu \*1, \norm{f}_{2,\mu} = 1 \right\}.	
\end{align*}
Furthermore,  the spectral gap can also be characterized as~\cite[Remark 13.8]{levin2017markov}:
\begin{align}
\label{eq-comp-gap-def}
\spgap{gap}{}(P) = \inf \left\{ \frac{\+E_P(f,f)}{\Var[\mu]{f}} \mid f \in \mathbb{R}^{\Omega(\mu)}, \Var[\mu]{f} \neq 0 \right\},	
\end{align}
The Poincar\'e inequality follows from the above equation:
\begin{align}
\label{eq:Poincare-inequality}
\forall f\in\mathbb{R}^{\Omega(\mu)},\quad 
\+E_{P}(f,f) 
\ge 
\spgap{gap}{}(P) \cdot \Var[\mu]{f}.
\end{align}

\subsubsection{Coupling of Markov chains}
Let $\mu$ and $\nu$ be two distributions on $\Omega$. A coupling of $\mu$ and $\nu$ is a joint distribution $(X,Y)$ over $\Omega \times \Omega$ such that the marginal distributions of $X$ and $Y$ are $\mu$ and $\nu$ respectively. The following is the  well-known coupling lemma.
\begin{lemma}[\text{\cite[Proposition 4.7]{levin2017markov}}]
For any coupling $(X,Y)$ of $\mu$ and $\nu$,
  \begin{align*}
    \Pr[]{X \neq Y} \ge \DTV{\mu}{\nu} = \frac{1}{2}\sum_{\sigma \in \Omega}\abs{\mu(\sigma)-\nu(\sigma)}.
  \end{align*}
Furthermore, there is an optimal coupling $(X,Y)$ such that $\Pr[]{X \neq Y} = \DTV{\mu}{\nu}$.
\end{lemma}

Let $P$ denote a Markov chain over the state space $\Omega$.
A coupling of Markov chain is a joint stochastic process $(\*X_t,\*Y_t)_{t \geq 0}$ such that each individual process $(\*X_t)_{t \geq 0}$ and $(\*Y_t)_{t \geq 0}$ follow the transition rule of $P$, and if $\*X_t = \*Y_t$, then $\*X_k = \*Y_k$ for all $k \geq t$.
The following lemma connects  the spectral gap of Markov chain and the contraction rate in coupling.
\begin{lemma}[\text{\cite{MFC98} and \cite[ch.~13.1]{levin2017markov}}]\label{lemma:MFC}
  Let $\mu$ be a distribution with support $\Omega$, and $\Phi$  a metric on $\Omega$, where $\abs{\Omega} \ge 2$. 
  Let $P \in \mathbb{R}^{\Omega \times \Omega}_{\ge 0}$ be the transition matrix of a Markov chain that is reversible with respect to $\mu$.
  If there exists  $ 0 < r  < 1$ such that for any $X,Y \in \Omega$, there exists a coupling $(X,Y) \to (X', Y')$ of $P$ such that
  \[
    \E{\Phi(X',Y')\mid X, Y} \le (1-r) \Phi(X,Y),
  \]
  then the spectral gap of $P$ satisfies 
  \begin{align*}
  	1 - \lambda_2(P) \geq r,
  \end{align*}
where $\lambda_2(P)$ is the second largest eigenvalue of $P$.
\end{lemma}

\subsection{Uniform block dynamics and Glauber dynamics}
Let $\mu$ be a distribution over $\{\0,\1\}^V$, where $V$ is a ground set.
For any integer $1\leq \ell \leq \abs{V}$, let 
$
\binom{V}{\ell} \triangleq \{S \subseteq V \mid \abs{S} = \ell \}	
$
denote the collection of all subsets of size $\ell$.
The (heat-bath) uniform block dynamics for $\mu$ is defined as follows. 
\begin{definition}[Heat-bath uniform $\ell$-block dynamics]
\label{definition-block-dynamics}
For any positive integer $1\leq \ell \leq \abs{V}$, the (heat-bath) uniform $\ell$-block dynamics for $\mu$ is a Markov chain $(X_t)_{t \geq 0}$ over $\Omega(\mu)$.
The chain starts from an arbitrary configuration $X_0 \in \Omega$. In the $t$-th transition step, the chain evolves as follows:
\begin{itemize}
\item pick a set $S \in \binom{V}{\ell}$ uniformly at random, and set $X_t(V \setminus S ) = X_{t-1}(V \setminus S)$;
\item sample $X_t \sim \mu^{X_{t}(V \setminus S)}_S$, where $\mu^{X_{t}(V \setminus S)}_S$ denotes the marginal distribution on $S$ induced from $\mu$ conditional on the assignment of $X_{t}(V \setminus S)$.
\end{itemize}
In particular, the uniform $1$-block-dynamics is known as the \textit{Glauber dynamics} for $\mu$.
\end{definition}

The following proposition was known.
\begin{proposition}[\text{\cite{DGU14,levin2017markov, alev2020improved}}]
Let $\mu$ be a distribution over $\{\0,\1\}^V$. For any integer $1\leq \ell \leq |V|$,  the uniform $\ell$-block dynamics for $\mu$ is reversible with respect to $\mu$, and its transition matrix is positive semidefinite.
\end{proposition}

Let $\mu$ be a distribution over $\{\0,\1\}^V$ and $1\leq \ell \leq |V|$ an integer.
Let $P$ denote the transition matrix of the uniform $\ell$-block dynamics for $\mu$.
\begin{align*}
T_{\mathrm{mix}}(\epsilon) \leq 	\frac{1}{1-\lambda_2(P)}\log\tp{\frac{1}{\epsilon \mu_{\min}}}, \quad \text{where } \mu_{\min} = \min_{\sigma \in \Omega(\mu)}\mu(\sigma),
\end{align*}
where $\lambda_2(P)$ is the second largest eigenvalue of $P$.

For Glauber dynamics (uniform 1-block dynamics), its spectral gap can be further characterized by the \emph{approximate tensorization of variance}. 
For any variable $v \in V$, and any function $f \in \mathbb{R}^{\Omega(\mu)}$, define
\begin{align*}
\mu[\Var[v]{f}] \triangleq \sum_{\sigma \in \Omega(\mu_{V \setminus \{v\}})}\mu_{V \setminus \{v\}}(\sigma)\Var[\mu^\sigma]{f},
\end{align*}
where $\mu^\sigma$ is the distribution over $\{\0,\1\}^V$ obtained from $\mu$ conditional on the configuration on $V \setminus \{v\}$ being fixed as $\sigma$. 
We remark that $\mu[\Var[v]{f}]$ is also well-defined if $f \in \mathbb{R}^{S}$ for some $S \supseteq \Omega(\mu)$.
The approximate tensorization of variance is defined as follows.
\begin{definition}
  \label{definition-app-ten-var}
  Let $C > 0$ be a parameter. A distribution $\mu$ over $\{\0,\1\}^V$ with $\abs{\Omega(\mu)} \geq 2$ is said to satisfy the \emph{approximate tensorization of variance} with parameter $C$ if for all $f \in \mathbb{R}^{\Omega(\mu)}$,
  \begin{align*}
  \Var[\mu]{f} \leq C \sum_{v \in V} \mu[\Var[v]{f}].	
  \end{align*}
  \end{definition}
  
  \begin{lemma}[Fact A.3 \cite{chen2020optimal}]
  \label{lemma-variance-ten}
  A distribution $\mu$ over $\{\0,\1\}^V$ with $\abs{\Omega(\mu)}\geq 2$ satisfies the approximate tensorization of variance with parameter $C$ if and only if the spectral gap $\spgap{gap}{}(P) \geq \frac{1}{Cn}$, where $n = |V|$, and $P$ denotes the Glauber dynamics for $\mu$.
  \end{lemma}


\subsection{Multivariate hypergeometric distribution}
Let $V$ be a set of $n$ buckets, each of them has $k$ balls.
Suppose we pick $\ell$ balls from all $kn$ balls uniformly at random, without replacement.
For each bucket $v \in V$, let $a_v \in \mathbb{Z}_{\geq 0}$ denote the number of balls picked from the bucket $v$,
then $\*a = (a_v)_{v \in V}$ follows multivariate hypergeometric distribution.

Formally, given a set $V$ of size $n$, an integer $k \geq 1$ and an integer $0\leq \ell \leq kn$, the multivariate hypergeometric distribution $\HyperGeo$ is defined as follows.
The support of  $\HyperGeo$ is defined by
\begin{align}
\label{eq-support-hypergeo}
  \Omega(\HyperGeo) \triangleq \left\{\boldsymbol{a} = (a_v)_{v \in V} \mid  \sum_{v \in V} a_v = \ell \mbox{ and } \forall v \in V , a_v \in \mathbb{Z}_{\geq 0} \right\},
\end{align}
For any $\boldsymbol{a} \in \Omega(\HyperGeo)$, it holds that
\begin{align}
\label{eq-hypergeo}
  \HyperGeo (\boldsymbol{a}) = \frac{\prod_{v \in V} \binom{k}{a_v}}{\binom{kn}{\ell}}.
\end{align}

By the negative association property~\cite{joag-dev1983} of hypergeometric distribution, we have the following Chernoff-Hoeffding inequality.
\begin{lemma}[\cite{joag-dev1983} and \cite{dubhashi1998balls}]\label{lemma:hypergometric-concentration}
Let $\*a \sim  \HyperGeo$. 
For any $v \in V$ and $\epsilon \in (0,1)$, it holds that
\begin{align*}
  \Pr{\left\vert \frac{a_v}{k} - \frac{\ell}{kn} \right\vert \geq \epsilon} &\leq 2\exp\tp{-2\epsilon^2 k}.
\end{align*}
\end{lemma}
\begin{proof}
For each bucket $v \in V$, we use $(v,1),(v,2),\ldots,(v,k)$ to denote all balls in bucket $v$.
For each ball $(v,i)$, we use random variable $X_{(v,i)} \in \{0,1\}$ to indicate whether the ball $(v,i)$ is picked.
It holds that 
$$a_v = \sum_{i \in [k]}X_{(v,i)}.$$
Since $(X_{(v,i)})_{i \in [k]}$ are negative associated~\cite[Lemma 2.11]{joag-dev1983}, the Chernoff-Hoeffding inequality \cite{dubhashi1998balls} can be applied to $a_v$.
\end{proof}

\section{Comparing Glauber Dynamics with Field Dynamics}\label{section-FD-comparison-lemma}
In this section, we prove the comparison lemma for the field dynamics (\Cref{lemma-compare-field-Glauber}).

Let $\mu$ be a distribution over $\{\0,\1\}^V$ and $\theta\in(0,1)$. 
Let $\pi=\Mag{\mu}{\theta}$, which is defined as in~\Cref{defintion-magnetizing}.
Note that $\mu$ and $\pi$ have the same support $\Omega\triangleq {\Omega(\pi)}={\Omega(\mu)}$ for $\theta >0$.
Formally, $\pi$ is  defined as:
\begin{align}
\label{eq-redef-pi}
\forall\sigma\in\{\0,\1\}^V,\quad 
\pi(\sigma) = \frac{\mu(\sigma)\theta^{\|\sigma\|_{+}} }{Z_\pi},
\quad\text{ where }Z_\pi \triangleq \sum_{\sigma \in \{\0,\1\}^V}\mu(\sigma) \theta^{\|\sigma\|_{+}}.
\end{align}

Moreover, we use $\OPr[R\subseteq V]{\cdot}$ to denote the law for subset $R\subseteq V$ that is randomly generated by including each $v\in V$ into $R$ independently with probability $1-\theta$. Specifically, for every $\Lambda\subseteq V$,
\begin{align}
\OPr[R\subseteq V]{R=\Lambda}=(1 - \theta)^{\abs{\Lambda}} \theta^{\abs{V} - \abs{\Lambda}}.\label{eq:definition-binomial-R}
\end{align}
%
%

Let $\Plim$ denote the field dynamics for distribution $\mu$ with parameter $\theta$.
The Dirichlet form of $\Plim$ can be calculated as follows.

\begin{lemma}\label{lemma:dirichlet-decomposition}
Let $\pi=\Mag{\mu}{\theta}$ be defined as in~\eqref{eq-redef-pi}. For all $f \in \mathbb{R}^{\Omega}$,
\begin{align}
\label{eq-dicr-EF}
\+E_{\Plim}(f,f) 
  &=   \frac{Z_\pi}{\theta^{\abs{V}}}\OEp[R \subseteq V]{\pi_R(\*1_R)\cdot\Var[\pi^{\*1_R}]{f}}.
\end{align}
where  $\*1_R\in\{\0,\1\}^R$ denotes the all-$(\1)$ configuration specified on $R\subseteq V$, 
and the expectation is calculated assuming the convention that $\pi_R(\*1_R)\Var[\pi^{\*1_R}]{f}=0$ when $\pi_R(\*1_R)=0$.\footnote{Equivalently, one may think this as an expectation $\OEp[R\subseteq V]{\mathcal{P}(R)}$ of a function $\mathcal{P}(R)$, where $\mathcal{P}(R)=0$ if $\pi_R(\*1_R)=0$ and $\mathcal{P}(R)=\pi_R(\*1_R)\Var[\pi^{\*1_R}]{f}$ if $\pi_R(\*1_R)>0$, in which case the variance $\Var[\pi^{\*1_R}]{f}$ is well-defined.}
\end{lemma}

Intuitively, the  $R\subseteq V$ corresponds  to $R=V\setminus S$ where $S$ is the random set of resampled variables generated in the transition step of the field dynamics, such that every $v\in V$ with current value $\0$ is selected into $S$, and thus the configuration over $R=V\setminus S$ must be $\*1_R$.
And the quantity $\Var[\pi^{\mathbf{1}_R}]{f}$ arises because the current configuration over $S$ is resampled according to $\pi=\Mag{\mu}{\theta}$ conditional on the current configuration over $R=V\setminus S$, which is $\*1_R$.
The proof of \Cref{lemma:dirichlet-decomposition} is postponed to the end of this section.

Consider the Glauber dynamics for each $\pi^{\*1_R}$. 
%
%
By \Cref{lemma-variance-ten}, for all $R\subseteq V$ with $\pi_R(\*1_R)>0$,  the distribution $\pi^{\*1_R}$ satisfies the approximate tensorization of variance with parameter $\frac{1}{n\cdot\spgap{gap}{\GD}(\pi^{\*1_R})}$:
\[
\forall f\in\mathbb{R}^{\Omega},\quad
\Var[\pi^{\mathbf{1}_R}]{f}
\le
\frac{1}{n\cdot\spgap{gap}{\GD}(\pi^{\*1_R})}\sum_{v \in V} \pi^{\*1_R}[\Var[v]{f}]
\le
\frac{1}{n\cdot\spgap{\sgap}{\GD}(\pi)}\sum_{v \in V} \pi^{\*1_R}[\Var[v]{f}],
\]
where $\spgap{\sgap}{\GD}(\pi)\le \spgap{gap}{\GD}(\pi^{\sigma})$ for all feasible partial configuration $\sigma$, as defined in~\eqref{eq-def-worst-gap}.

Then by~\eqref{eq-dicr-EF}, the Dirichlet form $\+E_{\Plim}(f,f)$ is upper bounded as 
\begin{align}
\label{eq:FD-Dirichlet-upper-bound}
\+E_{\Plim}(f,f) \le 
\frac{Z_\pi}{n\cdot\spgap{\sgap}{\GD}(\pi)\cdot\theta^{\abs{V}}}
\sum_{v \in V}
\OEp[R \subseteq V]{\pi_R(\*1_R)\cdot\pi^{\*1_R}[\Var[v]{f}]}.
\end{align}

Let $I[\cdot]$ denote the indicator variable.
The following identity holds for all $f\in\mathbb{R}^{\Omega}$:
\begin{align}
\OEp[R \subseteq V]{\pi_R(\*1_R)\cdot\pi^{\*1_R}[\Var[v]{f}]}
&\overset{(\star)}{=}
\OEp[R \subseteq V]{\OEp[\sigma\sim\pi_{V\setminus\{v\}}]{I\left[R\subseteq\sigma^{-1}(\1)\right]\cdot\Var[\pi^{\sigma}]{f}}}\notag\\
&=
\OEp[\sigma\sim\pi_{V\setminus\{v\}}]{\OPr[R \subseteq V]{R\subseteq\sigma^{-1}(\1)}\cdot\Var[\pi^{\sigma}]{f}}\notag\\
&=
\OEp[\sigma\sim\pi_{V\setminus\{v\}}]{\theta^{|V|-\|\sigma\|_{+}}\cdot\Var[\pi^{\sigma}]{f}},\label{eq:variance-expectation-R-identity}
\end{align}
where the nontrivial equation $(\star)$ holds by verifying for every choice of $v\in V$ and $R\subseteq V$ as follows: 
\begin{itemize}
\item For the case that $\pi_R(\*1_R)>0$ and $v\not\in R$, it holds that
\[
\pi_R(\*1_R)=\OPr[\sigma\sim\pi]{R\subseteq\sigma^{-1}(\1)}=\OPr[\sigma\sim\pi_{V\setminus\{v\}}]{R\subseteq\sigma^{-1}(\1)},
\]
and the variance $\pi^{\*1_R}[\Var[v]{f}]$ is well-defined, such that
\[
\pi^{\*1_R}[\Var[v]{f}]=\OEp[\sigma\sim\pi_{V\setminus\{v\}}^{\*1_R}]{\Var[\pi^{\sigma}]{f}}=\OEp[\sigma\sim\pi_{V\setminus\{v\}}]{\Var[\pi^{\sigma}]{f}\mid R\subseteq\sigma^{-1}(\1)}.
\]
Therefore, 
\begin{align*}
\pi_R(\*1_R)\cdot\pi^{\*1_R}[\Var[v]{f}]
&=
\OPr[\sigma\sim\pi_{V\setminus\{v\}}]{R\subseteq\sigma^{-1}(\1)}
\cdot
\OEp[\sigma\sim\pi_{V\setminus\{v\}}]{\Var[\pi^{\sigma}]{f}\mid R\subseteq\sigma^{-1}(\1)}\\
&=
\OEp[\sigma\sim\pi_{V\setminus\{v\}}]{I\left[R\subseteq\sigma^{-1}(\1)\right]\cdot \Var[\pi^{\sigma}]{f}}.
\end{align*}
\item For the case that $\pi_R(\*1_R)=0$ or $v\in R$, both sides are 0. 
On the left-hand-side, if $\pi_R(\*1_R)=0$, then by convention 
\[
\pi_R(\*1_R)\cdot\Var[\pi^{\*1_R}]{f}=0;
\] 
or else, if $\pi_R(\*1_R)>0$ but $v\in R$, then the variance $\pi^{\*1_R}[\Var[v]{f}]$ is well-defined, but for $\sigma\sim\pi_{V\setminus\{v\}}^{\*1_R}$, 
the $\*1_R\uplus\sigma$ gives a configuration fully specified on $V$ and hence the variance becomes trivial, i.e. 
\[
\pi^{\*1_R}[\Var[v]{f}]=\OEp[\sigma\sim\pi_{V\setminus\{v\}}^{\*1_R}]{\Var[\pi^{\*1_R\uplus\sigma}]{f}}=0.
\]
On the right-hand-side, if $\pi_R(\*1_R)=0$ or $v\in R$, then for $\sigma\sim\pi_{V\setminus\{v\}}$, the event $R\subseteq\sigma^{-1}(\1)$ can never occur, and hence
\[
\OEp[\sigma\sim\pi_{V\setminus\{v\}}]{I\left[R\subseteq\sigma^{-1}(\1)\right]\cdot \Var[\pi^{\sigma}]{f}}=0.
\]
\end{itemize}
This gives the equation $(\star)$ in~\eqref{eq:variance-expectation-R-identity}.
Meanwhile, the other two equations in~\eqref{eq:variance-expectation-R-identity} follows respectively from  linearity of expectation and the fact that $\Pr[R \subseteq V]{R\subseteq \Lambda}=\theta^{|V|-|\Lambda|}$ for all $\Lambda\subseteq V$.

Furthermore, it can be verified that
\begin{align}
  \OEp[\sigma \sim \pi_{V \setminus\{v\}}]{\frac{1}{\theta^{\norm{\sigma}_+}}\Var[\pi^\sigma]{f}}
  &= \sum_{\sigma \in \Omega(\pi_{V\setminus \{v\}})} \frac{1}{\theta^{\norm{\sigma}_+}}\pi_{V\setminus\{v\}}(\sigma)  \pi^\sigma_v(-1)  \pi^\sigma_v(+1)  (f(\sigma_+) - f(\sigma_{-}))^2 \notag\\
  &= \frac{1}{Z_\pi} \sum_{\sigma \in \Omega(\pi_{V\setminus \{v\}})} \mu_{V\setminus\{v\}}(\sigma) \mu^\sigma_v(-1)  \pi^\sigma_v(+1)   (f(\sigma_+) - f(\sigma_{-}))^2 \notag\\
  &\leq \frac{1}{Z_\pi} \sum_{\sigma \in \Omega(\pi_{V\setminus \{v\}})} \mu_{V\setminus\{v\}}(\sigma) \mu^\sigma_v(-1)  \mu^\sigma_v(+1)  (f(\sigma_+) - f(\sigma_{-}))^2 \notag\\
  &= \frac{1}{Z_\pi} \mu[\Var[v]{f}], \label{eq:tensorization-change-base}
\end{align}
where  $\sigma_{\pm}\in\{\0,\1\}^V$ denote the configurations on $V$ where $\sigma_{\pm}(V\setminus\{v\})=\sigma$ and $\sigma_{\pm}(v)=\pm1$,
 the second equation is due to the chain rule and the relation between $\pi$ and $\mu$ in~\eqref{eq-redef-pi},
and the inequality is due to the relaxation 
\begin{align}
\pi_{v}^{\sigma}(\1)\le\mu_v^{\sigma}(\1),\label{eq:margin-monotone-pi-mu}
\end{align}
which holds because $\pi=\Mag{\mu}{\theta}$ is obtained by biasing every variable with a local field $\theta\in(0,1)$.
Indeed, 
\begin{align*}
\pi^\sigma_v(\1) = \frac{\pi(\sigma_{+})}{\pi(\sigma_{-}) + \pi(\sigma_{+}) }	= \frac{\theta \mu(\sigma_{+})}{\mu(\sigma_{-}) + \theta\mu(\sigma_{+}) } \leq \frac{\mu(\sigma_{+})}{\mu(\sigma_{-}) + \mu(\sigma_{+}) } = \mu^\sigma_v(\1),
\end{align*}
where the inequality holds for $ \theta \in(0, 1)$.

Combining~\eqref{eq:FD-Dirichlet-upper-bound}, \eqref{eq:variance-expectation-R-identity}, and~\eqref{eq:tensorization-change-base}, we have the following upper bound on the Dirichlet form:
\[
\+E_{\Plim}(f,f) \le 
\frac{1}{n\cdot\spgap{\sgap}{\GD}(\Mag{\mu}{\theta})}
\sum_{v \in V}
\mu[\Var[v]{f}].
\]
Due to the Poincar\'e's inequality~\eqref{eq:Poincare-inequality} for the field dynamics $\Plim$, for all $f\in\mathbb{R}^{\Omega}$,
\[
\Var[\mu]{f}
\le \frac{1}{\spgap{gap}{\FD}\left(\mu,\theta\right)}\+E_{\Plim}(f,f) 
\le \frac{1}{n\cdot \spgap{gap}{\FD}\left(\mu,\theta\right)\cdot\spgap{\sgap}{\GD}\left(\Mag{\mu}{\theta}\right)}
\sum_{v \in V}
\mu[\Var[v]{f}].
\] 
This shows that $\mu$ satisfies the approximate tensorization of variance with parameter 
\[
\frac{1}{n\cdot \spgap{gap}{\FD}\left(\mu,\theta\right)\cdot\spgap{\sgap}{\GD}\left(\Mag{\mu}{\theta}\right)}.
\]
By \Cref{lemma-variance-ten}, it gives us the following lower bound on the spectral gap of the Glauber dynamics for $\mu$
\[
\spgap{gap}{\GD}\left(\mu\right)\ge \spgap{gap}{\FD}\left(\mu,\theta\right)\cdot\spgap{\sgap}{\GD}\left(\Mag{\mu}{\theta}\right),
\]
which proves the comparison lemma for the field dynamics (\Cref{lemma-compare-field-Glauber}).

\begin{remark}[\textbf{Tightness of \Cref{lemma-compare-field-Glauber}}]
Apart from the Poincar\'e's inequalities, 
our proof of \Cref{lemma-compare-field-Glauber} is tight almost everywhere.
The only exception is the relaxation of the marginal probability in~\eqref{eq:margin-monotone-pi-mu}.
\end{remark}

\begin{proof}[Proof of \Cref{lemma:dirichlet-decomposition}]
In each transition step of the field dynamics, a random subset $S \subseteq V$ of variables is generated for resampling.
Let $R = V \setminus S$. 
The Dirichlet form can be calculated as
\begin{align*}
 \+E_{\Plim}(f,f) &= \frac{1}{2} \sum_{\sigma,\tau \in \Omega(\mu)} \mu(\sigma) \Plim(\sigma,\tau) (f(\sigma)-f(\tau))^2\\
&= \frac{1}{2} \sum_{\sigma,\tau \in \Omega(\mu)} \sum_{R \subseteq \sigma^{-1}(\1)} (1-\theta)^{|R|}\theta^{\|\sigma\|_{+}-|R|} \pi^{\mathbf{1}_R}(\tau) \mu(\sigma) (f(\sigma)-f(\tau))^2,
\end{align*}
where the second equation is due to the following observations: 
\begin{itemize}
\item
$R$ can only be a subset of $\sigma^{-1}(\1)$ since all variables $v\in V$ with $\sigma_v=\0$ are selected into $S=V\setminus R$; 
\item
each $v\in\sigma^{-1}(\1)$ is independently selected into $R$ with probability $1-\theta$;
\item
given an $R\subseteq \sigma^{-1}(\1)$, the current configuration $\sigma$ transits to $\tau$ with probability $\pi^{\sigma_R}(\tau)=\pi^{\*1_R}(\tau)$.
\end{itemize}
Hence, fix any $R\subseteq V$, we only need to consider the pair $\sigma,\tau \in \Omega(\mu)$ such that $\sigma_R = \tau_R = \mathbf{1}_R$. We have 
\begin{align*}
 \+E_{\Plim}(f,f) = \frac{1}{2} \sum_{R \subseteq V: \mathbf{1}_R \in \Omega(\mu_R)}\tp{\frac{1-\theta}{\theta}}^{|R|} \sum_{\sigma, \tau \in \Omega(\mu^{\mathbf{1}_R})}\theta^{\|\sigma\|_{+}} \pi^{\mathbf{1}_R}(\tau) \mu(\sigma) (f(\sigma)-f(\tau))^2.
\end{align*}
%
By~\eqref{eq-redef-pi}, it holds that $\pi(\sigma)Z_{\pi} = \mu(\sigma)\theta^{\|\sigma\|_{+}}$. We have
\begin{align*}
 \+E_{\Plim}(f,f) &= \frac{Z_\pi}{2} \sum_{R \subseteq V: \mathbf{1}_R \in \Omega(\pi_R)}\tp{\frac{1-\theta}{\theta}}^{|R|} \sum_{\sigma, \tau \in \Omega(\pi^{\mathbf{1}_R})}\pi^{\mathbf{1}_R}(\tau) \pi(\sigma) (f(\sigma)-f(\tau))^2\\
&=	\frac{Z_\pi}{2} \sum_{R \subseteq V: \mathbf{1}_R \in \Omega(\pi_R)}\tp{\frac{1-\theta}{\theta}}^{|R|} \pi_R(\mathbf{1}_R) \sum_{\sigma, \tau \in \Omega(\pi^{\mathbf{1}_R})}\pi^{\mathbf{1}_R}(\tau) \pi^{\mathbf{1}_R}(\sigma) (f(\sigma)-f(\tau))^2\\
 & = \frac{Z_\pi}{\theta^{\abs{V}}} \sum_{R \subseteq V: \mathbf{1}_R \in \Omega(\pi_R)}\tp{1-\theta}^{|R|}\theta^{|V|-|R|} \pi_R(\mathbf{1}_R) \Var[\pi^{\mathbf{1}_R}]{f}\\
  &=   \frac{Z_\pi}{\theta^{\abs{V}}}\OEp[R \subseteq V]{\pi_R(\*1_R)\cdot\Var[\pi^{\*1_R}]{f}},
\end{align*}
where the second equation is due to the fact that $\sigma_R = \mathbf{1}_R$ for $\sigma \in \Omega(\pi^{\mathbf{1}_R})$ and the last equation adopts the convention that $\pi_R(\mathbf{1}_R) \Var[\pi^{\mathbf{1}_R}]{f}=0$ when $\pi_R(\mathbf{1}_R)=0$.
\end{proof}

\section{Mixing of Field Dynamics}
\label{section-anslysis-EF-dynamics}

In this section, and the next two sections, 
prove the mixing lemma for the field dynamics (\Cref{lemma-field-dynamics-mixing}).
%
%
In \Cref{section-EF-stationary}, we prove \Cref{theorem-field-dynamics-basic} for the convergence and reversibility of the field dynamics.
In \Cref{section-ex-field-block-dynamics}, we start our proof of \Cref{lemma:block-dynamics-spectral-gap-implies-limit-chain-spectral-gap}, which states that the field dynamics is the limiting instance for the uniform block dynamics on $k$-transformed distributions.
 
\subsection{Convergence and reversibility (proof of \Cref{theorem-field-dynamics-basic})}
\label{section-EF-stationary}
Fix a joint distribution $\mu$ over $\{\0,\1\}^V$.
%
%
For any $0 < \theta < 1$, let $\Plim$ denote the transition matrix of the field dynamics for $\mu$
with parameter $\theta$.

First, we verify that $\Plim$ is irreducible and aperiodic. 
Let $\pi = \mu^{(\theta)}$, which is as defined in \Cref{defintion-magnetizing}.
Note that $\pi$ and $\mu$ have the same support since $\theta > 0$, i.e.
\[
\Omega \triangleq \Omega(\pi) = \Omega(\mu).
\]
In each transition step, the chain constructs a random subset $S \subseteq V$ of variables that are going to be resampled.
It holds that $S=V$ with positive probability for $\theta>0$, in which case the current configuration $\sigma$ is entirely resampled according to $\pi=\Mag{\mu}{\theta}$.
This means $\Plim(\sigma,\tau) > 0$ for all pairs $\sigma,\tau \in \Omega$ of feasible configurations.
Hence, the Markov chain is  irreducible and aperiodic.

Next, we prove the reversibility of the chain with respect to $\mu$ by verifying the detailed balance equation:
\begin{align}
\label{eq:field-dynamics-DBE}
\forall \sigma,\tau \in \Omega, \quad  \mu(\sigma)\Plim(\sigma,\tau) = \mu(\tau)\Plim(\tau,\sigma).
\end{align}


Suppose that the current configuration is $\sigma \in \{\0,\1\}^V$. 
By the transition rule of the field dynamics, a random $S\subseteq V$ is constructed so that each $v\in V$ with $\sigma_v = \1$ is added to $S$ with probability $\theta$ and each $v\in V$ with  $\sigma_v = \0$ is added to $S$ with probability $1$.
Let $R = V \setminus S$ denote the complement of the resampled set. 
It must hold that
$R \subseteq \sigma^{-1}(\1)$ since every variable $v\in V$ with $\sigma_v=\0$ is selected into $S$. 
To successfully transform from $\sigma$ to $\tau$, in addition, it should be satisfied that $R \subseteq \tau^{-1}(\1)$ and the resampling step generates the configuration $\tau_S$ on $S$.
Therefore, for any $\sigma,\tau \in \Omega$,
\begin{align}
\mu(\sigma)\Plim(\sigma,\tau)	
&= \mu(\sigma) \sum_{R \subseteq \sigma^{-1}(\1) \cap \tau^{-1}(\1)} \tp{1-\theta}^{|R|} \theta^{\|\sigma\|_+-|R|} \pi^{\mathbf{1}_R}_S(\tau_S)\label{eq-2.8-0}\\
&= \mu(\sigma) \cdot \theta^{\|\sigma\|_+} \sum_{R \subseteq \sigma^{-1}(\1) \cap \tau^{-1}(\1)} \tp{\frac{1-\theta}{\theta}}^{|R|}  \pi^{\mathbf{1}_R}(\tau)
&&(\text{since } \tau_R = \mathbf{1}_R)\notag\\
&=\mu(\sigma) \cdot \theta^{\|\sigma\|_+} \sum_{R \subseteq \sigma^{-1}(\1) \cap \tau^{-1}(\1)} \tp{\frac{1-\theta}{\theta}}^{|R|}  \frac{\pi(\tau)}{\pi_R(\mathbf{1}_R)}.
&&(\text{since } \tau_R = \mathbf{1}_R)\label{eq-2.8-1}
\end{align}
Note $\sigma \in \Omega$ is feasible with respect to both $\mu$ and $\pi$.
And $\pi_R(\mathbf{1}_R)>0$ since $\sigma_R = \mathbf{1}_R$. Thus, the conditional probability $\pi^{\mathbf{1}_R}$ in \eqref{eq-2.8-0} and the ratio in~\eqref{eq-2.8-1} are well-defined.
Recall that $\pi = \Mag{\mu}{\theta}$. By \Cref{defintion-magnetizing}, 
\begin{align*}
\forall \tau \in \{\0,\1\}^V, 	\quad \pi(\tau) = \frac{\mu(\tau)\cdot\theta^{\|\tau\|_+}}{Z(\mu,\theta)},
\end{align*}
where $Z(\mu,\theta) \triangleq \sum_{\tau \in \{\0,\1\}^V}\mu(\tau)\theta^{\|\tau\|_+}$ denotes the normalizing factor that depends only on $\mu$ and $\theta$. We have
\begin{align*}
\mu(\sigma)\Plim(\sigma,\tau) = \tp{\mu(\sigma)\mu(\tau)\cdot\theta^{\|\sigma\|_+ +\|\tau\|_+} }\sum_{R \subseteq \sigma^{-1}(\1) \cap \tau^{-1}(\1)} \tp{\frac{1-\theta}{\theta}}^{|R|}  \frac{1}{Z(\mu,\theta) \cdot \pi_R(\mathbf{1}_R)},
\end{align*}
which is symmetric in $\sigma$ and $\tau$. Therefore, the detailed balanced equation~\eqref{eq:field-dynamics-DBE} holds. This concludes the proof of \Cref{theorem-field-dynamics-basic}.

\subsection{Field dynamics as the $k$-transformed block dynamics (proof of \Cref{lemma:block-dynamics-spectral-gap-implies-limit-chain-spectral-gap})}
\label{section-ex-field-block-dynamics}
Let $\mu$ be a distribution over $\{\0,\1\}^V$ and $k \in \mathbb{N}^+$. 
Let $\mu_k = \Rd(\mu,k)$ denote the $k$-transformation of $\mu$ defined in \Cref{definition-k-transformation}.

We use $V_k$ to denote the ground set for $\mu_k$, that is,
\begin{align}
V_k \triangleq V \times [k].\label{eq:definition-V-k}
\end{align}
For each $v \in V$ and $i \in [k]$, for convenience, we denote 
\begin{align*}
v_i \triangleq (v,i) \in V_k.	
\end{align*}
Finally, for each $v \in V$, we denote
\begin{align}
C_v \triangleq \{v_i \mid i \in [k]\}.	\label{eq:definition-C-v}
\end{align}

Recall that 
$P_{k,\ell}$ denotes the uniform $\ell$-block dynamics for $\mu_k=\Rd(\mu,k)$.
To relate the field dynamics $\Plim$ with  $P_{k,\ell}$, we consider the following natural projection operation which maps a configuration in the $k$-transformed distribution $\mu_k$ back into $\{\0,\1\}^V$.

\begin{definition} [projection of configuration] 
Let $\mu$ be a distribution over $\{\0,\1\}^V$ and $k\geq 1$ an integer.
Let $\mu_k = \Rd(\mu,k)$ denote the $k$-transformation of $\mu$.
For any $\sigma \in \Omega(\mu_k)$, the projection $\p:\sigma$ of $\sigma$ is a configuration in $\{\0,\1\}^V$ such that for any $v \in V$
\begin{align}
\label{eq-def-sigma-star}
\p:\sigma_v \triangleq \begin{cases}
 \1 &\text{if } \exists 1 \leq i \leq k \text{ s.t. } \sigma_{v_i} = \1\\
 \0 &\text{otherwise}.	
 \end{cases}
 \end{align}
\end{definition}

The projection operation defined as above naturally transforms the uniform $\ell$-block dynamics for the $k$-transformed distribution $\mu_k$ to a new stochastic process, called \emph{projected block dynamics}, defined on the original space $\{\0,\1\}^V$.

\begin{definition}[projected block dynamics] \label{definition:projected-block-dynamics-intuitive}
Let $\mu$ be a joint distribution over $\{\0,\1\}^V$, where $n = |V|$.
Let $k \geq 1$ be an integer and $\mu_k = \Rd(\mu, k)$.
Let $1 \leq \ell \leq kn$ be an integer and $\tp{X_t}_{t\geq 0}$ the uniform $\ell$-block dynamics for $\mu_k$.
The \emph{$(k, \ell)$-projected-block dynamics} $\Pproj[k, \ell]$ is a stochastic process defined as $\tp{\p:{X_t}}_{t\geq 0}$.
\end{definition}

Not surprisingly, the stochastic process $\Pproj[k, \ell]$ is a well-defined reversible Markov chain on space $\Omega(\mu)$ with stationary distribution $\mu$.
%

\begin{proposition}
\label{lemma-proj-block-Markov-chain}
Let $\mu$ be a distribution over $\{\0,\1\}^V$, where $n = |V|$. For all integers $k \geq 1$ and $1 \leq \ell \leq kn$, the \emph{$(k, \ell)$-projected-block dynamics} $\Pproj[k, \ell]$ is a  Markov chain on  $\Omega(\mu)$ that is reversible with respect to~$\mu$.
\end{proposition}

We use $\spgap{gap}{}(\Pproj[k, \ell])$ to denote the spectral gap of the reversible chain $\Pproj[k, \ell]$.
We first observe that the projection in~\eqref{eq-def-sigma-star} applied on the uniform block dynamics $P_{k,\ell}$ does not decrease the spectral gap. 

\begin{lemma} \label{lemma:projected-block-dynamics-vs-block-dynamics}
  For all integers $k \geq 1$ and $1 \leq \ell \leq kn$,  
\[
\spgap{gap}{}\tp{\Pproj[k,\ell]} \geq \spgap{gap}{}\tp{P_{k,\ell}}.
\]
\end{lemma}

Note that the projected bock dynamics $\Pproj[k, \ell]$ has the same state space $\Omega(\mu)$ as the field dynamics $\Plim$.
Therefore, we can compare their transition matrices entry-wisely.
The following is a key lemma which states that the  projected bock dynamics $\Pproj[k, \ell]$ with $\ell=\lceil \theta k n\rceil$ gives an entry-wise approximation of the field dynamics $\Plim$ as $k$ grows to infinity.

\begin{lemma} \label{lemma:limit-of-projected-chain-is-limit-chain}
  The following holds for all $0 < \theta < 1$.
  For any $\varepsilon > 0$, there is a finite $K = K(\mu,\theta,\epsilon) \geq 1$ such that for all $k \geq K$, and all $\*X, \*Y \in \Omega(\mu)$,
  \begin{align*}
    \left\vert \Plim(\*X, \*Y) - \Pproj[k, \lceil \theta k n\rceil](\*X, \*Y) \right\vert &\leq \varepsilon.
  \end{align*}
\end{lemma}

We are now ready to prove
\Cref{lemma:block-dynamics-spectral-gap-implies-limit-chain-spectral-gap}
by assuming \Cref{lemma-proj-block-Markov-chain}, \Cref{lemma:projected-block-dynamics-vs-block-dynamics} and \Cref{lemma:limit-of-projected-chain-is-limit-chain},
whose proofs are postponed to \Cref{section-analysis-projected-block-dynamics}.


\begin{proof}[Proof of \Cref{lemma:block-dynamics-spectral-gap-implies-limit-chain-spectral-gap}]
By \Cref{theorem-field-dynamics-basic}, the field dynamics $\Plim$ is reversible with respect to distribution $\mu$.
Fix an integer $k \geq 1$. Let $\ell =\ell(k) \triangleq \lceil \theta k n \rceil$.
By \Cref{lemma-proj-block-Markov-chain}, the projected block dynamics $\Pproj[k,\ell]$ is also reversible with respect to distribution $\mu$.
Let $\Omega = \Omega(\mu)$ denote the support of $\mu$.
By the Courant-Fischer theorem~\cite[Lemma 13.7]{levin2017markov},
we have
  \begin{align*}
    \spgap{gap}{}\tp{\Plim} &= \inf_{\substack{f \in \mathbb{R}^\Omega\\ \norm{f}_{2,\mu} = 1, f \perp_{\mu} \boldsymbol{1}}}\inner{f}{\tp{I - \Plim}f}_{\mu}\\
    &= \inf_{\substack{f \in \mathbb{R}^\Omega\\ \norm{f}_{2,\mu} = 1 , f \perp_{\mu} \boldsymbol{1}}}\left(\inner{f}{\tp{I - \Pproj[k,\ell]} f}_{\mu} - \inner{f}{\tp{\Plim - \Pproj[k,\ell]}f}_{\mu}\right).
  \end{align*}
Since  $\mu$, $\theta$ and $k$ are all fixed, and $f$ is a bounded function due to $\norm{f}_{2,\mu} = 1$, both $\inner{f}{\tp{I - \Pproj[k, \ell]}f}_{\mu}$ and $\inner{f}{\tp{\Plim - \Pproj[k, \ell]}f}_{\mu}$ are bounded. Therefore, the above quantity can be bounded from below as:  
\begin{align*}
\spgap{gap}{}\tp{\Plim} 
&= \inf_{\substack{f \in \mathbb{R}^\Omega\\ \norm{f}_{2,\mu} = 1 , f \perp_{\mu} \boldsymbol{1}}}\left(\inner{f}{\tp{I - \Pproj[k,\ell]} f}_{\mu} - \inner{f}{\tp{\Plim - \Pproj[k,\ell]}f}_{\mu}\right)\\
&\ge \inf_{\substack{f \in \mathbb{R}^{\Omega} \\ \norm{f}_{2,\mu} = 1,f \perp_{\mu} \boldsymbol{1}}} \inner{f}{\tp{I - \Pproj[k, \ell]}f}_{\mu} - \sup_{\substack{f \in \mathbb{R}^\Omega \\ \norm{f}_{2,\mu} = 1}} \inner{f}{\tp{\Plim - \Pproj[k, \ell]}f}_{\mu}\\
    &= \spgap{gap}{}\tp{\Pproj[k, \ell]} - \lambda_{\max}\tp{\Plim - \Pproj[k, \ell]}.
\end{align*}
The last equation holds because $\Plim - \Pproj[k, \ell]$ is a self-adjoint operator for inner-product $\inner{\cdot}{\cdot}_\mu$.
  
  By \Cref{lemma:limit-of-projected-chain-is-limit-chain}, we know that for any $\epsilon > 0$, there exists $K \geq 1$ such that for all $k \geq K$, 
  \begin{align*}
  \forall \*X, \*Y \in \Omega(\mu), \quad 
    \left\vert \Plim(\*X, \*Y) - \Pproj[k, \ell](\*X, \*Y) \right\vert &\leq \frac{\varepsilon}{\left\vert \Omega(\mu) \right\vert},
  \end{align*}
  where $ \ell = \ell(k) \triangleq \lceil \theta k n \rceil$.
  Therefore,
  \begin{align*}
    \lambda_{\max}\tp{\Plim - \Pproj[k, \ell]}
    &\leq \max_{\*X \in \Omega(\mu)} \sum_{\*Y \in \Omega(\mu)} \left\vert \Plim(\*X, \*Y) - \Pproj[k, \ell](\*X, \*Y) \right\vert \leq \varepsilon.
  \end{align*}
  Hence, for all $k \geq K$ and $\ell = \lceil \theta k n \rceil$, 
  \begin{align*}
    \spgap{gap}{}\tp{\Plim} &\geq \spgap{gap}{}\tp{\Pproj[k, \ell]} - \varepsilon.
  \end{align*}
  Applying~\Cref{lemma:projected-block-dynamics-vs-block-dynamics}, we have that for any $\epsilon > 0$, there exists $K\ge1$ such that for all $k \geq K$  and $\ell = \lceil \theta k n \rceil$, 
  \begin{align*}
    \spgap{gap}{}\tp{\Plim} &\geq \spgap{gap}{}\tp{P_{k, \ell}} - \varepsilon.
  \end{align*}
  Since $0 \leq \spgap{gap}{}\tp{P_{k, \ell}} \leq 2$ for all $k \geq 1$, $\limsup_{k \to \infty} \spgap{gap}{}\tp{P_{k, \ell}}$ exists. We have
  \begin{align*}
    \spgap{gap}{}\tp{\Plim} &\geq \limsup_{k \to \infty} \spgap{gap}{}\tp{P_{k, \ell}}. \qedhere
  \end{align*}	
\end{proof}

\section{Approximation of Field Dynamics}
\label{section-analysis-projected-block-dynamics}
This section is dedicated to the analyses of the projected block dynamics $\Pproj[k, \ell]$, which is an approximation of the field dynamics. 
In  \Cref{section-proof-proj-block-Markov-chain}, 
we prove \Cref{lemma-proj-block-Markov-chain} for the well-defined-ness and  reversibility of this projected chain;
then in \Cref{section-proof-projected-block-dynamics-vs-block-dynamics},
we prove \Cref{lemma:projected-block-dynamics-vs-block-dynamics} which relates its spectral gap to the block dynamics $P_{k,\ell}$ for the $k$-transformed distribution $\mu_k$;
and finally in \Cref{section-limit-of-projected-chain-is-limit-chain}, we prove \Cref{lemma:limit-of-projected-chain-is-limit-chain} for its entry-wise approximation of the field dynamics $\Plim$.
Altogether, they imply \Cref{lemma:block-dynamics-spectral-gap-implies-limit-chain-spectral-gap}.

\subsection{Well-definedness of projected dynamics (proof of \Cref{lemma-proj-block-Markov-chain})}
\label{section-proof-proj-block-Markov-chain}
We show that the projected block dynamics in \Cref{definition:projected-block-dynamics-intuitive} is precisely the following Markov chain $\+M$.
The chain $\+M$ starts from an arbitrary $\*X \in \Omega(\mu)$. 
In each transition step, the current configuration $\*X \in \Omega(\mu)$ is updated  as:
\begin{itemize}
\item Sample $\*a = (a_v)_{v \in V}$ according to the multivariate hypergeometric distribution $\HyperGeo$, and let $\*b = \frac{\*a}{k}$.
\item Construct a random $S \subseteq V$ by independently selecting each $v \in V$ into $S$ with probability
\begin{align*}
q_v = \begin{cases}
 1 &\text{if } X_v = \0,\\
 b_v &\text{if }	X_v = \1.
 \end{cases}
 \end{align*}
\item Replace $X_S$ by a random partial configuration sampled according to $\mu_S^{(\*b, S),X_{V \setminus S}}$, where $\mu_S^{(\*b, S),X_{V \setminus S}}$ denotes the marginal distribution on $S$ induced from $\mu^{(\*\phi)}$ conditional on $X_{V \setminus S}$, where $\*\phi \in \mathbb{R}_{\geq 0}^V$ is defined as
\begin{align}
\label{eq-def-M-field}
\forall v \in V,\quad \phi_v = \begin{cases}
b_v &\text{if } v \in S,\\
 1 & \text{if } v \notin S.
 \end{cases}
\end{align}
\end{itemize}

Note that there may exist a subset $H \subseteq V$ such that $\phi_H = \*0_H$, where $\*0_H$ denotes the all-zero vector on $H$. And by definition of $\Mag{\mu}{\*\phi}$, this is equivalent to conditioning on the configuration on $H$ being fixed as $(-\*1)_H$, which is the all-$(\0)$ configuration on $H$.
For all such $v \in H$, it must hold that $v \in S$ and $b_v = 0$, which means $X_v = \0$. 
Since $\*X \in \Omega(\mu)$, we have $\mu_H((-\*1)_H) > 0$.
Hence, $\mu^{(\*\phi)}$ is well-defined.
It is straightforward to verify $\*X$ is feasible with respect to $\mu^{(\*\phi)}$, thus $\mu_S^{(\*b, S),X_{V \setminus S}}$ is also well-defined.

\begin{lemma}
\label{lemma-M-and-proj}
The Markov chain $\+M$ is precisely the projected block dynamics $\Pproj[k, \ell]$.	
Formally, for all configurations $\*X,\*Y \in \Omega(\mu)$, we have the following identity:
\begin{align}
\label{eq-M-and-proj}
\forall \sigma \in \Omega(\mu_k) \text{ that } \p:\sigma = \*X, \quad \+M(X,Y) = \sum_{\tau \in \Omega(\mu_k): \p:\tau = \*Y}P_{k,\ell}(\sigma,\tau). 
\end{align}
\end{lemma}
The identity~\eqref{eq-M-and-proj} automatically confirms the Markovian property of $\Pproj[k, \ell]$, because it confirms that in the pre-projection chain $P_{k,\ell}$, the transition probability 
$\sum_{\tau \in \Omega(\mu_k): \p:\tau = \*Y}P_{k,\ell}(\sigma,\tau)$ from any $\sigma \in \Omega(\mu_k)$ to the class of configurations $\tau \in \Omega(\mu_k)$ projected to the same $\p:\tau = \*Y$, is constant for all $\sigma \in \Omega(\mu_k)$ that $\p:\sigma = \*X$.

Assuming that \Cref{lemma-M-and-proj} holds, we can prove \Cref{lemma-proj-block-Markov-chain}.  

\begin{proof}[Proof of \Cref{lemma-proj-block-Markov-chain} assuming \Cref{lemma-M-and-proj}]
It suffices to show that the Markov chain $\+M$ is reversible with respect to $\mu$. 
Fix two feasible configurations $\*X,\*Y \in \Omega(\mu)$. Recall that $X^{-1}(\1)\triangleq\{v \in V \mid X_v = \1\}$ denotes the pre-image of $\1$ under $\*X$.
In each transition step, the chain $\+M$ first generates a vector $\*b$ with probability $\HyperGeo(k \*b)$, and  then samples a random subset $S \subseteq V$.
Consider $R = V \setminus S$.
To transform from $\*X$ to $\*Y$, it must hold that $R \subseteq X^{-1}(\1) \cap Y^{-1}(\1)$.
We denote $\overline{R} \triangleq V \setminus R=S$.
The following equation holds:
\begin{align}
\label{eq-transition-M}
\+M(\*X,\*Y) &= \sum_{\*b:k\*b \in \Omega(\HyperGeo)}\HyperGeo(k\*b) \sum_{ R \subseteq X^{-1}(\1) \cap Y^{-1}(\1) } \tp{\prod_{v \in X^{-1}(\1) \setminus R}b_v\prod_{v \in R}(1-b_v)}\mu^{(\*b,\overline{R} ),\*1_R}_{\overline{R}}(Y_{\overline{R}})\notag\\
&=\sum_{\*b:k\*b \in \Omega(\HyperGeo)}\HyperGeo(k\*b) \sum_{R \subseteq X^{-1}(\1) \cap Y^{-1}(\1) } \tp{\prod_{v \in X^{-1}(\1) \setminus R}b_v\prod_{v \in R}(1-b_v)}\mu^{(\*b,\overline{R} ),\*1_R}(\*Y),
\end{align}
where the last equation holds because $Y_R = \*1_R$.  Also due to $Y_R = \*1_R$,  we have
\begin{align*}
\mu^{(\*b,\overline{R} ),\*1_R}(\*Y) = \frac{\mu^{(\*b,\overline{R})}(\*Y) }{ \mu^{(\*b,\overline{R})}_{R}(\*1_R)}	.
\end{align*}
By definition of $\mu^{(\*b,\overline{R})}=\Mag{\mu}{\*\phi}$ where $\*\phi\in\mathbb{R}_{\ge0}^V$ is as defined in~\eqref{eq-def-M-field}, we have
\begin{align*}
\mu^{(\*b,\overline{R})}(\*Y) = \frac{1}{Z}\mu(\*Y)\prod_{v \in Y^{-1}(\1) \setminus R}b_v,	
\end{align*}
where $Z = Z(\mu,\*b,R) \triangleq \sum_{\sigma \in \Omega(\mu)}\mu(\sigma)\prod_{v \in \sigma^{-1}(\1) \setminus R}b_v  $. 
Hence, $\mu(\*X)\+M(\*X,\*Y)$ can be expressed as
\begin{align*}
\frac{\mu(\*X)\mu(\*Y)}{Z}\sum_{\*b:k\*b \in \Omega(\HyperGeo)}\HyperGeo(k\*b) \sum_{R \subseteq X^{-1}(\1) \cap Y^{-1}(\1) } \tp{\frac{1}{\mu^{(\*b,\overline{R})}_{R}(\*1_R)} \prod_{v \in X^{-1}(\1) \setminus R}b_v\prod_{v \in Y^{-1}(\1) \setminus R}b_v\prod_{v \in R}(1-b_v)},
\end{align*}
which is symmetric in $\*X$ and $\*Y$. Therefore, the detailed balance equation is satisfied:
\begin{align*}
\mu(\*X)\+M(\*X,\*Y)= \mu(\*Y)\+M(\*Y,\*X).
\end{align*}
The chain $\+M$ is reversible with respect to $\mu$.
\end{proof}

%
%
Recall that  we use $V_k \triangleq V \times [k]	$
to denote the ground set of $\mu_k$. For each $v \in V$ and $i \in [k]$, we denote $v_i \triangleq (v,i) \in V_k$ and $C_v \triangleq \{v_i \mid i \in [k]\}$.

Let  $\Lambda \subseteq V_k$ and  $\rho \in \Omega(\mu_{k,\Lambda})$.
Define
\begin{align*}
F(\rho) \triangleq \{v \in V \mid \exists i \in [k] \text{ s.t. } v_i \in \Lambda \text{ and } \rho_{v_i} = \1 \}.	
\end{align*}
Define the local fields $\* \phi_\rho$ specified by $\rho \in \Omega(\mu_{k,\Lambda})$ as
\begin{align*}
\forall v \in V,\quad \phi_\rho(v) \triangleq \begin{cases}
\frac{\abs{C_v \setminus \Lambda}}{k}, 	&\text{if } v \notin F(\rho)\\
1 &\text{if } v \in F(\rho).
 \end{cases}
\end{align*}
We need the following lemma to prove \Cref{lemma-M-and-proj}.

\begin{lemma} \label{lemma:correctness-of-projected-distribution}
	For all $\Lambda \subseteq V_k$, $\rho \in \Omega(\mu_{k,\Lambda})$, and $\xi \in \Omega(\mu)$, the distribution $\mu^{(\*\phi_\rho),\*1_{F(\rho)}}$ is well-defined and
\begin{align*}
\mu^{(\*\phi_\rho),\*1_{F(\rho)}}(\xi) = \sum_{\substack{\tau \in \Omega(\mu_k):\\ \p:\tau = \xi}}\mu^\rho_k(\tau),
\end{align*}
where $\mu^{(\*\phi_\rho),\*1_{F(\rho)}}$ is induced from $\mu^{(\*\phi_\rho)}$ conditional on $\*1_{F(\rho)}$, the all-$(\1)$ configuration on $F(\rho)$.
\end{lemma}

We first use \Cref{lemma:correctness-of-projected-distribution} to prove \Cref{lemma-M-and-proj}, and then prove \Cref{lemma:correctness-of-projected-distribution}.

\begin{proof}[Proof of \Cref{lemma-M-and-proj}]
It suffices to verify the identity in~\eqref{eq-M-and-proj}.

Consider the transition step $\sigma \to \tau$ in the block dynamics.
The block dynamics pick a subset $S \in \binom{V_k}{\ell}$ uniformly at random. 
Observe that the following two processes are equivalent:
\begin{itemize}
\item Pick a subset  $S \in \binom{V_k}{\ell}$ uniformly at random.
\item Sample $\*a = (a_v)_{v \in V}$ according to the multivariate hypergeometric distribution $\HyperGeo$; for each $v \in V$, pick a subset $S_v \in \binom{C_v}{a_v}$ uniformly at random; and finally let $S = \cup_{v \in V}S_v$.
\end{itemize}
Imagine that we use the second process to pick the subset $S=\cup_{v \in V}S_v$. 
We have 

\begin{align}
\label{eq-P-kl-sigma-tau}
  P_{k,\ell}(\sigma,\tau) &=\sum_{\*a \in \Omega(\HyperGeo) }	\HyperGeo(a) \sum_{\substack{(S_v)_{v \in V}: \\ S_v \subseteq C_v \text{ and } |S_v| = a_v}}	 \tp{\prod_{v \in V} \binom{k}{a_v}^{-1} } \mu_k^{\sigma_{V_k \setminus S}}(\tau), \,\,\text{where } S = \cup_{v \in V}S_v\notag\\
   &=\E[\*a \sim \HyperGeo]{\E[\substack{(S_v)_{v \in V}:\\ S_v \sim_U \binom{C_v}{a_v}}]{\mu_k^{\sigma_{V_k \setminus S}}(\tau)}}, \,\,\text{where } S = \cup_{v \in V}S_v,
\end{align}
where the expectation is taken over the random choices of $(S_v)_{v \in V}$, where each $S_v \sim_U \binom{C_v}{a_v}$ is sampled from $\binom{C_v}{a_v}$ uniformly and independently.
Recall that the set $F(\sigma)\subseteq V$ is defined as
\begin{align}
\label{eq-def-F-sigma}
F(\sigma) \triangleq \{v \in V \mid \exists\, i^* \in [k] \text{ s.t. } v_{i^*} \in V_k \text{ and } \sigma_{v_{i^*}} = \1 \}.	
\end{align}
Note that for all $v \in V$,  the index $i^* \in [k]$ is unique.
Since $\p:\sigma = \*X$, it is straightforward to verify that 
\begin{align*}
\*X^{-1}(\1)=F(\sigma).	
\end{align*}
Given a sequence of sets $(S_v)_{v \in V}$, where $S_v \subseteq C_v$, we can define a subset of $\*X^{-1}(\1)$ as follows
\begin{align*}
R_{(S_v)_{v \in V}} \triangleq \{v \in \*X^{-1}(\1) \mid  v_{i^*} \notin S_v \}.
\end{align*}
where $i^* = i^*(v) \in [k]$ is the unique index for $v$ in~\eqref{eq-def-F-sigma}.
In other words, $R$ is the subset of $\*X^{-1}(\1)$ satisfying that the unique variable $v_{i^*}$ with $\sigma_{v_{i^*}} = \1$ is not picked by set $S$. 
%
%
Given $R \subseteq \*X^{-1}(\1)$ and a sequence of sets $(S_v)_{v \in V}$, where $S_v \subseteq C_v$, we say $(S_v)_{v \in V}$ is \emph{consistent} with $R$ if $R = R_{(S_v)_{v \in V}}$. 
By~\eqref{eq-P-kl-sigma-tau},
\begin{align*}
P_{k,\ell}(\sigma,\tau) = \E[\*a \sim \HyperGeo]{ \sum_{R \subseteq \*X^{-1}(\1)} \E[\substack{(S_v)_{v \in V}:\\ S_v \sim_U \binom{C_v}{a_v}}]{\mu_k^{\sigma_{V_k \setminus S}}(\tau) \*1[R = R_{(S_v)_{v \in V}}]}},	
\end{align*}
where $S = \cup_{v \in V} S_v$. Fix any $\*Y \in \Omega(\mu)$. To prove~\eqref{eq-M-and-proj}, we need to calculate 
\begin{align*}
\sum_{\substack{\tau \in \Omega(\mu_k): \\ \p:\tau = \*Y}} P_{k,\ell}(\sigma,\tau) = \E[\*a \in \HyperGeo]{ \sum_{R \subseteq \*X^{-1}(\1)} \E[\substack{(S_v)_{v \in V}:\\ S_v \sim_U \binom{C_v}{a_v}}]{\sum_{\substack{\tau \in \Omega(\mu_k): \\ \p:\tau = \*Y}} \mu_k^{\sigma_{V_k \setminus S}}(\tau) \*1[R = R_{(S_v)_{v \in V}}]}}.	
\end{align*}

To simplify the above equation, we need the following result.
Consider the equation~\eqref{eq-P-kl-sigma-tau}. We fix  $\*a  \in  \Omega(\HyperGeo)$ and $R \subseteq \*X^{-1}(\1)$. We claim the following result holds:
For any sequence $(S_v)_{v \in V}$, where $S_v \subseteq C_v$, if the following two conditions hold together:
\begin{itemize}
\item for all $v \in V$, $|S_v| = a_v$,
\item $(S_v)_{v \in V}$ is consistent with $R$,	
\end{itemize}
then for $S = \cup_{v \in V}S_v$, and for all $\*Y \in \Omega(\mu)$, it holds that
\begin{align}
\label{eq-claim-common}
\sum_{\substack{\tau \in \Omega(\mu_k^{\ }):\\ \p:\tau = \*Y}}\mu^{\sigma_{V_k \setminus S }}_k(\tau) = \mu^{(\*\varphi_{\*a, R}),\*1_R}(\*Y),
\end{align}
where $\*\varphi_{\*a, R}$ is defined as
\begin{align*}
\forall v \in V,\quad \varphi_{\*a,R}(v) \triangleq \begin{cases}
\frac{a_v}{k}	&\text{if } v \notin R,\\
1 &\text{if } v \in R.
 \end{cases}
\end{align*}
Equation~\eqref{eq-claim-common} follows from \Cref{lemma:correctness-of-projected-distribution}.
This is because $F(\sigma_{V_k \setminus S}) = R$ and $C_v \setminus (V_k \setminus S)  = C_v \cap S$ and $|C_v \cap S| = |S_v| = a_v$.

Hence, we have the following equation
\begin{align*}
\sum_{\substack{\tau \in \Omega(\mu_k): \\ \p:\tau = \*Y}} P_{k,\ell}(\sigma,\tau) &= \E[\*a \in \HyperGeo]{ \sum_{R \subseteq \*X^{-1}(\1)}  \mu^{(\*\varphi_{\*a, R}),\*1_R}(\*Y) \E[\substack{(S_v)_{v \in V}:\\ S_v \sim_U \binom{C_v}{a_v}}]{ \*1[R = R_{(S_v)_{v \in V}}]}}\\
&= \E[\*a \in \HyperGeo]{ \sum_{R \subseteq \*X^{-1}(\1)}  \mu^{(\*\varphi_{\*a, R}),\*1_R}(\*Y) \Pr[\substack{(S_v)_{v \in V}:\\ S_v \sim_U \binom{C_v}{a_v}}]{R_{(S_v)_{v \in V}}=R}}\\
&= \sum_{\*a \in \Omega(\HyperGeo) }	\HyperGeo(a)  \sum_{R \subseteq \*X^{-1}(\1) \cap \*Y^{-1}(\1)} \tp{ \prod_{v \in \*X^{-1}(\1) \setminus R}\tp{\frac{a_v}{k}} \prod_{v \in R}\tp{1-\frac{a_v}{k}} }\mu^{(\*\varphi_{\*a, R}),\*1_R}(\*Y),
\end{align*}
where the last equation holds because for each $v \in V$ and $i \in [k]$, probability of event $v_i \in S_v$ is $\frac{a_i}{k}$ and $\mu^{(\*\varphi_{\*a, R}),\*1_R}(\*Y) = 0$ if $Y_R \neq \*1_R$. 
Recall that $\*b = \*a / k$.
By~\eqref{eq-transition-M}, we have
\begin{align*}
\+M(\*X,\*Y) = \sum_{\*a \in \Omega(\HyperGeo)}\HyperGeo(\*a) \sum_{R \subseteq X^{-1}(\1) \cap Y^{-1}(\1) } \tp{\prod_{v \in X^{-1}(\1) \setminus R}\tp{\frac{a_v}{k}}\prod_{v \in R}\tp{1-\frac{a_v}{k}}}\mu^{(\*a/k,\overline{R} ),\*1_R}(\*Y),
\end{align*}
where $\overline{R} = V \setminus R$.
By the definition in~\eqref{eq-def-M-field}, $\mu^{(\*a/k,\overline{R})}$ is obtained from $\mu$ by imposing the local fields $\*\phi$, where
\begin{align*}
\forall v \in V,\quad \phi_v = \begin{cases}
b_v = \frac{a_v}{k} &\text{if } v \notin R\\
 1 & \text{if } v \in R.
 \end{cases}
\end{align*}
Hence, $\mu^{(\*a/k,\overline{R})}$ and $\mu^{(\*\varphi_{\*a, R})}$ are the same distribution, which implies 
\begin{align*}
\+M(\*X,\*Y) = 	\sum_{\substack{\tau \in \Omega(\mu_k): \\ \p:\tau = \*Y}}P_{k,\ell}(\sigma,\tau).
\end{align*}
This proves~\eqref{eq-M-and-proj}.
\end{proof}


\begin{proof}[Proof of \Cref{lemma:correctness-of-projected-distribution}]

For any $v \in V$ such that $\phi_\rho(v) = 0$, it must hold that $v \notin F(\rho)$ and $|C_v \setminus \Lambda| = 0$, which implies $C_v \subseteq \Lambda$ and $\rho_{v_i} = \0$ for all $i \in [k]$. 
Therefore, for any $\sigma \in \Omega(\mu_k^{\rho})$, $\sigma^\star$ (defined in~\eqref{eq-def-sigma-star}) is a feasible configuration of $\mu^{(\*\phi_\rho),\*1_{F(\rho)}}$. This guarantees that $\mu^{(\*\phi_\rho),\*1_{F(\rho)}}$ is a well-defined distribution.

Fix $\xi \in \Omega(\mu)$. 
By definition of distribution $\mu_k$, we have
\begin{align*}
\sum_{\substack{\tau \in \Omega(\mu_k):\\ \p:\tau = \xi}}\mu^\rho_k(\tau) &= \sum_{\substack{\tau \in \Omega(\mu_k):\\ \p:\tau = \xi}}\frac{ \Pr[\*Y \sim \mu_k]{Y_\Lambda = \rho \text{ and } \*Y = \tau} }{\Pr[\*Y \sim \mu_k]{Y_\Lambda = \rho}} = \sum_{\substack{\tau \in \Omega(\mu_k):\\ \p:\tau = \xi, \tau_\Lambda = \rho}}\frac{ \Pr[\*Y \sim \mu_k]{\*Y = \tau} }{\Pr[\*Y \sim \mu_k]{Y_\Lambda = \rho}}.
\end{align*}
Note that if $\p:\tau = \xi$, then it holds that $\norm{\tau}_{\1} = \norm{\xi}_{\1}$. 
It holds that
\begin{align*}
\sum_{\substack{\tau \in \Omega(\mu_k):\\ \p:\tau = \xi, \tau_\Lambda = \rho}} 	 \Pr[\*Y \sim \mu_k]{\*Y = \tau}& = \mu(\xi)\tp{\frac{1}{k}}^{\norm{\xi}_{\1}} \sum_{\substack{\tau \in \Omega(\mu_k):\\ \p:\tau = \xi}}\*1[\tau_\Lambda = \rho]\\
 &= \mu(\xi)\tp{\frac{1}{k}}^{\norm{\xi}_{\1}} \*1[\xi_{F(\rho)} = \*1_{F(\rho)}] \prod_{v \in \xi^{-1}(\1) \setminus F(\rho)}\abs{C_v \setminus \Lambda}.
\end{align*}
The last equation holds because:
\begin{itemize} 
\item 
if $\p:\tau = \xi$ and $\tau_\Lambda = \rho$, then it must hold that $\xi_{F(\rho)}=\*1_{F(\rho)}$; 
\item 
for all $v \in \xi^{-1}(\1)$, we need to choose one index $i \in [k]$ and set $\tau_{v_i} = \1$, if $v \in F(\rho)$ such index is fixed by $\rho$, if $v \in \xi^{-1}(\1) \setminus F(\rho)$, there are $|C_v \setminus \Lambda|$ ways to choose such index.
\end{itemize}
Reorganizing above equation gives
\begin{align*}
\sum_{\substack{\tau \in \Omega(\mu_k):\\ \p:\tau = \xi, \tau_\Lambda = \rho}} 	 \Pr[\*Y \sim \mu_k]{\*Y = \tau} = \mu(\xi)  \*1[\xi_{F(\rho)} = \*1_{F(\rho)}]  \tp{\frac{1}{k}}^{\abs{F(\rho)}} \prod_{v \in \xi^{-1}(\1) \setminus F(\rho)}\frac{\abs{C_v \setminus \Lambda} }{k}.
\end{align*}
Next, we have
\begin{align*}
\Pr[\*Y \sim \mu_k]{Y_\Lambda = \rho} = \sum_{\eta \in \Omega(\mu)} \mu(\eta) \*1[ \eta_{F(\rho)} = \*1_{F(\rho)} ] \tp{\frac{1}{k}}^{|F(\rho)|}\prod_{v \in \eta^{-1}(\1) \setminus F(\rho)}\frac{\abs{C_v \setminus \Lambda} }{k}.
\end{align*}
By definition, $\*Y \sim \mu_k$ is obtained by first sampling $\*X \sim \mu$, then transforming $\*X$ to $\*Y$. The above equation enumerates all values $\eta$ for $\*X$. The event $Y_{\Lambda} = \rho$ occurs if and only if:
\begin{itemize}
\item $\eta_v = \1$ for all $v \in F(\rho)$; 
\item for all $v \in F(\rho)$, there is a unique index $i \in [k]$ fixed by $\rho$ such that $v_i = \1$ (this event occurs with probability $\tp{\frac{1}{k}}^{|F(\rho)|}$); 
\item for all other $v \in \eta^{-1}(\1) \setminus F(\rho)$, we choose an index $i \in [k]$ to set $v_i = \1$ and $v_i \notin \Lambda$ (this event occurs with probability $\prod_{v \in \eta^{-1}(\1) \setminus F(\rho)}\frac{\abs{C_v \setminus \Lambda} }{k}$).
\end{itemize}

Combining above two equations together, we have
\begin{align*}
\sum_{{\tau \in \Omega(\mu_k): \p:\tau = \xi}}\mu^\rho_k(\tau) 
&= 	\sum_{{\tau \in \Omega(\mu_k): \p:\tau = \xi, \tau_\Lambda = \rho}}\frac{ \Pr[\*Y \sim \mu_k]{\*Y = \tau} }{\Pr[\*Y \sim \mu_k]{Y_\Lambda = \rho}}\\
&= \frac{\mu(\xi)  \*1[\xi_{F(\rho)} = \*1_{F(\rho)}] \prod_{v \in \xi^{-1}(\1) \setminus F(\rho)}\frac{\abs{C_v \setminus \Lambda} }{k}}{\sum_{\eta \in \Omega(\mu)} \mu(\eta) \*1[ \eta_{F(\rho)} = \*1_{F(\rho)} ] \prod_{v \in \eta^{-1}(\1) \setminus F(\rho)}\frac{\abs{C_v \setminus \Lambda} }{k}}\\
&= \frac{\*1[\xi_{F(\rho)} = \*1_{F(\rho)}] \mu^{(\*\phi_\rho)}(\xi)}{\sum_{\eta \in \Omega(\mu) } \*1[ \eta_{F(\rho)} = \*1_{F(\rho)} ]\mu^{(\*\phi_\rho)}(\eta) }\\
 &=\frac{\*1[\xi_{F(\rho)} = \*1_{F(\rho)}] \mu^{(\*\phi_\rho)}(\xi)}{\sum_{\eta \in \Omega(\mu^{(\*\phi_\rho)}) } \*1[ \eta_{F(\rho)} = \*1_{F(\rho)} ]\mu^{(\*\phi_\rho)}(\eta) } \\
&= \mu^{(\*\phi_\rho),\*1_{F(\rho)}}(\xi),
\end{align*}
where the second to the last equation holds because $\Omega(\mu^{(\*\phi_\rho)} ) \subseteq \Omega(\mu)$. This proves the lemma.
\end{proof}


\subsection{Comparing projected block dynamics with block dynamics (proof of \Cref{lemma:projected-block-dynamics-vs-block-dynamics})}
\label{section-proof-projected-block-dynamics-vs-block-dynamics}
By \Cref{lemma-proj-block-Markov-chain}, $\Pproj[k, \ell]$ is a reversible Markov chain.
We establish the following lemma which implies \Cref{lemma:projected-block-dynamics-vs-block-dynamics}.

\begin{lemma} \label{lemma:dirichlet}
For any function $f: \Omega(\mu) \to \mathbb{R}$ satisfying $\Var[\mu]{f} \neq 0$, there exists a function $f': \Omega(\mu_k) \to \mathbb{R}$ such that  $\Var[\mu_k]{f'} \neq 0$ and 
\begin{align*}
\frac{\+E_{\Pproj[k, \ell]}(f,f)}{\mathbf{Var}_{\mu}(f)} = \frac{\+E_{P_{k, \ell}}(f',f')}{\mathbf{Var}_{\mu_k}(f')}.
\end{align*}
\end{lemma}

With this lemma,
\Cref{lemma:projected-block-dynamics-vs-block-dynamics} follows immediately from the characterization of spectral gap in~\eqref{eq-comp-gap-def}: 
\[\spgap{gap}{}(P) = \inf_{\substack{f\in\mathbb{R}^{\Omega(\mu)}\\ \Var[\mu]{f} \neq 0}} \frac{\+E_P(f,f)}{\Var[\mu]{f}}.
\]


Now, we prove \Cref{lemma:dirichlet}.

\begin{proof}[Proof of \Cref{lemma:dirichlet}]
  Given a function $f: \Omega(\mu) \to \mathbb{R}$ satisfying $\Var[\mu]{f} \neq 0$, the function $f': \Omega(\mu_k) \to \mathbb{R}$ is constructed as follows:
  \begin{align}
    \label{eq-def-f'}
    \forall \xi \in \Omega(\mu_k), \quad f'(\xi) = f(\p:\xi),
  \end{align}
where $\p:\xi$ is as defined in~\eqref{eq-def-sigma-star}.

Next, we prove the following identities that implies the lemma:
\begin{align*}
\+E_{P_{k,\ell}}(f',f') 
&= \+E_{\Pproj[k,\ell]}(f,f),\\
\mathbf{Var}_{\mu_k}(f') 
&= \mathbf{Var}_{\mu}(f).
\end{align*}
 
For any $\*X,\*Y \in \Omega(\mu)$, by the definition of $\mu_k$, it holds that
    \begin{align} \label{eq:map-mu}
     \forall \sigma \in \Omega(\mu_k) \text{ satisfying }\sigma^\star = \*X, \quad  \mu(\sigma) &= \sum_{\substack{\xi \in \Omega(\mu_k) \\ \p:\xi = \sigma}} \mu_k(\xi).
    \end{align}
    
  By~\eqref{eq-M-and-proj}, it holds that
    \begin{align} \label{eq:map-chain}
     \forall \sigma \in \Omega(\mu_k) \text{ satisfying }\sigma^\star = \*X, \quad   \Pproj[k, \ell](\*X, \*Y) &= \sum_{\substack{\tau \in \Omega(\mu_k) \\ \p:\tau = \*Y}} P_{k, \ell}(\sigma,\tau).
    \end{align}
  
Therefore, we have
 \begin{align*}
  \+E_{P_{k,\ell}}(f',f') &= \frac{1}{2}\sum_{\sigma, \tau \in \Omega(\mu_k)}\mu_k(\sigma)P_{k,\ell}(\sigma,\tau)\left(f'(\sigma)-f'(\tau)\right)^2\\	
 &= \frac{1}{2}\sum_{\*X,\*Y \in \Omega(\mu)} \sum_{\substack{\sigma,\tau \in \Omega(\mu_k) \\   \p:\sigma = \*X,  \p:\tau = \*Y} }\mu_k(\sigma)P_{k,\ell}(\sigma,\tau)\left(f'(\sigma)-f'(\tau)\right)^2\\
 \text{(by definition of $f'$)}\quad &=\frac{1}{2}\sum_{\*X,\*Y \in \Omega(\mu)} \sum_{\substack{\sigma,\tau \in \Omega(\mu_k) \\   \p:\sigma = \*X,  \p:\tau = \*Y} }\mu_k(\sigma)P_{k,\ell}(\sigma,\tau)\left(f(\*X)-f(\*Y)\right)^2\\
 &=\frac{1}{2}\sum_{\*X,\*Y \in \Omega(\mu)}(f(\*X)-f(\*Y))^2\sum_{\substack{\sigma \in \Omega(\mu_k)\\\p:\sigma = \*X}}\mu_k(\sigma)\sum_{\substack{\tau \in \Omega(\mu_k)\\\p:\tau = \*Y}}P_{k,\ell}(\sigma,\tau)\\
 \text{(by \eqref{eq:map-mu} and \eqref{eq:map-chain})}\quad&=\frac{1}{2}\sum_{\*X,\*Y \in \Omega(\mu)}\mu(\*X)\Pproj[k,\ell](\*X,\*Y)(f(\*X)-f(\*Y))^2 = \+E_{\Pproj[k,\ell]}(f,f).
 \end{align*}

Similarly, 
\begin{align*}
\Var[\mu_k]{f'} &=	\frac{1}{2}\sum_{\sigma, \tau \in \Omega(\mu_k) }\mu_k(\sigma)\mu_k(\tau)\left(f'(\sigma)-f'(\tau)\right)^2\\
&=\frac{1}{2}\sum_{\*X,\*Y \in \Omega(\mu)}(f(\*X)-f(\*Y))^2\sum_{\substack{\sigma \in \Omega(\mu_k)  \\ \p:\sigma = \*X}}\mu_k(\sigma)\sum_{\substack{ \tau \in \Omega(\mu_k) \\ \p:\tau = \*Y }}\mu_k(\tau)\\
(\text{by \eqref{eq:map-mu}}) \quad &=\frac{1}{2}\sum_{\*X,\*Y \in \Omega(\mu)}\mu(\*X)\mu(\*Y)(f(\*X)-f(\*Y))^2 = \Var[\mu]{f}.
\end{align*}
This proves the lemma.
\end{proof}

\subsection{Approximation of the field dynamics (proof of \Cref{lemma:limit-of-projected-chain-is-limit-chain})}\label{section-limit-of-projected-chain-is-limit-chain}
Let $n = |V|$.
Fix a real number $0 < \theta < 1$.
Given $\epsilon > 0$, we show that for any $\*X,\*Y \in \Omega(\mu)$, there is a $K = K(\*X,\*Y, \mu, \theta, \epsilon)$ such that for any $k \geq K$,
\begin{align}
\label{eq-proof-limit-target}
    \left\vert \Plim(\*X, \*Y) - \Pproj[k, \lceil \theta k n\rceil](\*X, \*Y) \right\vert &\leq \varepsilon.	
\end{align}
Note that $\Omega(\mu)$ is independent of $k$.
\Cref{lemma:limit-of-projected-chain-is-limit-chain}  follows by taking $K(\mu, \theta, \epsilon) = \max_{\*X,\*Y \in \Omega(\mu)}K(\*X,\*Y, \mu, \theta, \epsilon)$.

Fix $\*X,\*Y \in \Omega(\mu)$ and integer $k \geq 1$. Let $\ell \triangleq \lceil \theta nk\rceil$.
By \Cref{lemma-M-and-proj} and \Cref{eq-transition-M},
\begin{align}\label{eq:project}
  \Pproj[k,\ell](\*X,\*Y) = \sum_{\*a:\*a \in \Omega(\HyperGeo)}\HyperGeo(\*a) \sum_{R \subseteq X^{-1}(\1) \cap Y^{-1}(\1) } \tp{\prod_{v \in X^{-1}(\1) \setminus R}\frac{a_v}{k}\prod_{v \in R}\tp{1-\frac{a_v}{k}}}\mu^{(\*a/k,\overline{R} ),\*1_R}(\*Y),
\end{align}
where $\mu^{(\*a/k,\overline{R})}$ is distribution $\mu^{(\phi)}$ for local fields $\*\phi \in \mathbb{R}^V_{\ge 0}$ defined by
\begin{align*}
  \forall v \in V,\quad \phi_v = \begin{cases}
  \frac{a_v}{k} &\text{if } v \notin R\\
   1 & \text{if } v \in R.
   \end{cases}
\end{align*}

Let $f_{\mu,
\*X,\*Y}(\*w): (0,1)^V \rightarrow \mathbb{R}$ be a function defined as follows:
\begin{align*}
  f_{\mu,\*X,\*Y}(\*w) &= \sum_{R \subseteq \*X^{-1}(\1) \cap \*Y^{-1}(\1)} \tp{\prod_{v \in \*X^{-1}(\1) \setminus R} w_v \prod_{v \in R} (1-w_v)} \mu^{(\*w,\overline{R} ),\*1_R}(\*Y)\\
  &=\sum_{R \subseteq \*X^{-1}(\1) \cap \*Y^{-1}(\1)} \tp{\prod_{v \in \*X^{-1}(\1) \setminus R} w_v \prod_{v \in R} (1-w_v)}  \frac{\mu(\*Y) \prod_{v \in \*Y^{-1}(\1) \setminus R  } w_v}{\sum_{\*Z \in \Omega(\mu): Z_R = \*1_R }\mu(\*Z) \prod_{v \in \*Z^{-1}(\1) \setminus R } w_v  }
\end{align*}
Since $\*Y \in \Omega(\mu)$ and $Y_R = \*1_R$, for all $\*w \in (0,1)^V$, it holds that $\sum_{\*Z \in \Omega(\mu): Z_R = \*1_R }\mu(\*Z) \prod_{v \in \*Z^{-1}(\1) \setminus R } w_v > 0$.
The function $f_{\mu,\*X,\*Y}(\*w)$ is a rational function. 
Thus $f_{\mu,\*X,\*Y}(\*w)$ is continuous on $(0,1)^V$. Formally, we have:
\begin{fact}\label{fact:continuous}
For any $\epsilon>0$ and $\*w \in (0,1)^V$, there exists a constant $\delta=\delta(\mu,\*X,\*Y,\*w,\epsilon) > 0$ such that 
for all $\*w' \in (0,1)^V$ satisfying $\norm{\*w-\*w'}_{\infty} < \delta$, it holds that
  \begin{align*}
    \abs{f_{\mu,\*X,\*Y}(\*w)-f_{\mu,\*X,\*Y}(\*w')} < \epsilon,
  \end{align*}
where $\norm{\*w - \*w'}_{\infty} \triangleq \max_{v \in V}|w(v) - w'(v)|$ is the infinity norm of $\*w - \*w'$.
\end{fact}

We now prove~\eqref{eq-proof-limit-target}. Fix a parameter $\epsilon > 0$. 
Let $\*\theta \triangleq (\theta)_{v \in V}=\theta\*1$. 
By \Cref{fact:continuous}, there exists a $\delta=\delta(\mu,\*X,\*Y,\theta,\epsilon) \in \tp{0, \frac{\min\{\theta, 1 - \theta\}}{2}}$ such that for all $\*w \in (0,1)^V$ satisfying $\norm{\*w-\*\theta}_{\infty} < \delta$
  \begin{align}\label{eq:bound-1}
    \abs{f_{\mu,\*X,\*Y}(\*w)-f_{\mu,\*X,\*Y}(\*\theta)} < \frac{\epsilon}{2}.
  \end{align}
Define a subset of $\+B_{\delta,k} \subseteq \Omega(\HyperGeo)$ by 
\begin{align}\label{eq:bad-event}
  \+B_{\delta,k} = \left\{\*a \in \Omega(\HyperGeo) \mid  \norm{\frac{\*a}{k} - \*\theta}_{\infty} \ge \delta\right\}.
\end{align}
Note that the function $f_{\mu,\*X,\*Y}$ satisfies 
\begin{align*}
\forall \*w \in (0,1)^V, \quad 0 < f_{\mu,\*X,\*Y}(\*w) &\leq\sum_{R \subseteq \*X^{-1}(\1) \cap \*Y^{-1}(\1)} \tp{\prod_{v \in \*X^{-1}(\1) \setminus R} w_v \prod_{v \in R} (1-w_v)} \leq 1.
\end{align*}

Combining~\eqref{eq:bound-1} and~\eqref{eq:bad-event}, for any $\*a \notin  \+B_{\delta,k}$, since $\delta \in \tp{0, \frac{\min\{\theta, 1 - \theta\}}{2}}$, $\*a$ is a positive vector and it holds that $\abs{f_{\mu,\*X,\*Y}(\*a/k)-f_{\mu,\*X,\*Y}(\*\theta)} < \frac{\epsilon}{2}$.
By~\eqref{eq:project}, on the one hand, 
\begin{align}
\label{eq-proj-lower-bound-pre}
 \Pproj[k,\ell](\*X,\*Y) 
 &\geq  \sum_{\*a \in \Omega(\HyperGeo) \setminus \+B_{\delta,k} }\HyperGeo(\*a) f_{\mu,\*X,\*Y}\tp{\frac{\*a}{k}} \notag\\
 &\geq \tp{f_{\mu,\*X,\*Y}(\*\theta) - \frac{\epsilon}{2}}\sum_{\*a \in \Omega(\HyperGeo) \setminus \+B_{\delta,k} }\HyperGeo(\*a)\notag\\
 &=\tp{f_{\mu,\*X,\*Y}(\*\theta) - \frac{\epsilon}{2}} \tp{1 - \Pr[\*a \sim \HyperGeo ]{\*a \in \+B_{\delta,k} }}.
\end{align}
This implies that 
\begin{align}
\label{eq-proj-lower-bound}
 \Pproj[k,\ell](\*X,\*Y) &\geq  f_{\mu,\*X,\*Y}(\*\theta) - \frac{\epsilon}{2} -  \Pr[\*a \sim \HyperGeo ]{\*a \in \+B_{\delta,k} }.	
\end{align}
This is because \eqref{eq-proj-lower-bound} holds trivially if $f_{\mu,\*X,\*Y}(\*\theta) \leq  \frac{\epsilon}{2}$, and  if otherwise $f_{\mu,\*X,\*Y}(\*\theta) >  \frac{\epsilon}{2}$, since $f_{\mu,\*X,\*Y}(\*\theta) \leq 1$,~\eqref{eq-proj-lower-bound} is a consequence of~\eqref{eq-proj-lower-bound-pre}.

On the other hand,
\begin{align}
\label{eq-proj-upper-bound}
 \Pproj[k,\ell](\*X,\*Y) &=	 \sum_{\*a \in \Omega(\HyperGeo)}\HyperGeo(\*a)f_{\mu,\*X,\*Y}\tp{\frac{\*a}{k}}  \notag\\
 &= \sum_{\*a \in \+B_{\delta,k}}\HyperGeo(\*a)f_{\mu,\*X,\*Y}\tp{\frac{\*a}{k}} + \sum_{\*a \in \Omega(\HyperGeo)\setminus \+B_{\delta,k} }\HyperGeo(\*a)f_{\mu,\*X,\*Y}\tp{\frac{\*a}{k}}\notag\\
&\leq \sum_{\*a \in \+B_{\delta,k}}\HyperGeo(\*a) + \tp{f_{\mu,\*X,\*Y}(\*\theta) + \frac{\epsilon}{2}}\sum_{\*a \in \Omega(\HyperGeo) \setminus \+B_{\delta,k} }\HyperGeo(\*a)\notag\\
 &= \Pr[\*a \sim \HyperGeo ]{\*a \in \+B_{\delta,k} } + \tp{f_{\mu,\*X,\*Y}(\*\theta) + \frac{\epsilon}{2}} \tp{1 - \Pr[\*a \sim \HyperGeo ]{\*a \in \+B_{\delta,k} }}\notag\\
 &\leq \Pr[\*a \sim \HyperGeo ]{\*a \in \+B_{\delta,k} } + f_{\mu,\*X,\*Y}(\*\theta) + \frac{\epsilon}{2},
\end{align}

Finally, we also claim that there is a $K = K(\mu, \*X,\*Y, \theta, \epsilon)$ such that for any $k \geq K$
\begin{align}
\label{eq-bad-event-prob}
\Pr[\*a \sim \HyperGeo ]{\*a \in \+B_{\delta,k} } \leq \frac{\epsilon}{2}.	
\end{align}
Altogether, \eqref{eq-proj-lower-bound}, \eqref{eq-proj-upper-bound} and~\eqref{eq-bad-event-prob} implies~\eqref{eq:project}. 

It remains to prove~\eqref{eq-bad-event-prob}.
Recall that $n = |V|$ and $\ell = \lceil \theta kn \rceil$.
Observe that when $k \ge \frac{2}{\delta}$,
  \begin{align}
  \label{eq:bad-1}
    \Pr[\*a \sim \HyperGeo ]{\*a \in \+B_{\delta,k}} 
    &= \Pr[\*a \sim \HyperGeo ]{\exists v \in V, \abs{\frac{a_v}{k}-\theta} \ge \delta} \notag\\
    &\le \sum_{v \in V} \Pr[\*a \sim \HyperGeo]{\abs{\frac{a_v}{k}-\theta} \ge \delta}\notag\\
    &\le \sum_{v \in V} \Pr[\*a \sim \HyperGeo]{\abs{\frac{a_v}{k} - \frac{\lceil \theta kn \rceil}{kn}} \ge \frac{\delta}{2}},
  \end{align}
where the last inequality holds for $k \geq \frac{2}{\delta}$.

Furthermore, by \Cref{lemma:hypergometric-concentration}, there exists a constant $K_0=K_0(\mu, \delta,\epsilon)$ such that for $k \ge K_0$,
  \begin{align}\label{eq:bad-2}
  \forall v \in V,\quad  \Pr[\*a \sim \HyperGeo]{\abs{\frac{a_v}{k} - \frac{\lceil \theta kn \rceil}{kn}} \ge \frac{\delta}{2}} \leq  2 \exp \tp{\frac{-\delta^2 k }{2} } \le \frac{\epsilon}{2n}.
  \end{align}
  Note that $\delta = \delta(\mu, \*X,\*Y, \theta, \epsilon)$ and $n$ is determined by $\mu$.
   Combining ~\eqref{eq:bad-1} and ~\eqref{eq:bad-2} proves~\eqref{eq-bad-event-prob}.

\section{Mixing of Block Dynamics}\label{sec:block-mixing}
In this section, we prove \Cref{lemma:block-dynamics-spectral-gap} for the mixing of the uniform block dynamics for the $k$-transformed distribution $\mu_k= \Rd(\mu,k)$, assuming the complete spectral independence of $\mu$.
Together with \Cref{lemma:block-dynamics-spectral-gap-implies-limit-chain-spectral-gap} proved in the last two sections, this proves the mixing lemma for the field dynamics (\Cref{lemma-field-dynamics-mixing})

With the Chen-Liu-Vigoda theorem (\Cref{theorem:spectral-independent-imply-spectral-gap}), to prove \Cref{lemma:block-dynamics-spectral-gap}, we only need to verify that the $k$-transformed distribution $\mu_k$ is spectrally independent knowing that the original distribution $\mu$ is  completely spectrally independent.

\begin{lemma}
\label{lemma-spind-mu-to-muk}
Let $\mu$ be a distribution over $\{\0,\1\}^V$ and $\eta > 0$.
If $\mu$ is completely $\eta$-spectrally independent, then for all integers $k \geq 1$, $\mu_k= \Rd(\mu,k)$ is $(\eta+2)$-spectrally independent.
\end{lemma}

\Cref{lemma:block-dynamics-spectral-gap} is an easy consequence of \Cref{theorem:spectral-independent-imply-spectral-gap} and \Cref{lemma-spind-mu-to-muk}. 

Our remaining task is to prove \Cref{lemma-spind-mu-to-muk}. 
Recall that  we use $V_k \triangleq V \times [k]	$
to denote the ground set of $\mu_k$. For each $v \in V$ and $i \in [k]$, we denote $v_i \triangleq (v,i) \in V_k$ and $C_v \triangleq \{v_i \mid i \in [k]\}$.

\subsection{Spectral independence of $k$-transformed distribution (proof of \Cref{lemma-spind-mu-to-muk})}
We define a mapping from feasible partial configurations for distribution $\mu_k$ to feasible partial configurations for distribution $\mu$.
%
%
Let $\Lambda \subseteq V_k$ and $\sigma \in \Omega(\mu_{k,\Lambda})$, where $\mu_{k,\Lambda}$ denotes the marginal distribution on $\Lambda$ projected from $\mu_k$. 
We define the following subsets of $V$:
\begin{equation}
\label{eq-lambda-three}
\begin{split}
\r:{\Lambda}_{\sigma,\0} &\triangleq \{v \in V \mid  \forall i \in [k],\, v_i \in \Lambda \land \sigma_{v_i} = \0  \},\\
\r:{\Lambda}_{\sigma,\1} &\triangleq  \{v \in V \mid  \exists i \in [k] \,\text{s.t.}\, v_i \in \Lambda \land \sigma_{v_i} = \1 \},\\
\r:{\Lambda}_\sigma &\triangleq \r:{\Lambda}_{\sigma,\0} \cup \r:{\Lambda}_{\sigma,\1}.
\end{split}
\end{equation}
Note that $\r:{\Lambda}_{\sigma,\0}$ and $\r:{\Lambda}_{\sigma,\1}$ are disjoint.
Let $\r:{\sigma} \in \{\0,\1\}^{\r:{\Lambda}_\sigma}$ indicate whether each $v\in\r:{\Lambda}_{\sigma}$ is in $\r:{\Lambda}_{\sigma,\1}$ or $\r:{\Lambda}_{\sigma,\0}$. Specifically, 
\begin{align}
\label{eq-def-sigma-ast}
\forall v \in \r:{\Lambda}_\sigma, \quad \r:{\sigma}_v = \begin{cases}
\0 &\text{if } v \in \r:{\Lambda}_{\sigma,\0},\\
\1 &\text{if } v \in \r:{\Lambda}_{\sigma,\1}.	
 \end{cases}
 \end{align}
In other words, $\sigma^*$ fixes a $v\in V$ to be $\0$ if $\sigma$ fixes all $v_i\in C_v$ to be $\0$; and $\sigma^*$ fixes a $v\in V$ to be $\1$ if $\sigma$ fixes some $v_i\in C_v$ to be $\1$.

Given a $\sigma \in \Omega(\mu_{k,\Lambda})$, we construct the following vector of local fields $\*\lambda_\sigma$,
\begin{align}
\label{eq-def-lambda-sigma}
\forall v \in V, \quad \lambda_\sigma(v) = \begin{cases}
1 &\text{if } v \in \r:{\Lambda}_\sigma,\\
\frac{\abs{C_v \setminus \Lambda }}{k} &\text{if } v \notin \r:{\Lambda}_\sigma.
 \end{cases}
\end{align}
By definition, for every $v \in V \setminus \r:{\Lambda}_\sigma$, it holds that $\abs{C_v \setminus \Lambda } > 0$. Hence, $\*\lambda_\sigma$ is a positive vector. 

We have the following lemma that bounds the spectral radius of the influence matrix.

\begin{lemma} \label{lemma:bounded-spectral-independence}
  For any integer $k\geq 1$,  any $\Lambda \subseteq V_k$ and  any $\sigma \in \Omega(\mu_{k,\Lambda})$, it holds that $\sigma^*$ is a feasible partial configuration with respect to $\nu \triangleq \mu^{(\*\lambda_\sigma)}$, and 
  \begin{align*}
    \rho\tp{\Psi^{\sigma}_{\mu_k}} &\leq \rho\tp{\Psi^{\r:\sigma}_{\nu}} + 2,
  \end{align*}
 where $\Psi^{\sigma}_{\mu_k}$ and $\Psi^{\r:\sigma}_{\nu}$ are the influence matrices defined as in \Cref{definition-weight-tot-inf}, and $\rho(\cdot)$ denotes the spectral radius.
\end{lemma}


\Cref{lemma-spind-mu-to-muk} immediately holds from \Cref{lemma:bounded-spectral-independence}.
\begin{proof}[Proof of \Cref{lemma-spind-mu-to-muk} assuming \Cref{lemma:bounded-spectral-independence}]
Fix an integer $k \geq 1$, a subset $\Lambda \subseteq V_k$ and a feasible partial configuration $\sigma \in \Omega(\mu_{k,\Lambda})$. By \Cref{lemma:bounded-spectral-independence}, we have
\begin{align*}
 \rho\tp{\Psi^{\sigma}_{\mu_k}} &\leq \rho\tp{\Psi^{\r:\sigma}_{\nu}} + 2, \quad \text{where } \nu \triangleq \mu^{(\*\lambda_\sigma)}. 	
\end{align*}
Note that $\nu = \mu^{(\*\lambda_\sigma)}$ and $\*\lambda_\sigma(v) \in(0, 1)$ for all $v \in V$.
%
%
Since $\mu$ is completely $\eta$-spectrally independent,
\begin{align*}
 \rho\tp{\Psi^{\r:\sigma}_{\nu}} \leq \eta.
\end{align*}
This proves the lemma. 
\end{proof}

\subsection{Spectrum-preservation of $k$-transformation (proof of \Cref{lemma:bounded-spectral-independence})}
We first prove that $\r:\sigma$ defined in~\eqref{eq-def-sigma-ast} is feasible with respect to $\nu$ as long as $\sigma \in \Omega(\mu_{k,\Lambda})$.
By definition of $\r:\sigma$ and $\mu_k$, if $\sigma \in \Omega(\mu_{k,\Lambda})$, then  $\r:\sigma$ is feasible with respect to $\mu$.
Consequently, $\r:\sigma$ is also feasible with respect to $\nu$, because $\nu = \mu^{(\*\lambda_\sigma)}$ has the same support as $\mu$ for a positive vector $\*\lambda_\sigma$.

It remains to prove 
\begin{align}
\label{eq-comp-spectral}
    \rho\tp{\Psi^{\sigma}_{\mu_k}} &\leq \rho\tp{\Psi^{\r:\sigma}_{\nu}} + 2.
  \end{align}
  
For any distribution $\nu$ over $\{\0,\1\}^V$, any $H \subseteq V$, and any feasible partial configuration $\sigma_H \in \{\0,\1\}^H$ with respect to $\nu$, 
let $\+I_{\nu}^{\sigma_H}\in\mathbb{R}_{\ge0}^{V\times V}$ be a matrix defined as:
\begin{align}
\label{eq-def-gen-inf}
\forall u,v \in V,\quad
\INF{\nu}{\sigma_{H}}{u}{v} \triangleq \max_{i,j \in \Omega\tp{\nu^{\sigma_H}_u}} \DTV{\nu_v^{\sigma_H, u\gets i}}{\nu_v^{\sigma_H, u \gets j}}.
\end{align}
Note that by definition, if $\abs{\Omega(\nu^{\sigma_H}_u)} = 1$ or $\abs{\Omega(\nu^{\sigma_H}_v)} = 1$, then $\INF{\nu}{\sigma_H}{u}{v} = 0$, and for all $u \in V$, $\INF{\nu}{\sigma_H}{u}{u} = \*1[|\Omega(\nu^{\sigma_H}_v)| \neq 1]$.

To prove~\eqref{eq-comp-spectral}, we define two new matrices $\widetilde{\Psi}^{\sigma}_{\mu_k}$ and $\widehat{\Psi}^{\r:\sigma}_{\nu}$. Recall that $\sigma \in \{\0,\1\}^\Lambda$ is a partial configuration on $\Lambda \subseteq V_k$. For each variable $v \in V$,  we denote by  
\begin{align*}
m_v \triangleq \abs{C_v \setminus \Lambda}
\end{align*}
the number of variables in $C_v$ whose values are not fixed by $\sigma$. Without loss of generality, we enumerate the variables in  $C_v \setminus \Lambda$ as
\begin{align*}
C_v \setminus \Lambda = \{v_i \mid i \in [m_v] \}, \text{where } [m_v]=\{1,2,3,\ldots,m_v\}.	
\end{align*}
The matrix  $\widetilde{\Psi}^{\sigma}\in\mathbb{R}^{(V_k \setminus \Lambda)\times (V_k \setminus \Lambda)}$ is defined by
\begin{align}
\label{eq-def-tilde}
\forall u,v \in V, i \in [m_u], j \in [m_v], \quad  \widetilde{\Psi}^{\sigma}(u_i,v_j) \triangleq \begin{cases}
 \frac{2}{m_v} &\text{if } u=v \text{ and } v \notin \r:{\Lambda}_\sigma,\\
 \frac{1}{m_v} \INF{\nu}{\r:\sigma}{u}{v} &\text{if } u \neq v \text{ or } v \in \r:{\Lambda}_\sigma.	
 \end{cases}
\end{align}
This is well-defined because it only involves those $u,v \in V$ such that $m_v \geq 1$ and $m_u \geq 1$. 
The following is easy to observe.
\begin{observation}
\label{observation-block-matrix}
For any $u,v \in V$, all the entries $\widetilde{\Psi}^{\sigma}(u_i,v_j)$ have the same value for  $i \in [m_u]$ and $j \in [m_v]$. 	
\end{observation}

Next, we define the matrix $\widehat{\Psi}^{\r:\sigma}\in\mathbb{R}^{(V \setminus \Lambda_\sigma^\ast) \times (V \setminus \Lambda_\sigma^\ast)}$, where $\Lambda_\sigma^\ast$ is defined in~\eqref{eq-lambda-three}. 
By \Cref{observation-block-matrix}, $\widetilde{\Psi}^{\sigma}$ can be decomposed into a set of blocks, all the entries in the same block have the same value. The matrix $\widehat{\Psi}^{\r:\sigma}$ is defined by compressing blocks in $\widetilde{\Psi}^{\sigma}$.
Formally,
\begin{align}
\label{eq-def-hat}
\forall u,v \in V \setminus \Lambda_\sigma^\ast, \quad \widehat{\Psi}^{\r:\sigma}(u,v) \triangleq \sum_{j \in [m_v]} \widetilde{\Psi}^{\sigma}(u_1,v_j).
\end{align}
The matrix  $\widehat{\Psi}^{\r:\sigma}$ is well-defined because by the definition of $\Lambda_\sigma^\ast$, it is straightforward to verify $m_u \geq 1$ and $m_v \geq 1$ for all $u,v \in V \setminus \r:\Lambda_\sigma$, thus the $u_1$ in \eqref{eq-def-hat} exists.
By \Cref{observation-block-matrix}, it holds that for all $u,v \in V \setminus \r:\Lambda_\sigma$,
\begin{align*}
\forall i \in [m_u],\quad \widehat{\Psi}^{\r:\sigma}(u,v) = \sum_{j \in [m_v]} \widetilde{\Psi}^{\sigma}(u_i,v_j).
\end{align*}

The next lemma bounds the relation between spectral radiuses of $\Psi^{\sigma}_{\mu_k}, \Psi^{\r:\sigma}_{\nu}$ and $\widetilde{\Psi}^{\sigma}, \widehat{\Psi}^{\r:\sigma}$.
\begin{lemma}
\label{lemma-relation-hat-tilde}
It holds that 	$\rho\tp{\Psi^{\sigma}_{\mu_k}} \leq \rho\tp{\widetilde{\Psi}^{\sigma}}$ and $\rho\tp{\widehat{\Psi}^{\r:\sigma}} \leq \rho\tp{\Psi^{\r:\sigma}_{\nu}} + 2$.
\end{lemma}

We first use \Cref{lemma-relation-hat-tilde} to prove \Cref{lemma:bounded-spectral-independence}, then prove \Cref{lemma-relation-hat-tilde}.

\begin{proof}[Proof of \Cref{lemma:bounded-spectral-independence} assuming \Cref{lemma-relation-hat-tilde}]
  Due to \Cref{lemma-relation-hat-tilde}, it suffices to show that 
\begin{align*}
 \rho\tp{\widetilde{\Psi}^{\sigma}} \leq 	\rho\tp{\widehat{\Psi}^{\r:\sigma}}.
\end{align*}
We will prove that for any eigenvalue $\lambda_{\mathrm{eig}} \in \mathbb{C}$ of $\widetilde{\Psi}^{\sigma}$, if $\lambda_{\mathrm{eig}} \neq 0$, then $\lambda_{\mathrm{eig}} $ is also an eigenvalue of $\widehat{\Psi}^{\r:\sigma}$.
Hence, $\rho\tp{\widetilde{\Psi}^{\sigma}} \leq 	\rho\tp{\widehat{\Psi}^{\r:\sigma}}$, this proves the lemma.

Fix an eigenvalue $\lambda_{\mathrm{eig}}  \neq 0$ of $\widetilde{\Psi}^{\sigma}$. Let $f: V_k \setminus \Lambda \to \mathbb{C}$ denote the corresponding eigenvector.
For any variable $v \in V$, any $i,j \in [m_v]$, by \Cref{observation-block-matrix}, it holds that $\widetilde{\Psi}^{\sigma}(v_i,\cdot) = \widetilde{\Psi}^{\sigma}(v_j,\cdot)$. We have
\begin{align*}
\lambda_{\mathrm{eig}}f (v_i) = ( \widetilde{\Psi}^{\sigma} f)(v_i) = ( \widetilde{\Psi}^{\sigma} f)(v_j) = \lambda_{\mathrm{eig}}f(v_j).
\end{align*}
Since $\lambda_{\mathrm{eig}} \neq 0$, we have $f(v_i) = f(v_j)$.

Next, we define a vector $g: V\setminus \r:\Lambda_\sigma \to \mathbb{C}$ according to $f$.
For any $v \in V\setminus \r:\Lambda_\sigma$, by~\eqref{eq-lambda-three}, $m_v \geq 1$. Define
\begin{align*}
g(v) = f(v_1).	
\end{align*}
It holds that 
\begin{align}
\label{eq-rela-f-g}
\forall v \in V \setminus \r:\Lambda_\sigma , i \in [m_v], \quad g(v)  = f(v_i).	
\end{align}
We show that if $f$ is an eigenvector of $\widetilde{\Psi}^{\sigma}$ with eigenvalue $\lambda_{\mathrm{eig}} \neq 0$, then  $g$ is not a zero-vector. This fact can be verified according to the following two cases.
\begin{itemize}
\item If there exists $v \notin \r:\Lambda_\sigma$ such that $f(v_1) \neq 0$, then $g$ is not a zero-vector.
\item Otherwise, for all $v \in V$ and $i \in [m_v]$, if $f(v_i) \neq 0$, then $v \in \r:\Lambda_\sigma$. For any $v \in \r:\Lambda_\sigma$ and $i \in [m_v]$, by~\eqref{eq-def-tilde}, it is straightforward to verify that $\widetilde{\Psi}^{\sigma}(w,v_i) = 0$ for all $w \in V_k \setminus \Lambda$, thus
 it holds that $\widetilde{\Psi}^{\sigma} f = \*0$. Since  $\lambda_{\mathrm{eig}} \neq 0$, this case cannot occur.
\end{itemize}

We show that $\lambda_{\mathrm{eig}}$ is also an eigenvalue of $\widehat{\Psi}^{\r:\sigma}$ with eigenvector $g$. We have:
\begin{align}
\label{eq-same-eig-value}
\forall u \in V\setminus \r:\Lambda_\sigma,\quad
\tp{\widehat{\Psi}^{\r:\sigma} g} (u)
    &= \sum_{v \in V\setminus \r:\Lambda_\sigma} \widehat{\Psi}^{\r:\sigma}(u,v)g(v)\notag\\
    &= \sum_{v \in V\setminus \r:\Lambda_\sigma} \sum_{j \in [m_v]} \widetilde{\Psi}^{\sigma}(u_1,v_j)g(v)
    && \text{(by~\eqref{eq-def-hat})}\notag\\
    & = \sum_{v \in V\setminus \r:\Lambda_\sigma} \sum_{j \in [m_v]} \widetilde{\Psi}^{\sigma}(u_1,v_j)f(v_j).
    &&\text{(by~\eqref{eq-rela-f-g})}
\end{align}
We claim that the following equation holds
\begin{align}
\label{eq-claim-zeros}
\sum_{v \in V\setminus \r:\Lambda_\sigma} \sum_{j \in [m_v]} \widetilde{\Psi}^{\sigma}(u_1,v_j)f(v_j)	= \sum_{v \in V} \sum_{j \in [m_v]}\widetilde{\Psi}^{\sigma}(u_1,v_j)f(v_j).
\end{align}
To verify~\eqref{eq-claim-zeros}, we consider a $v \in \r:\Lambda_\sigma$ with $m_v \geq 1$. 
Since $u \in V \setminus \r:\Lambda_\sigma$, it must hold that $u \neq v$.
Since $v \in \r:\Lambda_\sigma$, the value of $v$ is fixed by $\r:\sigma$, thus $|\Omega(\nu^{\r:\sigma}_v)| = 1$.
By definition in~\eqref{eq-def-gen-inf}, we have
\begin{align*}
\forall v \in \r:\Lambda_\sigma, j \in [m_v],\quad \widetilde{\Psi}^{\sigma}(u_1,v_j) =  \frac{1}{m_v} \INF{\nu}{\r:\sigma}{u}{v} = 0.
\end{align*}
Hence
\begin{align*}
\sum_{v \in V} \sum_{j \in [m_v]}\widetilde{\Psi}^{\sigma}(u_1,v_j)f(v_j) &= 	\sum_{v \in V\setminus \r:\Lambda_\sigma} \sum_{j \in [m_v]} \widetilde{\Psi}^{\sigma}(u_1,v_j)f(v_j) + \sum_{v \in \r:\Lambda_\sigma} \sum_{j \in [m_v]} \widetilde{\Psi}^{\sigma}(u_1,v_j)f(v_j)\\
 &= \sum_{v \in V\setminus \r:\Lambda_\sigma} \sum_{j \in [m_v]} \widetilde{\Psi}^{\sigma}(u_1,v_j)f(v_j)	.
\end{align*}
%
%
%
This proves~\eqref{eq-claim-zeros}.
Combining~\eqref{eq-claim-zeros} and~\eqref{eq-same-eig-value}, we have
\begin{align*}
\forall u \in V\setminus \r:\Lambda_\sigma,\quad
\tp{\widehat{\Psi}^{\r:\sigma} g} (u)
    &= 	\sum_{v \in V} \sum_{j \in [m_v]}\widetilde{\Psi}^{\sigma}(u_1,v_j)f(v_j) = \tp{\widetilde{\Psi}^{\sigma} f}(u_1) = \lambda_{\mathrm{eig}}f(u_1) =  \lambda_{\mathrm{eig}} g(u).
\end{align*}
Since $g$ is a non-zero vector,  $\lambda_{\mathrm{eig}}$ is also an eigenvalue of $\widehat{\Psi}^{\r:\sigma}$.
\end{proof}

\subsection{Spectral radius bounds for the intermediate matrices (proof of \Cref{lemma-relation-hat-tilde})}
Let $A,B \in \mathbb{R}^{n \times n}$, if for all $i,j \in [n]$, $A(i,j) \leq B(i,j)$, then we denote
$A \leq B$.  We need the following proposition in linear algebra.
\begin{proposition}[\text{\cite[Corollary 8.1.19]{horn2012matrix}}] \label{lemma:nonnegative-matrix-spectral-ratio-comparation} 
  Let $A, B \in \mathbb{R}^{n\times n}_{\geq 0}$ be two nonnegative matrices.
  If $A \leq B$, then it holds that $\rho(A) \leq \rho(B)$.
\end{proposition}

We first prove that 
\begin{align}
\label{eq-radius-proof-1}
\rho\tp{\widehat{\Psi}^{\r:\sigma}} \leq \rho\tp{\Psi^{\r:\sigma}_{\nu}} + 2	
\end{align}
By the definition of $\widehat{\Psi}^{\r:\sigma}$ in~\eqref{eq-def-hat},  for any $u,v \in V \setminus \r:\Lambda_\sigma$ satisfying $u \neq v$, it holds that
\begin{align*}
\widehat{\Psi}^{\r:\sigma}(u,v) &= \sum_{j \in [m_v]} \widetilde{\Psi}^{\sigma}(u_1,v_j)	= \sum_{j \in [m_v]} \frac{\INF{\nu}{\r:\sigma}{u}{v}}{m_v} =\INF{\nu}{\r:\sigma}{u}{v} = \Psi_{\nu}^{\r:\sigma}(u,v),
\end{align*}
where the last equation holds due to \Cref{definition-weight-tot-inf} and the fact $u \neq v$. 
For any $v \in V \setminus \r:\Lambda_\sigma$,
\begin{align*}
\widehat{\Psi}^{\r:\sigma}(v,v) &= \sum_{j \in [m_v]} \widetilde{\Psi}^{\sigma}(v_1,v_j)	= \sum_{j \in [m_v]} \frac{2}{m_v}\leq 2.	
\end{align*}
Combining above inequalities together, we have
\begin{align*}
\widehat{\Psi}^{\r:\sigma} \leq 	\Psi_{\nu}^{\r:\sigma} + 2 I,
\end{align*}
where $I$ is the identity matrix. By \Cref{lemma:nonnegative-matrix-spectral-ratio-comparation}, we have
\begin{align*}
\rho \tp{	\widehat{\Psi}^{\r:\sigma}} \leq \rho \tp{	\Psi_{\nu}^{\r:\sigma} + 2 I} \leq \rho \tp{\Psi_{\nu}^{\r:\sigma} } + 2.
\end{align*}
This proves~\eqref{eq-radius-proof-1}.

We next prove that
\begin{align}
\label{eq-radius-proof-2}
\rho\tp{\Psi^{\sigma}_{\mu_k}} \leq \rho\tp{\widetilde{\Psi}^{\sigma}}.
\end{align}
By \Cref{lemma:nonnegative-matrix-spectral-ratio-comparation}, it suffices to show that $\Psi^{\sigma}_{\mu_k} \leq \widetilde{\Psi}^{\sigma}$.
By \Cref{definition-weight-tot-inf}, for all  $w \in V_k \setminus \Lambda$, it holds that
$\Psi^{\sigma}_{\mu_k}(w,w) = 0$, 
and for all $w, w' \in V_k \setminus \Lambda$ with $w \neq w'$, it holds that $\Psi^{\sigma}_{\mu_k}(w,w') = \INF{\mu_k}{\sigma}{w}{w'}$.
By the definition of $\widetilde{\Psi}^{\sigma}$ in~\eqref{eq-def-tilde}, to prove \Cref{lemma-relation-hat-tilde}, it only remains to verify the following facts:
\begin{align}
&\forall u,v \in V \text{ with } v \in \r:\Lambda_\sigma, \forall i \in [m_u], j \in [m_v]: 
&&\Psi_{\mu_k}^\sigma(u_i,v_j) =\widetilde{\Psi}^{\sigma}(u_i,v_j)=0;\label{eq-target-trivial}\\
&\forall u,v \in V \text{ with } u \neq v , \forall i \in [m_u], j \in [m_v]: 
&&\Psi_{\mu_k}^\sigma(u_i,v_j) =  \INF{\mu_k}{\sigma}{u_i}{v_j} \leq \frac{1}{m_v}\INF{\nu}{\r:\sigma}{u}{v} = \widetilde{\Psi}^{\sigma}(u_i,v_j)\label{eq-target-1};\\
&\forall u \in V, i,j \in [m_u] \text{ with } v \notin \r:\Lambda_\sigma \text{ and } i \neq j: 
&&\Psi_{\mu_k}^\sigma(u_i,u_j)  = \INF{\mu_k}{\sigma}{u_i}{u_j} \leq \frac{2}{m_u} = \widetilde{\Psi}^{\sigma}(u_i,u_j).\label{eq-target-2}
\end{align}
We first verify~\eqref{eq-target-trivial}, which holds trivially because $|\Omega(\mu_{k,v_j}^\sigma)| =|\Omega(\nu_{v}^{\r:\sigma})| = 1$.

The rest of this section is dedicated to verifying the facts given in~\eqref{eq-target-1} and~\eqref{eq-target-2}.
Without loss of generality, when proving~\eqref{eq-target-1} and~\eqref{eq-target-2}, we assume that
\begin{align*}
 \Omega(\mu_{k,u_i}^\sigma) &= \{\0,\1\}	\text{ and } \Omega(\mu_{k,v_j}^\sigma) = \{\0,\1\},\\
  \Omega(\mu_{k,u_i}^\sigma) &= \{\0,\1\}	\text{ and }  \Omega(\mu_{k,u_j}^\sigma) = \{\0,\1\}.
\end{align*}
Otherwise, if $\Omega(\mu_{k,u_i}^\sigma) \neq  \{\0,\1\}$ or $ \Omega(\mu_{k,v_j}^\sigma) \neq \{\0,\1\} $, then it must hold that $|\Omega(\mu_{k,u_i}^\sigma)| = 1$ or $|\Omega(\mu_{k,v_j}^\sigma)| = 1$, thus $\INF{\mu_k}{\sigma}{u_i}{v_j} =0$ and~\eqref{eq-target-1} holds trivially;
if $\Omega(\mu_{k,u_i}^\sigma) \neq \{\0,\1\}$ or $\Omega(\mu_{k,u_j}^\sigma) \neq \{\0,\1\}$, then it must hold that $|\Omega(\mu_{k,u_i}^\sigma) | = 1$ or $|\Omega(\mu_{k,u_j}^\sigma)| = 1$, thus $\INF{\mu_k}{\sigma}{u_i}{u_j} =0$ and~\eqref{eq-target-2} holds trivially.
%
By our assumptions, we can conclude that 
\begin{align*}
u \notin \r:\Lambda_\sigma \text{ and } v \notin \r:\Lambda_\sigma.
\end{align*}
Because if $u \in \r:\Lambda_\sigma$, then it holds that $|\Omega(\mu_{k,u_i}^\sigma)| = 1$; and if $v \in \r:\Lambda_\sigma$, then it holds that $|\Omega(\mu_{k,v_j}^\sigma)| = 1$.

We first prove ~\eqref{eq-target-1}.
Note that $\INF{\mu_k}{\sigma}{u_i}{v_j} = \DTV{\mu_{k,v_j}^{\sigma,u_i\gets \1} }{\mu_{k,v_j}^{\sigma,u_i\gets \0}}$. By definition, we have
\begin{align*}
\mu_{k,v_j}^{\sigma,u_i\gets \1}(\1) = \frac{\Pr[\*Y \sim \mu_k]{Y_{v_j} = \1 \land Y_{u_i} = \1 \land Y_\Lambda = \sigma }  }{\Pr[\*Y \sim \mu_k]{ Y_{u_i} = \1 \land Y_\Lambda = \sigma }  }.
\end{align*}
Recall $\r:\Lambda_{\sigma,\1}$ is defined in~\eqref{eq-lambda-three}.
Let $N = | \r:\Lambda_{\sigma,\1}|$.
By the definition of distribution $\mu_k$, we have
\begin{align}
\label{eq-conditional-1}
&\,\Pr[\*Y \sim \mu_k]{Y_{v_j} = \1 \land Y_{u_i} = \1 \land Y_\Lambda = \sigma }\notag\\
=&\, \sum_{\*X  \in \Omega(\mu)}\mu(\*X)\one{X_{\r:\Lambda_\sigma} = \r:\sigma }\tp{\frac{1}{k}}^N \tp{\prod_{w \in V \setminus \r:\Lambda_\sigma:X_w = \1 }\frac{m_w}{k} } \tp{\frac{\one{X_v = \1 }}{m_v}} \tp{\frac{\one{X_u = \1 }}{m_u}}.
\end{align}
Recall that to draw a random sample $\*Y \sim \mu_k$, one needs to draw a random sample $\*X \sim \mu$, then transform $\*X$ to $\*Y$ according to the rules in \Cref{definition-k-transformation}.
Equation~\eqref{eq-conditional-1} holds due to the following arguments:
\begin{itemize}
\item
If $Y_\Lambda = \sigma $, by definitions in~\eqref{eq-lambda-three} and~\eqref{eq-def-sigma-ast}, it must have $X_{\r:\Lambda_\sigma} = \r:\sigma $ and for every $w \in \r:\Lambda_{\sigma,\1}$, a particular variable $w_\ell \in C_w$ (determined by $\sigma$) is picked and is assigned with value $\1$. 
The probability of such event is $(\frac{1}{k})^N$. 
Furthermore, for every $w \in V \setminus \r:\Lambda_\sigma$, $\sigma$ requires that $k - m_w$ variables in $C_w$ must take value $\0$. 
If $X_w = \1$, a variable among the remaining $m_w$ variables should be picked and assigned with value $\1$.
Such event occurs with probability $\prod_{w \in V \setminus \r:\Lambda_\sigma:X_w = \1 }\frac{m_w}{k}$.
\item
If $Y_{v_j} = \1$, then $X_v = \1$, and $v_j$ must be picked and is assigned with  value $\1$. Since $v \notin \r:\Lambda_\sigma$, conditional on previous events, this event holds with probability $\frac{1}{m_v}$.
\item  If $Y_{u_i} = \1$, then $X_u = \1$, and $u_i$ must be picked and is assigned with value $\1$. Since $u \notin \r:\Lambda_\sigma$, conditional on previous events, this event holds with probability $\frac{1}{m_u}$.  
\end{itemize}
Similarly, we have
\begin{align*}
\Pr[\*Y \sim \mu_k]{ Y_{u_i} = \1 \land Y_\Lambda = \sigma } = \sum_{\*X  \in \Omega(\mu)}\mu(\*X)\one{X_{\r:\Lambda_\sigma} = \r:\sigma }\tp{\frac{1}{k}}^N \tp{\prod_{w \in V \setminus \r:\Lambda_\sigma:X_w = \1 }\frac{m_w}{k} }  \tp{\frac{\one{X_u = \1 }}{m_u}} .	
\end{align*}
Recall $\nu=\mu^{(\*\lambda_\sigma)}$, where the local fields $\*\lambda_\sigma$ are constructed in~\eqref{eq-def-lambda-sigma}. We have
\begin{align}
\label{eq-muk-to-pi}
\mu_{k,v_j}^{\sigma,u_i\gets \1}(\1) &= \frac{ \sum_{\*X  \in \Omega(\mu)}\nu(\*X)\one{X_{\r:\Lambda_\sigma} = \r:\sigma } \tp{\frac{\one{X_v = \1 }}{m_v}}\tp{ \frac{\one{X_u = \1 }}{m_u}}  }{ \sum_{\*X  \in \Omega(\mu)}\nu(\*X)\one{X_{\r:\Lambda_\sigma} = \r:\sigma } \tp{\frac{\one{X_u = \1 }}{m_u}}}\notag\\
&=\frac{1}{m_v} \frac{\Pr[\*X \sim \nu]{X_{\r:\Lambda_\sigma} = \r:\sigma  \land X_v = \1 \land X_u =\1 } }{\Pr[\*X \sim \nu]{X_{\r:\Lambda_\sigma} = \r:\sigma \land X_u =\1 }} \notag\\
&= \frac{\nu^{\r:\sigma, u \gets \1}_v(\1)}{m_v}.
\end{align}
Note that since $\1 \in \Omega(\mu_{k,u_i}^\sigma)$, we have $\1 \in \Omega(\nu^{\r:\sigma}_u)$, and thus $\nu^{\r:\sigma, u \gets \1}_v(\1)$ is well defined. 

Next, we calculate 
\begin{align*}
\mu_{k,v_j}^{\sigma,u_i\gets \0}(\1) = \frac{\Pr[\*Y \sim \mu_k]{Y_{v_j} = \1 \land Y_{u_i} = \0 \land Y_\Lambda = \sigma }  }{\Pr[\*Y \sim \mu_k]{ Y_{u_i} = \0 \land Y_\Lambda = \sigma }  }.
\end{align*}
By the definition of distribution $\mu_k$, we have
\begin{align*}
&\,\Pr[\*Y \sim \mu_k]{Y_{v_j} = \1 \land Y_{u_i} = \0 \land Y_\Lambda = \sigma }\\
=&\, \sum_{\*X  \in \Omega(\mu)}\mu(\*X)\one{X_{\r:\Lambda_\sigma} = \r:\sigma } \tp{\frac{1}{k}}^N\tp{\prod_{\substack{w \in V \setminus \r:\Lambda_\sigma \\ X_w = \1 }}\frac{m_w}{k} } \tp{\frac{\one{X_v = \1 }}{m_v}}\tp{ \one{X_u = \1 }\frac{m_u - 1}{m_u} + \one{X_u = \0 } } .
\end{align*}
Compared with~\eqref{eq-conditional-1}, the only difference is calculating the probability of event $Y_{u_i} = \0$ conditional on other events. Note that $u \notin \r:\Lambda_\sigma$. If $X_u = \0$, then it must hold that $Y_{u_i}= \0$; and if otherwise $X_u = \1$,  a variable among all $m_u$ unfixed variables in $C_u$ is picked, and such variable cannot be $u_i$, which occurs with probability $\frac{m_u-1}{m_u}$.
Similarly, we have
\begin{align*}
&\,\Pr[\*Y \sim \mu_k]{ Y_{u_i} = \0 \land Y_\Lambda = \sigma }\\
=&\, \sum_{\*X  \in \Omega(\mu)}\mu(\*X)\one{X_{\r:\Lambda_\sigma} = \r:\sigma } \tp{\frac{1}{k}}^N\tp{\prod_{w \in V \setminus \r:\Lambda_\sigma:X_w = \1 }\frac{m_w}{k} } \tp{ \one{X_u = \1 }\frac{m_u - 1}{m_u} + \one{X_u = \0 } } .
\end{align*}
We define a vector of local fields $\*\lambda_\sigma'$ as
\begin{align*}
\forall w \in V,\quad \lambda_\sigma'(w) = \begin{cases}
\frac{m_u-1}{k} &\text{if } w = u,\\
 \frac{m_w}{k} &\text{if } w \in V \setminus (\{u\}\cup  \r:\Lambda_\sigma),\\
 1 &\text{if } w \in \r:\Lambda_\sigma.
 \end{cases}
\end{align*}
Note that $\*\lambda_\sigma'$ differs from $\*\lambda_\sigma$ only at $u$, where $\lambda_\sigma'(u) = \frac{m_u-1}{k}$ and $\lambda_\sigma(u) = \frac{m_u}{k}$.
Since $u,v \in V \setminus \r:\Lambda_\sigma$, we have
\begin{align*}
\Pr[\*Y \sim \mu_k]{Y_{v_j} = \1 \land Y_{u_i} = \0 \land Y_\Lambda = \sigma }&= \sum_{\*X  \in \Omega(\mu)}\frac{\mu(\*X)\one{X_{\r:\Lambda_\sigma} = \r:\sigma }}{k^N} \tp{\prod_{\substack{w \in V \setminus \r:\Lambda_\sigma \\ X_w = \1 }}\lambda_\sigma'(w)} \tp{\frac{\one{X_v = \1 }}{m_v}},
\end{align*}
and
\begin{align*}
\Pr[\*Y \sim \mu_k]{ Y_{u_i} = \0 \land Y_\Lambda = \sigma }&=\sum_{\*X  \in \Omega(\mu)}\frac{\mu(\*X)\one{X_{\r:\Lambda_\sigma} = \r:\sigma }}{k^N} \tp{\prod_{\substack{w \in V \setminus \r:\Lambda_\sigma \\ X_w = \1 }}\lambda_\sigma'(w)}.
\end{align*}
We define a new distribution $\pi$ by imposing the local fields $\*\lambda'_\sigma$ to $\mu$:
\begin{align*}
\pi = \mu^{(\*\lambda_\sigma')}.
\end{align*}
If $m_u > 1$, then $\*\lambda_\sigma'$ is a positive vector, the distribution $\pi$ is well-defined. 
If $m_u = 1$, then $\lambda_\sigma'(u) = 0$ and $\lambda_\sigma'(w) > 0$ for all $w \neq u$. 
In this case, variable $u$  can only take value $\0$, and  the  distribution $\pi$ is well-defined as long as $\mu_u(\0) > 0$. 
Note that $\Omega(\mu_{k,u_i}^\sigma) = \{\0,\1\}$.
Conditional on $\sigma$, $u_i$ takes value $\0$ with positive probability in distribution $\mu_k$.
Since $m_u = 1$, with a positive probability, all variables in $C_u$ take value $\0$.
By the definition of distribution $\mu_k$, we have
\begin{align}
\label{eq-mu=1}
m_u = 1  \quad \implies \quad \mu^{\r:\sigma}_u(\0) > 0.
\end{align}
Hence $\mu_u(\0) > 0$ and the distribution $\pi$ is always well-defined.

We have
\begin{align}
\label{eq-muk-to-nu}
\mu_{k,v_j}^{\sigma,u_i\gets \0}(\1) &=	\frac{\sum_{\*X  \in \Omega(\mu)}\mu(\*X)\one{X_{\r:\Lambda_\sigma} = \r:\sigma }\tp{\frac{\one{X_v = \1 }}{m_v}}\tp{\prod_{w \in V \setminus \r:\Lambda_\sigma:X_w = \1 }\lambda'_\sigma(w) }  }{\sum_{\*X  \in \Omega(\mu)}\mu(\*X)\one{X_{\r:\Lambda_\sigma} = \r:\sigma }\tp{\prod_{w \in V \setminus \r:\Lambda_\sigma:X_w = \1 }\lambda'_\sigma(w) } }\notag\\
&= \frac{\sum_{\*X  \in \Omega(\mu)}\pi(\*X)\one{X_{\r:\Lambda_\sigma} = \r:\sigma }\tp{\frac{\one{X_v = \1 }}{m_v}} }{\sum_{\*X  \in \Omega(\mu)}\pi(\*X)\one{X_{\r:\Lambda_\sigma} = \r:\sigma } }\notag\\
&= \frac{1}{m_v}\frac{\Pr[\*X\sim \pi]{X_{\r:\Lambda_\sigma} = \r:\sigma \land X_v = \1} }{\Pr[\*X \sim \pi]{X_{\r:\Lambda_\sigma} = \r:\sigma}}\notag\\
&= \frac{\pi^{\r:\sigma}_v(\1)}{m_v}.
\end{align}

Combining~\eqref{eq-muk-to-pi} and~\eqref{eq-muk-to-nu}, we have 
\begin{align}
\label{eq-25-1}
\INF{\mu_k}{\sigma}{u_i}{v_j} 
&= \DTV{\mu_{k,v_j}^{\sigma,u_i\gets \1} }{\mu_{v_j}^{\sigma,u_i\gets \0}} \notag\\
&= \abs{\mu_{k,v_j}^{\sigma,u_i\gets \1}(\1) - \mu_{v_j}^{\sigma,u_i\gets \0}(\1)}\notag\\
&= \frac{1}{m_v}\abs{\nu^{\r:\sigma, u \gets \1}_v(\1)-\pi^{\r:\sigma}_v(\1)} \notag\\
&= \frac{1}{m_v}\DTV{\nu^{\r:\sigma, u \gets \1}_v}{\pi^{\r:\sigma}_v}.
\end{align}

Finally, we bound $\DTV{\nu^{\r:\sigma, u \gets \1}_v}{\pi^{\r:\sigma}_v}$. We construct the following coupling between $\nu^{\r:\sigma, u \gets \1}_v$ and $\pi^{\r:\sigma}_v$:
\begin{itemize}
\item sample a random value $c \in \{\0,\1\}$ according to the distribution $\pi^{\r:\sigma}_u$;
\item sample $c_v,c_v'$ jointly according to the optimal coupling between $\nu^{\r:\sigma, u \gets \1}_v$ and $\pi^{\r:\sigma, u \gets c}_v$.	
\end{itemize}
Recall that $\nu = \mu^{(\*\lambda_\sigma)}$ and $\pi = \mu^{(\*\lambda_\sigma')}$.
We claim that 
\begin{align}
\label{eq-claim-feasible-solution}
\1 \in \Omega(\nu_u^{\r:\sigma}) \quad\text{and}\quad 	\Omega(\pi_u^{\r:\sigma}) \subseteq \Omega(\nu_u^{\r:\sigma}).
\end{align}

Assume that \eqref{eq-claim-feasible-solution} is correct. 
Recall that $\*\lambda_\sigma'$ differs from $\*\lambda_\sigma$ only at $u$, where $\lambda_\sigma'(u) = \frac{m_u-1}{k}$ and $\lambda_\sigma(u) = \frac{m_u}{k}$. 
It is straightforward to verify that 
\begin{align*}
\forall c \in \Omega(\pi_u^{\r:\sigma}), \quad \pi^{\r:\sigma, u \gets c}_v = \nu^{\r:\sigma, u \gets c}_v.
\end{align*}
By the coupling inequality, we have 
\begin{align}
\label{eq-25-2}
\DTV{\nu^{\r:\sigma, u \gets \1}_v}{\pi^{\r:\sigma}_v}
&\leq \Pr[]{c_v \neq c'_v} \notag\\
&= \sum_{c \in \Omega(\pi_u^{\r:\sigma})}	\pi_u^{\r:\sigma}(c)\DTV{\nu^{\r:\sigma, u \gets \1}_v}{\pi^{\r:\sigma, u \gets c}_v}\notag\\
&= \sum_{c \in \Omega(\pi_u^{\r:\sigma})}	\pi_u^{\r:\sigma}(c)\DTV{\nu^{\r:\sigma, u \gets \1}_v}{\nu^{\r:\sigma, u \gets c}_v}\notag\\
&\leq \max_{i,j \in \Omega(\nu^{\r:\sigma}_u) }\DTV{\nu^{\r:\sigma, u \gets i}_v}{\nu^{\r:\sigma, u \gets j}_v} \notag\\
&= \INF{\nu}{\r:\sigma}{u}{v}.
\end{align}
Combining~\eqref{eq-25-1} and~\eqref{eq-25-2} proves~\eqref{eq-target-1}. 

We now verify~\eqref{eq-claim-feasible-solution}. Since $\*\lambda_\sigma$ is a positive vector, then $\nu$ and $\mu$ have the same support. 
Note that we assume $\Omega(\mu^\sigma_{k,u_i}) = \{\0,\1\} $. 
By the definition of $\mu_k$ and $\r:\sigma$, it holds that  $\1 \in \Omega(\mu_u^{\r:\sigma})$, hence $\1 \in \Omega(\nu_u^{\r:\sigma})$.
We now show that $\Omega(\pi_u^{\r:\sigma}) \subseteq \Omega(\nu_u^{\r:\sigma})$. Note that $\lambda_\sigma'(u) = \frac{m_u-1}{k}$ and $\lambda_\sigma(u) = \frac{m_u}{k}$. There are two cases:
\begin{itemize}
\item If $m_u > 1$, then $\*\lambda_\sigma$ and $\*\lambda_\sigma'$ are both positive, $\nu$ and $\pi$ have the same support. Thus, $\Omega(\pi_u^{\r:\sigma}) = \Omega(\nu_u^{\r:\sigma})$.
\item If $m_u = 1$, in this case $\Omega(\pi_u^{\r:\sigma})=\{\0\}$, by~\eqref{eq-mu=1}, $\0 \in \Omega(\mu^{\r:\sigma}_u)$. Since $\*\lambda_\sigma$ is a positive vector, $\mu$ and $\nu$ have the same support, thus $\0 \in \Omega(\nu^{\r:\sigma}_u)$.
\end{itemize}

Next, we prove~\eqref{eq-target-2}. Recall that without loss of generality $\Omega(\mu_{k,u_i}^\sigma) = \{\0,\1\}$. We have 
\begin{align*}
\INF{\mu_k}{\sigma}{u_i}{u_j} = \DTV{\mu_{k,u_j}^{\sigma,u_i\gets \1} }{\mu_{k,u_j}^{\sigma,u_i\gets \0}}	= \mu_{k,u_j}^{\sigma,u_i\gets \0}(\1),
\end{align*}
which is due to that $\mu_{k,u_j}^{\sigma,u_i\gets \1}(\1) = 0$.
By definition,
\begin{align*}
\mu_{k,u_j}^{\sigma,u_i\gets \0}(\1) = \frac{ \Pr[\*Y \sim \mu_k]{Y_{u_j} = \1 \land Y_{u_i} = \0 \land Y_{\Lambda} = \sigma } }{\Pr[\*Y \sim \mu_k]{Y_{u_i} = \0 \land Y_{\Lambda} = \sigma } }.	
\end{align*}
Recall $N = | \r:\Lambda_{\sigma,\1}|$.
By the definition of distribution $\mu_k$, we have
\begin{align*}
&\Pr[\*Y \sim \mu_k]{Y_{u_j} = \1 \land Y_{u_i} = \0 \land Y_{\Lambda} = \sigma } \\
=&\, \sum_{\*X  \in \Omega(\mu)}\mu(\*X)\one{X_{\r:\Lambda_\sigma} = \r:\sigma }\tp{\frac{1}{k}}^N \tp{\prod_{w \in V \setminus \r:\Lambda_\sigma:X_w = \1 }\frac{m_w}{k} } \frac{\one{X_u = \1 }}{m_u},
\end{align*}
Compared with~\eqref{eq-conditional-1}, the only difference is calculating the probability of event $Y_{u_i} = \0 \land Y_{u_j} = \1$ conditional on other events. Note that $u \notin \r:\Lambda_\sigma$. 
If $Y_{u_i} = \0 \land Y_{u_j} = \1$, then $X_u = \1$, and we pick variable $u_j$ to set its value to $\1$. Conditional on other events, this event occurs with probability $\frac{1}{m_u}$.

We assume $\Pr[\*Y \sim \mu_k]{Y_{u_j} = \1 \land Y_{u_i} = \0 \land Y_{\Lambda} = \sigma } > 0$, since if otherwise $\INF{\mu_k}{\sigma}{u_i}{u_j} = \mu_{k,u_j}^{\sigma,u_i\gets \0}(\1) = 0$ and \eqref{eq-target-2} holds trivially.
Similarly, 
\begin{align*}
&\,\Pr[\*Y \sim \mu_k]{Y_{u_i} = \0 \land Y_{\Lambda} = \sigma } \notag\\
=&\, \sum_{\*X  \in \Omega(\mu)}\mu(\*X)\one{X_{\r:\Lambda_\sigma} = \r:\sigma }\tp{\frac{1}{k}}^N \tp{\prod_{w \in V \setminus \r:\Lambda_\sigma:X_w = \1 }\frac{m_w}{k} }\tp{ \one{X_u = \1 }\frac{m_u - 1}{m_u} + \one{X_u = \0 } }\\
\geq &\, \sum_{\*X  \in \Omega(\mu)}\mu(\*X)\one{X_{\r:\Lambda_\sigma} = \r:\sigma }\tp{\frac{1}{k}}^N \tp{\prod_{w \in V \setminus \r:\Lambda_\sigma:X_w = \1 }\frac{m_w}{k} }\tp{ \one{X_u = \1 }\frac{m_u - 1}{m_u}}\\
=&\, (m_u - 1)\Pr[Y \sim \mu_k]{Y_{u_j} = \1 \land Y_{u_i} = \0 \land Y_{\Lambda} = \sigma }.
\end{align*}
In~\eqref{eq-target-2}, we assume $i \neq j$, which implies $m_u \geq 2$. We have
\begin{align*}
\INF{\mu_k}{\sigma}{u_i}{u_j} 
&= 	\frac{ \Pr[Y \sim \mu_k]{Y_{u_j} = \1 \land Y_{u_i} = \0 \land Y_{\Lambda} = \sigma } }{\Pr[Y \sim \mu_k]{Y_{u_i} = \0 \land Y_{\Lambda} = \sigma } } \\
&\leq \frac{1}{m_u-1} \\
&\leq \frac{2}{m_u},
\end{align*}
which proves~\eqref{eq-target-2}. Together with the above proof of~\eqref{eq-target-1}, this proves \Cref{lemma-relation-hat-tilde}, which concludes our proof of \Cref{lemma:block-dynamics-spectral-gap}.

\section{Spectral Gaps of Two-Spin Systems}\label{section-spectral-gap-2-spin}
In this section, we prove \Cref{lemma-good-direction} and its slightly strengthened variants for the hardcore and Ising models, which respectively imply \Cref{theorem-2pin-gap}, \Cref{theorem-hardcore} and \Cref{theorem-Ising}.


Let $\+I = (V, E, \lambda, \beta, \gamma)$ denote a two-spin system on an $n$-vertex graph $G = (V, E)$ with maximum degree $\Delta=\Delta_G\ge 3$ specified by parameters $(\beta,\gamma,\lambda)$.
Let $\mu$ denote the Gibbs distribution induced by $\+I$.
%
%
We assume that the system is anti-ferromagnetic, that is, $0\le \beta\le\gamma$, $\beta\gamma<1$,  and $\gamma,\lambda>0$.
Let $\delta\in(0,1)$ be a positive gap.
We assume that $(\beta,\gamma,\lambda)$ is up-to-$\Delta$ unique with gap $\delta$ (\Cref{definition:up-to-Delta-unique}).

Given any $\*\chi \in \{\0, \1\}^V$, recall the distribution $\nu = \mathsf{flip}(\mu,\*\chi)$ obtained by flipping $\mu$ according to $\*\chi$:
\begin{align*}
  \forall \sigma \in \{\0, \1\}^V, \quad \nu(\sigma) \triangleq \mu(\sigma \odot \*\chi),
\end{align*}	
where $\sigma \odot \*\chi \in \{\0,\1\}^V$ satisfying $(\sigma \odot \*\chi)_v = \sigma_v \chi_v$ for all $v \in V$.
Obviously, $\mu=\mathsf{flip}(\mathsf{flip}(\mu,\*\chi),\*\chi)$.

Given any $\*\chi \in \{-1,+1\}^V$,
for any vector $\*\theta \in (0,1]^V$ or scalar $\theta\in(0,1]$, we write:
\begin{align}
\label{eq-vector-decomp}
\ExtVec{\*\theta}{\*\chi}\triangleq (\theta_v^{\chi_v})_{v\in V}
\quad \text{ and }\quad
\ExtVec{\theta}{\*\chi}\triangleq (\theta^{\chi_v})_{v\in V}.
\end{align}


The following fact for flipping follows easily from the factorization of local fields.
\begin{fact}\label{observation:dist-equiv-0}
For any $\*\theta \in (0,1]^V$ and $\*\chi \in \{\0,\1\}^V$, $\Mag{\nu}{\*\theta}=\mathsf{flip}(\Mag{\mu}{\*\theta^{\*\chi}},\*\chi)$ where $\nu = \mathsf{flip}(\mu,\*\chi)$.
\end{fact}
\begin{proof}
For all $\*\theta \in (0,1]^V$ and $\sigma,\tau \in \Omega(\nu)$,
  \begin{align*}
    \frac{\nu^{(\*\theta)}(\sigma)}{\nu^{(\*\theta)}(\tau)} = \frac{\nu(\sigma) \prod_{u \in \sigma^{-1}(\1)} \theta_u}{\nu(\tau) \prod_{v \in \tau^{-1}(\1)}\theta_v}
    = \frac{\mu(\sigma \odot \*\chi) \prod_{u \in \sigma^{-1}(\1)} \theta_u}{\mu(\tau \odot \*\chi) \prod_{v \in \tau^{-1}(\1)}\theta_v}
    \overset{(\ast)}{=} \frac{\mu(\sigma \odot \*\chi) \prod_{u \in (\sigma \odot \*\chi)^{-1}(\1)} \theta_u^{\chi_u}}{\mu(\tau \odot \*\chi) \prod_{v \in (\tau \odot \*\chi)^{-1}(\1)}\theta_v^{\chi_v}}
    = \frac{\mu^{(\*\theta^{\*\chi})}(\sigma \odot \*\chi)}{\mu^{(\*\theta^{\*\chi})}(\tau \odot \*\chi)},
  \end{align*}
where the equation $(\ast)$ holds because 
\begin{align*}
\frac{\prod_{u \in (\sigma \odot \*\chi)^{-1}(\1)} \theta_u^{\chi_u}}{\prod_{v \in (\tau \odot \*\chi)^{-1}(\1)}\theta_v^{\chi_v}} &= \frac{\prod_{u \in \sigma^{-1}(\1) \cap \*\chi^{-1}(\1)}\theta_u \cdot  \prod_{v \in \sigma^{-1}(\0) \cap \*\chi^{-1}(\0)}\frac{1}{\theta_v} \cdot \prod_{w \in \*\chi^{-1}(\0)}\theta_w }{\prod_{u \in \tau^{-1}(\1) \cap \*\chi^{-1}(\1)}\theta_u \cdot \prod_{v \in \tau^{-1}(\0) \cap \*\chi^{-1}(\0)}\frac{1}{\theta_v} \cdot \prod_{w \in \*\chi^{-1}(\0)}\theta_w } \\
&=\frac{\prod_{u \in \sigma^{-1}(\1) \cap \*\chi^{-1}(\1)}\theta_u \cdot  \prod_{v \in \sigma^{-1}(\1) \cap \*\chi^{-1}(\0)}\theta_v }{\prod_{u \in \tau^{-1}(\1) \cap \*\chi^{-1}(\1)}\theta_u \cdot \prod_{v \in \tau^{-1}(\1) \cap \*\chi^{-1}(\0)}\theta_v } 	= \frac{\prod_{u \in \sigma^{-1}(\1)} \theta_u}{\prod_{v \in \tau^{-1}(\1)}\theta_v}.
\end{align*}
The above equation implies $\nu^{(\*\theta)}(\sigma) = \mu^{(\*\theta^{\*\chi})}(\sigma \odot \*\chi)$, i.e.~$\Mag{\nu}{\*\theta}=\mathsf{flip}(\Mag{\mu}{\*\theta^{\*\chi}},\*\chi)$.
\end{proof}

We generalize the notion of complete spectral independence (\Cref{definition-complete-SI}) to different directions. 
\begin{definition}[complete SI in a direction]
Let 
$\*\chi \in \{\0, \1\}^V$.
A distribution $\mu$ over $\{\0,\1\}^V$ is said to be \emph{completely $\eta$-spectrally independent in direction $\*\chi$} if $\Mag{\mu}{\*\theta^{\*\chi}}$ is $\eta$-spectrally independent for all $\*\theta \in (0,1]^V$.
\end{definition}


The complete $\eta$-spectral independence defined in \Cref{definition-complete-SI}
is the special case of the above definition with direction $\*\chi=\*1$, where $\*1=(\1)_{v\in V}$ denotes the all-$(\1)$ vector.
The following observation follows from \Cref{observation:dist-equiv-0} and the fact that flipping defines isomorphisms for measures and dynamics.
\begin{observation}\label{observation:dist-equiv}
Let $\nu = \mathsf{flip}(\mu,\*\chi)$ for distribution $\mu$ over $\{\0,\1\}^V$ and $\*\chi \in \{\0,\1\}^V$. 
\begin{itemize}
\item $\nu$ is completely $\eta$-spectrally independent if and only if $\mu$ is completely $\eta$-spectrally independent in direction $\*\chi$;
\item For any $\*\theta \in (0,1]^V$, it holds that 
$\spgap{\sgap}{\GD}(\Mag{\mu}{\*\theta^{\*\chi}}) = \spgap{\sgap}{\GD}( \Mag{\nu}{\*\theta})$.
\end{itemize}
\end{observation}

The good direction $\*\chi\in\{\0,\1\}^V$ is defined as follows.
\begin{align} \label{eq:direction}
\forall v\in V,\quad
  \chi_v = \mathrm{sgn}(\Delta_v, \lambda) &\triangleq \begin{cases}
    \1 & \lambda \leq \tp{\frac{\gamma}{\beta}}^{\Delta_v / 2},\\
    \0 & \text{otherwise,}
  \end{cases}
\end{align}
where $\Delta_v$ denotes the degree of $v$ in $G$ and we assume that $\frac{\gamma}{0} = +\infty$ for $\gamma > 0$.

With~\Cref{observation:dist-equiv}, \Cref{lemma-good-direction} can be equivalently stated as follows.

\begin{lemma}\label{lm:two-spin-property}
For all $\delta\in(0,1)$, 
for every anti-ferromagnetic two-spin system on an $n$-vertex graph $G=(V,E)$ with maximum degree $\Delta=\Delta_G\ge 3$ 
that is up-to-$\Delta$ unique with gap $\delta$, 
for  the $\*\chi \in \{\0, \1\}^V$ in~\eqref{eq:direction},
  \begin{enumerate}
    \item the Gibbs distribution $\mu$ is completely $\frac{144}{\delta}$-spectrally independent in direction $\*\chi$;
    \item for $\theta=\frac{\delta^2}{64}$ and $C=\frac{8}{\delta}$, it holds that $\spgap{\sgap}{\GD}\tp{\mu^{\tp{\theta^{\*\chi}}}} \ge \frac{1}{Cn}$.
  \end{enumerate}
\end{lemma}

The rest  of \Cref{section-spectral-gap-2-spin} is organized as folows:
In \Cref{sec:req-1} we prove the first part of \Cref{lm:two-spin-property}, and in \Cref{sec:req-2} we prove the second part of \Cref{lm:two-spin-property}.

\subsection{Complete spectral independence up to uniqueness (in the right direction)}\label{sec:req-1}
Let $\+I = (V, E, \beta, \gamma,\lambda)$ denote the anti-ferromagnetic two-spin system specified by $(\beta,\gamma,\lambda)$ on a $n$-vertex graph $G=(V,E)$ with maximum degree $\Delta=\Delta_G\ge 3$.
%
%
Let $\*\theta\in(0,1]^V$ be arbitrary and $\*\chi\in\{\0,\1\}^V$ as defined in \eqref{eq:direction}.
Let $\mu$ denote the Gibbs distribution of $\+I$ and $\pi=\Mag{\mu}{\*\theta^{\*\chi}}$. 
Note that $\Omega(\pi)=\Omega(\mu)$ since $\*\theta$ is positive.

We denote by $\+I_\pi = (V, E, \beta, \gamma, (\lambda_v)_{v \in V})$ the two-spin system that is the same as $\+I$
except that in $\+I_\pi$ each vertex $v\in V$ is associated with the local field 
\begin{align}\label{eq:good-direction-local-field}
\lambda_v = \lambda \theta_v^{\chi_v}.
\end{align} 
It is easy to see that $\pi=\Mag{\mu}{\*\theta^{\*\chi}}$ is the Gibbs distribution induced by $\+I_\pi$, more specifically:
\[
\pi(\sigma)\propto \mu(\sigma)\prod_{v\in V:\sigma(v) = \1}\theta_v^{\chi_v} \propto \beta^{m_{\0}(\sigma)}\gamma^{m_{\1}(\sigma)}\prod_{v\in V:\sigma(v)=\1}\lambda_v,
\]
where $m_{i}(\sigma) \triangleq \abs{\{(u,v) \in E\mid \sigma_u = \sigma_v = i\}}$ for $i \in \{\0,\1\}$.
It is then sufficient to verify the $\frac{144}{\delta}$-spectral independence of $\pi$ for arbitrary $\*\theta\in(0,1]^V$.

For every $v\in V$, let $\Delta_v$ denote the degree of $v$ in $G=(V,E)$ and $d_v\triangleq \Delta_v-1$. We have $\Delta=\max_{v\in V}\Delta_v$.
\begin{lemma}\label{lemma-unique-closed-in-good-direction}
If $(\beta,\gamma,\lambda)$ is up-to-$\Delta$ unique with gap $\delta$,
then for every $v\in V$, $(\beta,\gamma,\lambda_v)$ is $d_v$-unique with gap $\delta$.
\end{lemma}

The above lemma is implied by the following proposition from \cite{LLY13}.

\begin{proposition}[\cite{LLY13}] \label{prop:uniqueness-property}
If $\beta = 0$, then the following holds for all integers $d \geq 1$:
\begin{itemize}
\item $(0,\gamma,\lambda)$ is $d$-unique with gap $\delta$ iff $\lambda \le \lambda_{c,\delta}(d)= \frac{(1-\delta) d^d \gamma^{d+1}}{(d-1+\delta)^{d+1}}$.	
\end{itemize}
Assume $\beta > 0$.
Let $\overline{\Delta} \triangleq \frac{1 + \sqrt{\beta\gamma}}{1 - \sqrt{\beta\gamma}}$. The followings hold for all integers $d \geq 1$:
\begin{itemize}
\item If $d < (1 - \delta)\overline{\Delta}$, then $(\beta, \gamma,\lambda)$ is $d$-unique with gap $\delta$ for all $\lambda >0$.
\item If $d \geq (1 - \delta)\overline{\Delta}$, then $f_d(x) = 1 - \delta$ has two nonnegative roots $x_1(d)$ and $x_2(d)$ such that
  \begin{align*}
    x_1(d) &= \frac{\zeta_\delta(d) - \sqrt{\zeta_\delta(d)^2 - 4(1 - \delta)^2\beta\gamma}}{2(1- \delta)\beta}
    \quad \mbox{and} \quad
    x_2(d) = \frac{\zeta_\delta(d) + \sqrt{\zeta_\delta(d)^2 - 4(1 - \delta)^2\beta\gamma}}{2(1- \delta)\beta},
  \end{align*}
  where $\zeta_\delta(d) \triangleq d(1 - \beta\gamma) - (1 - \delta)(1 + \beta\gamma) > 0$.
\item For $ d \geq (1 - \delta)\overline{\Delta}$ and $i \in \{1, 2\}$ let
  \[\lambda_i(d) = x_i(d)\tp{\frac{x_i(d) + \gamma}{\beta x_i(d) + 1}}^d.\]
   It holds that $\lambda_1(d)\lambda_2(d) = \tp{\frac{\gamma}{\beta}}^{d+1}$. 
   And $(\beta, \gamma,\lambda)$ is $d$-unique with gap $\delta$ iff $\lambda \in (0, \lambda_1(d)] \cup [\lambda_2(d), +\infty)$.
\end{itemize}
\end{proposition}

The proposition is slightly refined from the one proved in \cite{LLY13} by taking the gap $\delta$ into consideration. 
A proof of \Cref{prop:uniqueness-property} is included in \Cref{proof-prop:uniqueness-property} for completeness.

\begin{proof}[Proof of \Cref{lemma-unique-closed-in-good-direction}]
Fix an arbitrary $v \in V$. We consider the following two cases.

Assume $\beta = 0$. Note that $(0,\gamma,\lambda)$ is $d_v$-unique with gap $\delta$. By \Cref{prop:uniqueness-property}, $\lambda\le \lambda_{c,\delta}(d_v)$.
According to~\eqref{eq:direction},
when $\beta=0$, $\chi_v=\1$ regardless of $\Delta_v$. 
Hence, $\lambda_v = \lambda \theta_v \le \lambda \le \lambda_{c,\delta}(d_v)$. 
By \Cref{prop:uniqueness-property}, $(0,\gamma,\lambda_v)$ is $d_v$-unique with gap $\delta$.

Assume $\beta > 0$. 
We may further assume $d_v \ge (1-\delta) \overline{\Delta}$.
Otherwise, by \Cref{prop:uniqueness-property}, $(\beta,\gamma,\lambda_v)$ is $d_v$-unique.
Since $(\beta,\gamma,\lambda)$ is $d_v$-unique with gap $\delta$, 
by \Cref{prop:uniqueness-property}, it holds that 
$\lambda \in (0,\lambda_1(d_v)] \cup [\lambda_2(d_v),+\infty)$.  
We further consider the following two sub-cases.
\begin{itemize}
\item 	If $\lambda \le \lambda_1(d_v)$, since $\lambda_1(d_v) \lambda_2(d_v) = \tp{{\gamma}/{\beta}}^{d_v+1} = \tp{{\gamma}/{\beta}}^{\Delta_v}$, we have $\lambda \le \lambda_1(d_v) \le \tp{{\gamma}/{\beta}}^{\frac{\Delta_v}{2}}$.
By~\eqref{eq:direction}, we have $\lambda_v=\lambda\theta_v \le \lambda \le \lambda_1(d_v)$.
\item If $\lambda>\lambda_1(d_v)$, since $(\beta,\lambda,\gamma)$ is $d_v$-unique, we have $\lambda \ge \lambda_2(d_v) \ge \tp{{\gamma}/{\beta}}^{\frac{\Delta_v}{2}}$. By~\eqref{eq:direction}, we have  $\lambda_v = \lambda/\theta_v \ge \lambda\ge \lambda_2(d_v)$.
\end{itemize}
In both cases, $\lambda_v\in (0,\lambda_1(d_v)] \cup [\lambda_2(d_v),+\infty)$, and  by \Cref{prop:uniqueness-property}, $(\beta,\gamma,\lambda_v)$ is $d_v$-unique with gap $\delta$.
%
%
\end{proof}

We consider the following well known {tree recursions} for two-spin systems~\cite{Wei06}.
Let $0\le \beta\le \gamma$, $\gamma>0$ and $\lambda>0$ be reals (not necessarily satisfying the anti-ferromagnetic requirement $\beta\gamma<1$).
For integer $d \ge 0$ and real $\lambda >0$, the \emph{tree recursion for marginal-ratios} $F_{\lambda, d}: [0,+\infty]^d \to [0,+\infty]$ is defined as
\begin{align}
\label{eq:tree-recursion-marginal-ratio}
  F_{\lambda, d}(x_1, x_{2}, \cdots, x_{d}) &\triangleq \lambda \prod_{i = 1}^d \frac{\beta x_{i} + 1}{\gamma + x_{i}}. 
\end{align}
In particular, $F_{\lambda,0}\triangleq\lambda$ is a constant for the trivial case $d = 0$.
The \emph{tree recursion for log-marginal-ratios} $H_{\lambda, d}: [-\infty, +\infty]^d \to [-\infty, +\infty]$ is given by $H_{\lambda, d} = \log \circ F_{\lambda, d} \circ \exp$, specifically
\begin{align}
\label{eq:tree-recursion-log-marginal-ratio}
  H_{\lambda, d}(y_1, \cdots, y_d) &\triangleq \log \lambda + \sum_{i=1}^d \log \tp{\frac{\beta \mathrm{e}^{y_i} + 1}{\mathrm{e}^{y_i} + \gamma}}.
\end{align}
For $y \in [-\infty,+\infty]$, let
\begin{align}
\label{eq-def-function-h}
  h(y) &\triangleq - \frac{(1 - \beta\gamma)\mathrm{e}^y}{(\beta \mathrm{e}^y + 1)(\mathrm{e}^y + \gamma)}.
\end{align}
It holds that  $\frac{\partial}{\partial y_i} H_{\lambda,d}(y_1, \cdots, y_d) = h(y_i)$ for all $1 \leq i \leq d$.
Note that for $\beta>0$, $h(y)=0$ iff $y=\pm\infty$; and for $\beta=0$, $h(y)=0$ iff $y=-\infty$.


Given a function $\phi:[-\infty,+\infty] \rightarrow [0, +\infty)$ such that $\phi(y) > 0$ for any $y \in [-\infty,+\infty]$ that $|h(y)|>0$ ,
let $h^{\phi}:[-\infty,+\infty]\to[0,+\infty)$ be defined as that for any $y\in [-\infty,+\infty]$, $h^{\phi}(y)=0$ if $h(y)=0$, and if $h(y)\neq0$,
\begin{align}\label{eq:definition-h-phi}
h^{\phi}(y)= 
\frac{|h(y)|}{\phi(y)}. 
\end{align}
%
Furthermore, let $J_{\lambda, d} \triangleq H_{\lambda, d}[-\infty, +\infty]^d$ denote the image of $H_{\lambda, d}$.
Specifically, $J_{\lambda,0}=\{\log \lambda\}$ for $d = 0$ and if $d > 0$,
\[
J_{\lambda, d}=
\begin{cases}
\vspace{6pt}
\left[-\infty, \log \tp{\frac{\lambda}{\gamma^{d}}}\right] & \text{if }\beta = 0;\\
\vspace{6pt}
\left[\log\tp{\lambda \beta^{d}}, \log\tp{\frac{\lambda}{\gamma^{d}}}\right] & \text{if }0 < \beta\gamma \le 1;\\
\left[\log\tp{\frac{\lambda}{\gamma^{d}}}, \log\tp{\lambda\beta^{d}}\right] & \text{if }\beta\gamma > 1,
\end{cases}
\]


\begin{definition}[$(\alpha, c)$-potential function]\label{def:potential}
Let $\+I=(V,E,\beta,\gamma,(\lambda_v)_{v\in V})$ be a two-spin system with local fields, where $0\le\beta\le\gamma$, $\gamma>0$, and $\lambda_v>0$ for all $v\in V$.
For every $v\in V$, let $d_v\triangleq\Delta_v-1$, where $\Delta_v$ denotes the degree of $v$ in $G=(V,E)$.
Let $\phi:[-\infty,+\infty] \rightarrow [0, +\infty)$ be a function such that $\phi(y) > 0$ for any $y \in [-\infty,+\infty]$ that $|h(y)|>0$.
For any $\alpha\in(0,1)$ and $c>0$, we say $\phi$ is an $(\alpha,c)$-potential function with respect to $\+I$ if it satisfies: 
\begin{enumerate}
\item ($\alpha$-Contraction) 
For every $v\in V$ with $d_v \geq 1$ and every $(y_1,\ldots,y_{d_v})\in [-\infty, +\infty]^{d_v}$, we have
\[
  \phi(y)
  \sum_{i=1}^{d_v} h^{\phi}\tp{y_i}  \le 1-\alpha.
\]
where $y=H_{\lambda_v,d_v}(y_1,\ldots,y_d)$.
\item ($c$-Boundedness)
For every $u, v \in V$, every $y_u \in J_{\lambda_u,d_u}$ and $y_v \in J_{\lambda_v,d_v}$, we have
\begin{align*}
      \phi(y_v)\cdot h^{\phi}\tp{y_u} \leq \frac{2c}{\Delta_{u} + \Delta_{v}}.
\end{align*}
\end{enumerate}
\end{definition}


\begin{theorem}[\cite{chen2020rapid}]\label{theorem-good-potential-imply-SI-CLV}
Let $\+I=(V,E,\beta,\gamma,(\lambda_v)_{v\in V})$ be a two-spin system with local fields, where $0\le\beta\le\gamma$, $\gamma>0$, and $\lambda_v>0$ for all $v\in V$.
For $\alpha\in(0,1)$ and $c>0$,
if there is an $(\alpha,c)$-potential function $\phi$ with respect to $\+I$, 
then the Gibbs distribution $\pi$ of $\+I$ is $\frac{2c}{\alpha}$-spectrally independent.
\end{theorem}

\begin{remark}
The potential function $\phi$ in \Cref{def:potential} is in fact the derivative of the potential function in \cite[Definition 4]{chen2020rapid}.
\Cref{theorem-good-potential-imply-SI-CLV} holds without assuming that $\phi$ has an explicitly defined integration. 
\end{remark}

\begin{remark}
\Cref{theorem-good-potential-imply-SI-CLV} holds for general two-spin systems that are not necessarily anti-ferromagnetic.
The theorem stated here is in fact a refinement from the one proved in~\cite{chen2020rapid} to the two-spin systems with local fields.
In \Cref{proof-theorem-good-potential-imply-SI-CLV}, we reprove the theorem by going through the analyses in~\cite{chen2020rapid}.
Another slight difference is that our spectral independence is defined with the absolute influence matrix instead of signed influence matrix.
This is not an issue, because in the proof of \Cref{theorem-good-potential-imply-SI-CLV}, the spectral independence is guaranteed by establishing a sufficient condition of weighted total influences, which is sufficient to imply spectral independence with either absolute or signed matrix.
\end{remark}

It remains to verify the contraction and boundedness properties for the instance $\+I_{\pi}=(V,E,\beta,\gamma,(\lambda_v)_{v\in V})$, where the local fields $(\lambda_v)_{v\in V}$ are as specified in~\eqref{eq:good-direction-local-field}.

For anti-ferromagnetic $2$-spin systems, a good potential function $\phi$ is the one discovered in~\cite{LLY13}:
\begin{align}\label{eq:LLY-potential-function}
  \phi(y) = \sqrt{\abs{h(y)}}.
\end{align}
For such choice of potential function it obviously holds that $\phi(y) > 0$ for any $y \in [-\infty,+\infty]$ that $|h(y)|>0$,
and moreover, for any $y\in[-\infty,+\infty]$,
\[
h^{\phi}(y)=\phi(y)=\sqrt{\abs{h(y)}}.
\]

\begin{lemma}[\cite{LLY13}]\label{lm:contraction}
  Let $d \ge 1$ be an integer, and let $\beta,\gamma,\lambda$ be real numbers satisfying that $0\le \beta\le\gamma$, $\gamma> 0$, $\lambda> 0$ and $\beta\gamma<1$.
For the function $\phi$ defined in~\eqref{eq:LLY-potential-function},
for any $\delta\in(0,1)$,
if $(\beta,\gamma,\lambda)$ is $d$-unique with gap $\delta$, then for every $(y_1, \cdots, y_d) \in [-\infty, +\infty]^d$ and $y = H_{\lambda, d}(y_1, y_2, \cdots, y_d)$,
\begin{align}\label{eq:contraction-property-LLY-CLV}
    \phi(y)\sum_{i=1}^{d} h^{\phi}\tp{y_i}
    = \sum_{i = 1}^d \sqrt{\abs{h(y)}\abs{h(y_i)}} \le \sqrt{1-\delta} < 1-\frac{\delta}{2}.
  \end{align}
\end{lemma}
\begin{remark}
It was proved in \cite{LLY13} that \eqref{eq:contraction-property-LLY-CLV} holds for all integers $1\le d<\Delta$ if $(\beta,\gamma,\lambda)$ is up-to-$\Delta$ unique, which was in fact proved by showing that \eqref{eq:contraction-property-LLY-CLV} holds if $(\beta,\gamma,\lambda)$ is $d$-unique, which is the exact statement in \Cref{lm:contraction}.
A proof of \Cref{lm:contraction} is included in \Cref{proof-lm:contraction} for completeness,.
\end{remark}

Assume that $(\beta,\gamma,\lambda)$ is up-to-$\Delta$ unique with gap $\delta$.
It follows immediately from \Cref{lemma-unique-closed-in-good-direction} and \Cref{lm:contraction} that the potential function $\phi$ defined in~\eqref{eq:LLY-potential-function} satisfies $\frac{\delta}{2}$-contraction with respect to instance $\+I_{\pi}$.

{The following lemma for boundedness was proved in \cite{chen2020rapid}.}
\begin{lemma}[\text{\cite{chen2020rapid}}]\label{lemma-boundedness-CLV}
  Let $\Delta\ge 3$  be an integer, and let $\beta,\gamma,\lambda$ be real numbers satisfying that $0\le \beta\le\gamma$, $\gamma> 0$, $\lambda> 0$ and $\beta\gamma<1$.
For the potential  function $\phi$ defined in~\eqref{eq:LLY-potential-function}, for any $\delta\in(0,1)$,
if $(\beta, \gamma, \lambda)$ is up-to-$\Delta$ unique with gap $\delta$, then for any integers $0 \leq d_1, d_2 \le \Delta-1$, for every $y_1 \in J_{\lambda, d_1}$ and $y_2 \in J_{\lambda, d_2}$,
  \begin{align*}
   \phi(y_1)\cdot h^{\phi}\tp{y_2}
    = \sqrt{\abs{h(y_1)}\abs{h(y_2)}} &\leq \frac{72}{d_1 + d_2 + 2}.
  \end{align*}
\end{lemma}

Now, to verify the boundedness for $\phi$ with respect to $\+I_\pi$, it is sufficient to show that for every $v \in V$,
\begin{align} \label{eq:better-boundness}
  \max_{y \in J_{\lambda_v, d_v}}\abs{h(y)} &\leq \max_{y\in J_{\lambda,d_v}}\abs{h(y)}.
\end{align}
Note that the maximum on the right-hand-side is taken over $J_{\lambda,d_v}$ and the maximum on the left-hand-side is taken over $J_{\lambda_v,d_v}$ where $\lambda_v=\theta_v^{\chi_v}$ is the local field associated with $v\in V$ in $\+I_\pi$, specifically
\begin{align*}
  \lambda_v &= \begin{cases}
      \lambda \cdot \theta_v & \lambda \leq \tp{\frac{\gamma}{\beta}}^{\Delta_v / 2}, \\
      \lambda / \theta_v & \text{otherwise},
    \end{cases}
\end{align*}
for an arbitrarily fixed $\theta_v\in(0,1)$.

When $\beta=0$, $\abs{h(y)} = \frac{\mathrm{e}^y}{\mathrm{e}^y + \gamma}$ is monotonically increasing in $y$, and for every $v\in V$, we have $\lambda_v=\lambda\cdot\theta_v<\lambda$  and hence $\max J_{ \lambda_v, d_v} \leq \max J_{\lambda, d_v}$. It holds that
\begin{align*}
  \max_{y \in J_{\lambda_v,d_v}} \abs{h(y)} = \abs{h(\max J_{\lambda_v,d_v})}
  &\leq \abs{h(\max J_{\lambda,d_v})} = \max_{y \in J_{\lambda,d_v}}\abs{h(y)}.
\end{align*}
The inequality~\eqref{eq:better-boundness} follows.

In the following we assume $\beta>0$, i.e.~$0<\beta\gamma<1$.
Note that $\abs{h(y)} = \frac{(1 - \beta\gamma)\mathrm{e}^y}{(\beta \mathrm{e}^y + 1)(\mathrm{e}^y + \gamma)}$ is monotonically increasing in $y$ when $y \leq \frac{1}{2}\log\tp{\frac{\gamma}{\beta}}$, and monotonically decreasing in $y$ when $y \geq \frac{1}{2}\log\tp{\frac{\gamma}{\beta}}$, so that $\abs{h(y)}$ achieves the maximum at $y = \frac{1}{2}\log\tp{\frac{\gamma}{\beta}}$.
Then, we can verify \eqref{eq:better-boundness} by considering three cases: 
\begin{enumerate}
\item\label{boundedness-case-1}
$\frac{1}{2}\log\tp{\frac{\gamma}{\beta}} \in J_{\lambda,d_v}$; 
\item\label{boundedness-case-2}
$\frac{1}{2}\log\tp{\frac{\gamma}{\beta}} > \log\tp{\frac{\lambda}{\gamma^{d_v}}}$; 
\item \label{boundedness-case-3}
$\frac{1}{2}\log\tp{\frac{\gamma}{\beta}} < \log\tp{\lambda\beta^{d_v}}$.
\end{enumerate}
For Case.\ref{boundedness-case-1}, \eqref{eq:better-boundness} holds trivially because $\max_{y \in J_{\lambda,d_v}}\abs{h(y)}$ achieves the global maximum of $\abs{h(y)}$.

We then verify \eqref{eq:better-boundness} in Case.\ref{boundedness-case-2}. Notice that when $\frac{1}{2}\log\tp{\frac{\gamma}{\beta}} > \log\tp{\frac{\lambda}{\gamma^{d_v}}}$, it holds that
\begin{align*}
  \lambda
  &< \gamma^{d_v} \tp{\frac{\gamma}{\beta}}^{\frac{1}{2}}
    =\tp{\frac{\gamma}{\beta}}^{\frac{d_v + 1}{2}} \tp{\beta\gamma}^{\frac{d_v}{2}}
    \leq \tp{\frac{\gamma}{\beta}}^{\Delta_v/2},
\end{align*}
which implies $\lambda_v=\lambda\cdot\theta_v < \lambda$ and hence $\max J_{ \lambda_v, d_v} \leq \max J_{\lambda, d_v} = \log\tp{\frac{\lambda}{\gamma^{d_v}}}< \frac{1}{2}\log\tp{\frac{\gamma}{\beta}}$.
Since $\abs{h(y)}$ is monotonically increasing in $y$ when $y\le \frac{1}{2}\log\tp{\frac{\gamma}{\beta}}$, it holds that
\begin{align*}
  \max_{y \in J_{\lambda_v,d_v}} \abs{h(y)} = \abs{h(\max J_{\lambda_v,d_v})}
  &\leq \abs{h(\max J_{\lambda,d_v})} = \max_{y \in J_{\lambda,d_v}}\abs{h(y)}.
\end{align*}

Finally, we verify \eqref{eq:better-boundness} in Case.\ref{boundedness-case-3}. 
Notice that when $\frac{1}{2}\log\tp{\frac{\gamma}{\beta}} < \log\tp{\lambda\beta^{d_v}}$, it holds that
\begin{align*}
  \lambda
  &> \beta^{-d_v} \tp{\frac{\gamma}{\beta}}^{\frac{1}{2}}
    =\tp{\frac{\gamma}{\beta}}^{\frac{d_v + 1}{2}} \tp{\beta\gamma}^{-\frac{d_v}{2}}
    \geq \tp{\frac{\gamma}{\beta}}^{{\Delta_v}/{2}},
\end{align*}
which implies $\lambda_v=\lambda/\theta_v > \lambda$, and hence $\frac{1}{2}\log\tp{\frac{\gamma}{\beta}} <\log\tp{\lambda\beta^{d_v}}= \min J_{\lambda,d_v} \leq \min J_{\lambda_v,d_v}$. Since $\abs{h(y)}$ is monotonically decreasing in $y$ when $y \ge \frac{1}{2}\log\tp{\frac{\gamma}{\beta}}$, it holds that
\begin{align*}
  \max_{y \in J_{\lambda_v,d_v}}\abs{h(y)} = \abs{h(\min J_{\lambda_v,d_v})} \leq \abs{h(\min J_{\lambda,d_v})} = \max_{y \in J_{\lambda,d_v}} \abs{h(y)}.
\end{align*}
Therefore, we prove \eqref{eq:better-boundness}. 
It immediately follows from \Cref{lemma-boundedness-CLV} and \eqref{eq:better-boundness} that the potential function $\phi$ defined in~\eqref{eq:LLY-potential-function} satisfies $36$-boundedness with respect to instance $\+I_{\pi}$.

Along with the $\frac{\delta}{2}$-contraction we have established for $\phi$ with respect to $\+I_{\pi}$, this guarantees that $\phi$ is a $(\frac{\delta}{2},36)$-potential function with respect to $\+I_{\pi}$.
Due to \Cref{theorem-good-potential-imply-SI-CLV}, the distribution $\pi$ is $\frac{144}{\delta}$-spectrally independent.
Note that this holds for $\pi=\Mag{\mu}{\*\theta^{\*\chi}}$ for arbitrary $\*\theta\in(0,1)^V$.
This proves the complete $\frac{144}{\delta}$-spectral independence of $\mu$ in direction $\*\chi$ claimed in the first half of \Cref{lm:two-spin-property}.

\subsection{Spectral gap in an easier regime}\label{sec:req-2}
We now prove the second half of \Cref{lm:two-spin-property}.

Let $\delta \in (0, 1)$.
Assume that the anti-ferromagnetic $2$-spin system $\+I = (V, E, \beta, \gamma,\lambda)$ is up-to-$\Delta$ unique with gap $\delta$.
Fix $\theta = \delta^2/64$ and $C = 8/\delta$. 
Let 
\[
\pi = \Mag{\mu}{\theta^{\*\chi}},
\] 
where $\mu$ is the Gibbs distribution associated with $\+I$, and $\*\chi \in \{\0, \1\}^V$ is the good direction defined in~\eqref{eq:direction}.

Our goal is to show that
\begin{align}
  \spgap{\sgap}{\GD}(\pi) &\geq \frac{1}{Cn}.\label{eq:min-gap-bound-2spin}
\end{align}
where $\spgap{\sgap}{\GD}(\pi) \triangleq \min_{\Lambda \subseteq V, \sigma_\Lambda \in \Omega(\pi_\Lambda)} \spgap{gap}{\GD}\tp{\pi^{\sigma_\Lambda}}$ denotes the minimum spectral gap of the Glauber dynamics for $\pi$ with worst-case feasible boundary condition.

Let $\Lambda \subseteq V$ be a subset of vertices and $\sigma \in \Omega(\pi_\Lambda)$ a feasible partial configuration specified on $\Lambda$. Let $P_\sigma$ denote the Glauber dynamics for  $\pi^\sigma$.
We consider the following coupling of chain $P_\sigma$.

For any two configurations $\*X,\*Y \in \{\0,\1\}^V$, we use $\Phi(\*X,\*Y)$ to denote the \emph{weighted hamming distance} between $\*X$ and $\*Y$. Formally 
\begin{align*}
\Phi(\*X,\*Y) \triangleq \sum_{v \in V: X_v \neq Y_v} \Phi_v,	
\end{align*}
where for each $v \in V$, $\Phi_v$ is defined by
  \begin{align*}
    \Phi_v &\triangleq \begin{cases}
      1 - \frac{\delta}{8} &, \Delta_v = 1, \\
      \Delta_v &, \Delta_v > 1,
    \end{cases}
  \end{align*}	
where $\Delta_v$ denotes the degree of $v$ in graph $G=(V,E)$.

\begin{lemma} \label{lem:fd-to-coupling}
  Assume that $\+I$ is up-to-$\Delta$ unique with gap $\delta$.
  For any $\*X, \*Y \in \Omega(\pi^\sigma)$, there is a coupling $(\*X,\*Y)\to (\*X',\*Y')$ of Markov chain $P_\sigma$ such that 
  \begin{align*}
  \E{\Phi(\*X',\*Y') \mid \*X,\*Y } \leq   \tp{1 - \frac{\delta}{8n}} \Phi(\*X,\*Y). 
  \end{align*}
\end{lemma}

Due to the well known connection between coupling and spectral gap (\Cref{lemma:MFC}), our desired spectral gap bound in~\eqref{eq:min-gap-bound-2spin} is implied by \Cref{lem:fd-to-coupling}, and hence this proves the second half of \Cref{lm:two-spin-property}.

\begin{remark}
Usually when using coupling to analyze the mixing time upper bound, the metric $\Phi(\cdot,\cdot)$ is required to satisfy $\min_{\*X,\*Y \in \Omega}\Phi(\*X,\*Y) \geq 1$. 
Our weighted Hamming distance {per se} does not satisfy this.
This does not matter though, because our purpose of using coupling here is not to bound the mixing time, but the spectral gap, and \Cref{lemma:MFC} holds as long as there is sufficient decay in the metric.
If one wants to analyze the mixing time using our coupling, one can simply renormalize by replacing $\Phi$ with $2\Phi$, so that the minimum distance is bounded from below and the step-wise decay in \Cref{lem:fd-to-coupling} still holds.
\end{remark}

It only remains to prove \Cref{lem:fd-to-coupling}.

\begin{proof}[Proof of \Cref{lem:fd-to-coupling}]
Observe that $\pi=\Mag{\mu}{\theta^{\*\chi}}$ corresponds to the Gibbs distribution of a 2-spin system where the pairwise interactions on edges are the same as in $\mu$, but now each $v\in V$ has a local field:
\[
\lambda_v\triangleq \lambda \theta^{\chi_v}.
\]
More specifically, for every $\tau \in \{\0,\1\}^V$,
\[
\pi(\tau) \propto  \beta^{m_{\1}(\tau)} \gamma^{m_{\0}(\tau)}\prod_{v:\tau(v)=\1}\lambda_v,
\]
where 
$m_{i}(\tau) \triangleq \abs{\{(u,v) \in E\mid \tau_u = \tau_v = i\}}$ for $i \in \{\0,\1\}$.


The coupling is constructed by the path coupling~\cite{bubley1997path}.
Note that $P_\sigma$ is a Markov over $\Omega(\pi^\sigma)$. 
To apply the path coupling, we first extend $P_\sigma$ to the entire space $\{\0,\1\}^V$.
Suppose the current configuration is $\*X\in\{\0,\1\}^V$, not necessarily feasible with respect to $\pi$.
Upon transition, a vertex $v \in V$ is picked uniformly at random;
if $v \in \Lambda$, then the value of $v$ is fixed by $\sigma$ and  set $X_v \gets \sigma_v$;
if $v \in V \setminus \Lambda$, then the current value of $X_v$ is updated to a random value $c_v \in \{\0,\1\}$ such that
\begin{equation}
\label{eq-def-ext-prob}
\begin{split}
\Pr[]{c_v = \0} = \hat{\mu}^{X_{V \setminus \{v\}}}_v(\0) =  p_v(\0,s) \triangleq \frac{\gamma^s }{\gamma^s + \lambda_v \beta^{\Delta_v - s}},\\
\Pr[]{c_v = \1} = \hat{\mu}^{X_{V \setminus \{v\}}}_v(\1)  =p_v(\1,s) \triangleq \frac{\lambda_v \beta^{\Delta_v - s}}{\gamma^s + \lambda_v \beta^{\Delta_v - s}},
\end{split} 
\end{equation}
where $s \triangleq |\{u \mid v \in \Gamma_v \land X_v = \0 \}|$ denotes the  number of $\0$'s assigned by $X$ to the neighborhood $\Gamma_v$ of $v$ in $G$,
and the value of $\beta^{\Delta_v - s}$ is computed with convention $0^0=1$.\footnote{This convention is assumed throughout the proof without further mentioning.}

In~\eqref{eq-def-ext-prob}, $\hat{\mu}$ extends the definition of conditional distribution induced by $\mu$ to the boundary conditions that may be infeasible in general.
For any feasible configuration $\*X$, it is easy to see that $\hat{\mu}^{X_{V \setminus \{v\}}}_v = \mu^{X_{V \setminus \{v\}}}_v$.

Let $\*X,\*Y\in\{\0,\1\}^V$ be disagreeing with each other at only one vertex $v\in V$.
The coupling $(\*X,\*Y)\to(\*X',\*Y')$ is constructed as follows:
\begin{itemize}
\item the two chains pick the same vertex $w \in V$ uniformly at random, and $X'_u=Y'_u$ for all $u\neq w$;
\item $(X'_w,Y'_w)$ is drawn according to the optimal coupling of their marginal distributions in~\eqref{eq-def-ext-prob}.
\end{itemize}
If the Glauber dynamics picks $v$, then $X'_v = Y'_v$; otherwise, $X'_v \neq Y'_v$. We have $\Pr[]{X'_v \neq Y'_v \mid \*X,\*Y} = 	1 - \frac{1}{n}$.
For any $w \not\in\Gamma_v\cup\{v\}$, it holds that $\Pr[]{X'_w \neq Y'_w \mid \*X,\*Y} =  0$.
For any $u \in \Gamma_v$, $X'_u \neq Y'_u$ only if the Glauber dynamics picks $u$ and the coupling on vertex $u$ fails. 
We have
\begin{align}
\label{eq-coupling-neighbor}
\forall u \in \Gamma(v),\quad \Pr[]{X'_u  \neq Y'_u \mid \*X,\*Y} \leq \frac{1}{n}R(u,v),	
\end{align}
where $R(v, u)$ corresponds to the Dobrushin's influence matrix, formally defined as follows: 
\[
R(v, u) \triangleq \max_{(\sigma,\tau) \in B_v}\DTV{\hat{\mu}^{\sigma_{V \setminus \{u\} }}_u}{\hat{\mu}^{\tau_{V \setminus \{u\} }}_u },
\]
where $B_v$ denotes the set of all pairs $(\sigma,\tau) \in \{\0,\1\}^{V}\times\{\0,\1\}^{V}$ that disagree only at $v$. 
We have
\begin{align*}
R(v, u) 
&= \max_{0 \leq s \leq \Delta_u  -1 }\abs{p_u(\1, s+1) - p_u(\1,s)} \\
&= \max_{0 \leq s \leq \Delta_u  -1 }\frac{\lambda_u \beta^{\Delta_u-s-1}\gamma^s(1-\beta\gamma)}{ (\gamma^{s+1} + \lambda \beta^{\Delta_u-s-1})(\gamma^s + \lambda_u \beta^{\Delta_u - s} ) }\\
&{=} \max_{0 \leq s \leq \Delta_u - 1}\frac{\lambda_u \beta^{s}\gamma^{\Delta_u - s- 1}(1-\beta\gamma)}{ (\gamma^{\Delta_u-s} + \lambda \beta^{s})(\gamma^{\Delta_u-s-1} + \lambda_u \beta^{s+1} ) }
&&(\mbox{by replacing $s$ with $\Delta_u -s - 1$})\\
&=\max_{0 \leq s \leq \Delta_u - 1}\frac{\lambda_u \beta^{s}\gamma^{-\Delta_u + s + 1}(1-\beta\gamma)}{ (\gamma + \lambda \beta^s \gamma^{-\Delta_u+s + 1} )(1 + \lambda_u \beta^{s+1} \gamma^{-\Delta_u+s + 1}) }
&& (\text{since }\gamma > 0)\\
& = \max_{0 \leq s \leq \Delta_u - 1}  \frac{1}{\Delta_u}f_{\Delta_u}\tp{ \frac{\lambda_u (\beta\gamma)^s}{\gamma^{\Delta_u - 1}} },
\end{align*}
where  
the function $f_d$ for integer $d$ is as defined in~\eqref{eq-def-fd}: 
\begin{align*}
f_d(x) = \frac{d(1- \beta\gamma)x}{(\beta x + 1)(x + \gamma)}.	
\end{align*}

Altogether, we have
\begin{align}
\label{eq-path-coupling-final}
 \E{\Phi(\*X',\*Y') \mid \*X,\*Y } \leq \Phi_v\tp{1 - \frac{1}{n}} +\frac{1}{n} \sum_{u \in \Gamma_v}	 \frac{\Phi_u}{\Delta_u}f_{\Delta_u}\tp{ \frac{\lambda_u (\beta\gamma)^s}{\gamma^{\Delta_u - 1}} }.
\end{align}
\begin{remark}
Note that $f_{\Delta_u}(\cdot)$ captures the contraction of the tree recursion for the marginal ratio of $(\Delta_u+1)$-regular tree.
We know that $f_{\Delta_u}<1$ at the fixed point for the tree recursion, if we had assumed the $(\Delta_u+1)$-uniqueness.
However, we only assume the up-to-$\Delta$ uniqueness of $(\beta,\gamma,\lambda)$. Such discrepancy is due to the non-self-avoiding nature of path coupling argument.
Nevertheless, since we have moved to an easier regime $\pi=\Mag{\mu}{\theta^{\*\chi}}$ that effectively alters the local fields from $\lambda$ to $\lambda_v=\lambda\theta^{\chi_v}$, we could hope for that  $f_{\Delta_u}$ still contracts at $\frac{\lambda_u (\beta\gamma)^s}{\gamma^{\Delta_u - 1}}$ as long as it is fairly close to the fixed point.
\end{remark}
The above intuition is formally justified by the following claim. 
\begin{claim} \label{lem:IS-to-fd}
For any $u \in V$, if $\Delta_u > 1$, then for any integer $0 \leq s \leq \Delta_u - 1$, it holds that
  \begin{align*}
    f_{\Delta_u}\tp{\frac{\lambda_u(\beta\gamma)^s}{\gamma^{\Delta_u - 1}}} < 1 - \frac{\delta}{4}. 
  \end{align*}
\end{claim}

\Cref{lem:IS-to-fd} will be proved later.
We now use \Cref{lem:IS-to-fd} to bound~\eqref{eq-path-coupling-final}. 
Without loss of generality, we assume the underlying graph $G=(V,E)$  is connected. 
Otherwise, we can decompose the spin system $\+I$ into a set of independent systems.
Consider the following two cases: 
\begin{itemize}
\item Case $\Delta_v = 1$. We have $\Phi_v = 1 - \frac{\delta}{8}$. Let $\Gamma_v = \{u\}$. If $\Delta_u = 1$, since $G$ is connected, $G$ only contains two vertices, such instance is trivial.  If $\Delta_u > 1$, it holds that $\Phi_u = \Delta_u$. By \eqref{eq-path-coupling-final} and \Cref{lem:IS-to-fd},
\begin{align*}
 \E{\Phi(\*X',\*Y') \mid \*X,\*Y } 
 &\leq \tp{1 - \frac{\delta}{8}}\tp{1 - \frac{1}{n}} + \frac{1}{n}\tp{1 - \frac{\delta}{4}}\\
 &\le \tp{1 - \frac{\delta}{8n}} \Phi_v.
\end{align*}
\item Case $\Delta_v > 1$. We have $\Phi_v = \Delta_v$. Fix any neighbor $u \in \Gamma_v$. If $\Delta_u > 1$, then $\Phi_u = \Delta_u$, by \Cref{lem:IS-to-fd},
\begin{align*}
\frac{\Phi_u}{\Delta_u}f_{\Delta_u}\tp{ \frac{\lambda_u (\beta\gamma)^s}{\gamma^{\Delta_u - 1}} } \leq 1 - \frac{\delta}{4}.	
\end{align*}
If $\Delta_u = 1$, then $\Phi_u = 1 - \frac{\delta}{8}$, by the definition of $R(v, u)$, it is straightforward to verify that 
\[
R(v, u)  =\frac{1}{\Delta_u}f_{\Delta_u}\tp{ \frac{\lambda_u (\beta\gamma)^s}{\gamma^{\Delta_u - 1}} } \leq 1.
\] 
Thus, we have
\begin{align*}
\frac{\Phi_u}{\Delta_u}f_{\Delta_u}\tp{ \frac{\lambda_u (\beta\gamma)^s}{\gamma^{\Delta_u - 1}} } \leq \Phi_u = 	 1 - \frac{\delta}{8}.
\end{align*}
Therefore, by \eqref{eq-path-coupling-final}, we have
\begin{align*}
 \E{\Phi(\*X',\*Y') \mid \*X,\*Y } 
 &\leq \Delta_v\tp{1 - \frac{1}{n}} + \frac{\Delta_v}{n}\tp{1 - \frac{\delta}{8}}\\
 &= \tp{1 - \frac{\delta}{8n}} \Phi_v.
\end{align*}	
\end{itemize}

Recall that we have extended the Markov chain $P_\sigma$ to the entire space $\{\0,\1\}^V$.
Since the weighted Hamming distance is a well-defined metric on $\{\0,\1\}^V$, due to the path coupling theorem~\cite{bubley1997path}, for any $\*X,\*Y \in \{\0,\1\}^V$, there is a coupling $(\*X,\*Y) \to (\*X',\*Y')$ of $P_\sigma$ such that
\begin{align*}
 \E{\Phi(\*X',\*Y') \mid \*X,\*Y } \leq 	\tp{1 - \frac{\delta}{8n}}\Phi_v.
\end{align*}
This step-wise decay property holds for every $\*X,\*Y \in \{\0,\1\}^V$, thus it holds for every $\*X,\*Y \in \Omega(\pi^\sigma)$.
This proves the lemma. 
\end{proof}

It remains to prove \Cref{lem:IS-to-fd}.

\begin{proof}[Proof of \Cref{lem:IS-to-fd}]
Let $\delta \in (0, 1)$.
Fix a vertex $u \in V$ with $\Delta_u \geq 2$. The effective local field of $u$ in $\pi$ is given by $\lambda_u = \theta^{\chi_u}$, where $\theta=\frac{\delta^2}{64}$, and $\chi_u = \mathrm{sgn}(\Delta_u, \lambda)$ is as defined in~\eqref{eq:direction}:
\begin{align*}
  \mathrm{sgn}(\Delta_u, \lambda) &\triangleq \begin{cases}
    \1 & \lambda \leq \tp{\frac{\gamma}{\beta}}^{\Delta_u / 2} \\
    \0 & \text{otherwise}.
  \end{cases}
\end{align*}
We denote 
\begin{align*}
D = \Delta_u \quad \text{ and }\quad
d = D - 1.
\end{align*}
Clearly, $2 \leq  D \leq \Delta$ and $1 \leq  d \leq \Delta-1$.

Recall that  $(\beta,\gamma,\lambda)$ is up-to-$\Delta$ unique with gap $\delta$, which means that $(\beta,\gamma,\lambda)$ is $d$-unique with gap $\delta$.
%
%
To prove the claim, it suffices to prove the following proposition: 
For any $(\beta,\gamma,\lambda)$ with $0 \leq \beta \leq \gamma$, $\beta\gamma < 1$ and $\gamma,\lambda > 0$, 
for any integer $d \geq 1$, if $(\beta,\gamma,\lambda)$ is $d$-unique with gap $\delta$, i.e.~the following holds:
\begin{align}
\label{eq-proof-claim-assume-1}
f_d(x_d) = \frac{d(1-\beta \gamma)x_d}{(\beta x_d + 1)(x_d + \gamma)} \leq 1 - \delta,	
\end{align}
where $x_d=F(x_d)$ is the unique positive fixed point for $F_d(x) \triangleq \lambda \tp{ \frac{\beta x + 1}{x + \gamma} }^d$,
then for $D=d+1$,
\begin{align}
\label{eq:target-max}
\max_{0 \leq s \leq D-1}f_{D}\tp{\frac{\lambda\theta^{\mathrm{sgn}(D, \lambda)}(\beta\gamma)^s}{\gamma^{d}}} < 1 - \frac{\delta}{4}. 
\end{align}
%
%
%

We first consider the case $\beta = 0$. In this case, $f_d(x) = \frac{dx}{x+\gamma}$. Hence, $f_d(x) < 1 - \delta$ if and only if $x < \frac{(1-\delta)\gamma}{d-1+\delta}$. And the fixed point $x_d=F(x_d)$ satisfies $\lambda=x_d(x_d+\gamma)^d$. Thus, the assumption~\eqref{eq-proof-claim-assume-1} implies that 
\begin{align}
\label{eq-lambda-hardcore-strong}
\lambda = x_d (x_d + \gamma)^d \leq 	\tp{\frac{(1-\delta)\gamma}{d-1+\delta}}\tp{\frac{(1-\delta)\gamma}{d-1+\delta} + \gamma}^d = \frac{(1-\delta) \gamma^{d+1} d^d}{(d-1+\delta)^{d+1}}.
\end{align}
Note that in this case, $\mathrm{sgn}(D,\lambda) = 1$ for all integer $D \geq 2$ and real number $\lambda > 0$.
Also note that if $s \geq 1$, then $(\beta\gamma)^s = 0$ and $f_D(0) = 0$.
We then only need to prove \Cref{lem:IS-to-fd} when $s = 0$. Formally, this is equivalent to show that 
\[
f_D\tp{\frac{\theta \lambda}{\gamma^d}} < 1 - \frac{\delta}{4}.
\]
Note that
$f_D(x) = \frac{Dx}{x+\gamma}$ when $\beta = 0$.
Hence, $f_D(x) < 1 - \frac{\delta}{4}$ if and only if $x < \frac{\gamma(4-\delta)}{4(D-1)+\delta} = \frac{\gamma(4-\delta)}{4d + \delta}$. We then only need to verify
\begin{align*}
\theta \lambda <  \frac{\gamma^{d+1}(1 - \frac{\delta}{4})}{d + \frac{\delta}{4}}.
\end{align*}
Note that $d \geq 1$ and $0 < \delta < 1$.
By~\eqref{eq-lambda-hardcore-strong}, it suffices to verify 
\begin{align*}
\theta = \frac{\delta^2}{64} \leq \frac{1}{2} \cdot \tp{\frac{d - 1 + \delta }{d} }^{d+1}.
\end{align*}
Let $h(t) = \tp{\frac{t-1+\delta}{t}}^{t+1}$.
The above inequality can be expressed as $h(d)\ge \frac{\delta^2}{32}$, which holds because $h(t)$ is increasing in $t$ for $t \geq 1$ and $h(1) = \delta^2$. To verify the increasing of $h(t)$, observe that 
 \begin{align*}
  (\ln h(t))' &= \ln \tp{1 - \frac{1-\delta}{t}} + \frac{(1-\delta)(t+1)}{t(t-1+\delta)} \to 0 \text{ as } t \to \infty,\\
  (\ln h(t))'' &= \frac{(1-\delta)(-3t+\delta t + 1 - \delta)}{t^2(t-1+\delta)^2} < 0 \text{ if } t \geq 1.
  \end{align*}

Next, we focus on the main case $\beta > 0$. 
We first show that without loss of generality, we can assume $\lambda \le \tp{\gamma/\beta}^{D/2}$. 
Suppose $\lambda > \tp{\gamma/\beta}^{D/2}$.
Let $\lambda' = \frac{1}{\lambda} \cdot (\gamma/\beta)^{D}$.
The following two properties hold:
\begin{itemize}
\item if $(\beta,\gamma,\lambda)$ is $d$-unique with gap $\delta$, then $(\lambda', \beta,\gamma)$ is also $d$-unique with gap $\delta$;
\item $\lambda' < \tp{\gamma/\beta}^{D/2}$ and it holds that 
\begin{align}
\label{eq-equation-f}
  \max_{0 \leq s < D}f_{D}\tp{\frac{\lambda\theta^{\mathrm{sgn}(D, \lambda)}(\beta\gamma)^s}{\gamma^{d}}}
  = \max_{0 \leq s < D}f_{D}\tp{\frac{\lambda'\theta^{\mathrm{sgn}(D, \lambda')}(\beta\gamma)^s}{\gamma^{d}}}.
\end{align}
\end{itemize}

We verify the first property.
If $d < (1-\delta)\overline{\Delta}$, then  by \Cref{prop:uniqueness-property}, the property holds trivially.
Now suppose $d \geq (1-\delta)\overline{\Delta}$. By \Cref{prop:uniqueness-property}, $\lambda_1(d)\lambda_2(d) = \tp{\gamma/\beta}^{D}$. 
Since $(\beta,\gamma,\lambda)$ is $d$-unique with gap $\delta$ and $\lambda > \tp{\gamma/\beta}^{D/2}$,
we have $\lambda > \lambda_2(d)$ and $\lambda' = \frac{1}{\lambda} \cdot (\gamma/\beta)^{D} < \lambda_1(d)$, thus $(\lambda', \beta,\gamma)$ is also $d$-unique with gap $\delta$.

We then verify the second property. Since $\lambda > \tp{\gamma/\beta}^{D/2}$, $\lambda' = \frac{1}{\lambda} \cdot (\gamma/\beta)^{D} < (\gamma/\beta)^{D/2}$. To verify~\eqref{eq-equation-f}, observe that the following holds for the function $f_D$:
\begin{align*}
\forall x > 0, \quad f_D(x) =  \frac{D(1- \beta\gamma)x}{(\beta x + 1)(x + \gamma)} =f_D\tp{ \frac{\gamma}{x \beta} }.	
\end{align*}
This implies 
\begin{align*}
\max_{0 \leq s < D}f_{D}\tp{\frac{\lambda\theta^{\mathrm{sgn}(D, \lambda)}(\beta\gamma)^s}{\gamma^{d}}} 
&= \max_{0 \leq s < D}f_D\tp{\frac{\gamma}{\beta} \cdot \frac{\gamma^{d}}{\lambda\theta^{\mathrm{sgn}(D, \lambda)}(\beta\gamma)^s}}\\	
&= \max_{0 \leq s < D}f_{D}\tp{\frac{\lambda'\theta^{\mathrm{sgn}(D, \lambda')}(\beta\gamma)^{d-s}}{\gamma^{d}}}\\
&=\max_{0 \leq s < D}f_{D}\tp{\frac{\lambda'\theta^{\mathrm{sgn}(D, \lambda')}(\beta\gamma)^{s}}{\gamma^{d}}}.
\end{align*}

With the above two properties, we only need to prove the following result. 
Let $\delta \in (0,1)$.
For any $(\beta,\gamma,\lambda)$, any integer $d \geq 1$ and $D = d + 1$, if $(\beta,\gamma,\lambda)$ is $d$-unique with gap $\delta$ (formally,~\eqref{eq-proof-claim-assume-1}) and $\lambda \leq (\gamma/\beta)^{D/2}$, then it holds that
\begin{align}
\label{eq:target-single-s} 
\forall \text{ integer } 0 \leq s \leq D - 1, \quad
f_{D}\tp{\frac{\lambda\theta^{\mathrm{sgn}(D, \lambda)}(\beta\gamma)^s}{\gamma^{d}}} =
f_{D}\tp{\frac{\lambda\theta(\beta\gamma)^s}{\gamma^{d}}} \leq 1 - \frac{\delta}{4}.
\end{align}

%

Furthermore, we have the following technical proposition.
\begin{proposition}
\label{lem:lambda-x}
  For integer $d \geq (1 - \delta)\overline{\Delta}$, 
  if $(\beta,\gamma,\lambda)$ is $d$-unique with gap $\delta$, then  
  \begin{align*}
     \frac{\lambda(\beta\gamma)^s\delta}{\gamma^d(2-\delta)} \leq x_1(d),
  \end{align*}
  for any integer $0 \leq s \leq d$, where $\zeta_\delta(d) = d(1 - \beta\gamma) - (1 - \delta)(1 + \beta\gamma)$.
\end{proposition}
The proposition is proved later.

We now prove~\eqref{eq:target-max}. Consider two cases: $D \geq \tp{1 - \frac{\delta}{2}}\overline{\Delta} + 1$ and $D < \tp{1 - \frac{\delta}{2}}\overline{\Delta} + 1$.
\begin{case}[$D \geq \tp{1 - {\delta}/{2}}\overline{\Delta} + 1$]
  In this case, $d = D - 1 \geq (1-\frac{\delta}{2})\overline{\Delta} \geq (1 - \delta)\overline{\Delta}$, we claim
  \begin{align} \label{eq:lb-of-beta-x-gamma-x}
    \beta x_1(d) + \frac{\gamma}{x_1(d)} &\geq \frac{\delta}{2} (1 + \beta\gamma),
  \end{align}
  where $x_1(d)$ is defined in \Cref{prop:uniqueness-property}.
  Since $\theta = \frac{\delta^2}{64} \leq \frac{\delta}{8} \frac{\delta}{2 - \delta}$, we have
  \begin{align*}
    \frac{\lambda\theta(\beta\gamma)^s}{\gamma^d} &\leq \frac{\lambda(\beta\gamma)^s}{\gamma^d} \cdot \frac{\delta}{2 - \delta} \cdot \frac{\delta}{8} \overset{(\star)}{\leq} \frac{\delta}{8} x_1(d) \overset{(*)}{\leq} \sqrt{\gamma/\beta}.
  \end{align*}
  where $(\star)$ is due to \Cref{lem:lambda-x} and $(*)$ holds because $x_1(d)x_2(d) = \gamma/\beta$ and $x_1(d) \leq x_2(d)$, thus $x_1(d) \leq \sqrt{\gamma/\beta}$.
  Furthermore, it is easy to verify that $f_D(x)$ is increasing on $(0, \sqrt{\gamma/\beta}]$. We have
  \begin{align*}
    f_D\tp{\frac{\lambda \theta (\beta\gamma)^s}{\gamma^{d}}} 
    &\le f_D\tp{\frac{\delta}{8}x_1(d)}\\
    &= \frac{D(1-\beta\gamma)}{\frac{\delta}{8}\beta x_1(d) + \frac{8}{\delta}\frac{\gamma}{x_1(d)}+1+\beta\gamma} \\
    &\leq \frac{D(1-\beta\gamma)}{\frac{8}{\delta}\frac{\gamma}{x_1(d)}+1+\beta\gamma} \\
   &\leq \frac{D(1-\beta\gamma)}{\frac{4}{\delta}\tp{\beta x_1(d) + \frac{\gamma}{x_1(d)}}+1+\beta\gamma} 
   	&&(\mbox{since $x_1(d) \leq \sqrt{{\gamma}/{\beta}}$})\\
   &\leq \frac{D(1-\beta\gamma)}{\frac{2}{\delta}\tp{\beta x_1(d) + \frac{\gamma}{x_1(d)}}+1+\beta\gamma+2\tp{\beta x_1(d) + \frac{\gamma}{x_1(d)}}}
   	&& (\text{since }0 < \delta < 1)\\
   &\leq \frac{D(1-\beta\gamma)}{2(\beta x_1(d) + \frac{\gamma}{x_1(d)}+1+\beta\gamma)} 
   	&& (\text{by \eqref{eq:lb-of-beta-x-gamma-x}})\\
    &= \frac{D}{2d} f_d(x_1(d))  \\
    &\leq (1-\delta),
  \end{align*}
  where the last inequality holds because $x_1(d)$ is the root of $f_d(x) = 1 - \delta$ (due to \Cref{prop:uniqueness-property}).
  
  We finish the analysis of this case by verifying \eqref{eq:lb-of-beta-x-gamma-x}.
  Note that it holds that
  \begin{align}\label{eq:proof-lb-of-beta-x-gamma-x-1}
    \frac{1 + \beta\gamma}{d(1 - \beta\gamma)} 
    \leq 
    \frac{(1 + \sqrt{\beta\gamma})^2}{d(1 + \sqrt{\beta\gamma})(1 - \sqrt{\beta\gamma})} 
    = \frac{\overline{\Delta}}{d} \overset{(\ast)}{\leq} \frac{1}{1 - \frac{\delta}{2}},
  \end{align}
    where $(\ast)$ holds because $D \geq \tp{1 - \frac{\delta}{2}}\overline{\Delta} + 1$.
    Moreover, we have
  \begin{align}\label{eq:proof-lb-of-beta-x-gamma-x-2}
    \frac{\beta x_1(d) + \frac{\gamma}{x_1(d)}}{d(1- \beta\gamma)} + \frac{1 + \beta\gamma}{d(1 - \beta\gamma)} &= \frac{1}{f_d(x_1(d))} = \frac{1}{1 - \delta}.
  \end{align}
  Combining \eqref{eq:proof-lb-of-beta-x-gamma-x-1} and \eqref{eq:proof-lb-of-beta-x-gamma-x-2}, we have
  \begin{align*}
    \frac{\beta x_1(d) + \frac{\gamma}{x_1(d)}}{1 + \beta\gamma} + 1\geq \tp{\frac{1}{1-\delta}} \tp{1 - \frac{\delta}{2}} = \frac{\delta}{2-2\delta}+1 \geq \frac{\delta}{2} + 1.
  \end{align*}
\end{case}


\begin{case}[$D < \tp{1 - {\delta}/{2}}\overline{\Delta} + 1$]
  Without loss of generality, we assume $\overline{\Delta} \geq \frac{2}{2 - \delta}$, since otherwise we have
  $D \leq \tp{1 - \frac{\delta}{2}}\overline{\Delta} + 1 < 2$, hence the only possible value for $D$ is $1$. 
  However, in~\eqref{eq:target-max}, we assume $D = d + 1 \geq 2$.

  We claim that the following equations hold in this case.
  \begin{align}
    \forall c \in \tp{0, \frac{1}{2}},\quad f_D\tp{c \sqrt{\frac{\gamma}{\beta}}} &\leq \frac{4c \cdot D}{\overline{\Delta} - 1},  \label{eq:const-factor} \\
    \frac{\lambda(\beta\gamma)^s}{\gamma^d} &\leq \sqrt{\frac{\gamma}{\beta}} \cdot \frac{4}{\delta}. \label{eq:ub-of-input}
  \end{align}
  Note that $\theta = \frac{\delta^2}{64} = \frac{\delta}{16}\cdot \frac{\delta}{4}$.  Then by \eqref{eq:ub-of-input}, it holds that
  \begin{align*}
    \frac{\lambda\theta(\beta\gamma)^s}{\gamma^d} & \leq \frac{\delta}{16}\sqrt{\frac{\gamma}{\beta}} \leq \sqrt{\frac{\gamma}{\beta}}.
  \end{align*}
  It is easy to verify that $f_D(x)$ is monotonically increasing on $(0, \sqrt{\gamma/\beta}]$. 
  It holds that
  \begin{align*}
    f_D\tp{\frac{\lambda\theta(\beta\gamma)^s}{\gamma^d}} 
    &\leq f_D \tp{\frac{\delta}{16}\sqrt{\frac{\gamma}{\beta}}}
    \overset{(\ast)}{\leq} \frac{\delta D}{4(\overline{\Delta} - 1)} 
    \overset{(\star)}{\leq} \frac{\delta}{4}  \cdot \frac{\overline{\Delta} + 1}{\overline{\Delta} - 1} 
    \leq 
    1 - \frac{\delta}{4},
  \end{align*}
where ($\ast$) follows from \eqref{eq:const-factor}, ($\star$) is due to that $D \leq \overline{\Delta} + 1$, and the last inequality holds since  $\overline{\Delta} \geq \frac{2}{2-\delta}$.

  Now, we verify \eqref{eq:const-factor}, it holds that
  \begin{align*}
    f_D\tp{c\sqrt{\frac{\gamma}{\beta}}}
     &\leq \frac{D(1 - \beta\gamma)}{\frac{1}{c} \sqrt{\beta\gamma} + 1 + \beta\gamma}
     = \frac{D(1 - \sqrt{\beta\gamma})}{\frac{(\frac{1}{c} - 2)}{1 + \sqrt{\beta\gamma}} \sqrt{\beta\gamma} + (1 + \sqrt{\beta\gamma})}
     \leq \frac{D(1 - \sqrt{\beta\gamma})}{\frac{(\frac{1}{c} - 2)}{2} \sqrt{\beta\gamma} + (1 + \sqrt{\beta\gamma})},
  \end{align*}
  where the last inequality holds because $\sqrt{\beta \gamma} < 1$ and $c < \frac{1}{2}$,
  which implies 
  \begin{align*}
    f_D\tp{c\sqrt{\frac{\gamma}{\beta}}}
    &\leq \frac{D(1 - \sqrt{\beta\gamma})}{1 + \frac{1}{2c}\sqrt{\beta\gamma}}
     = \frac{2D}{\tp{1 + \frac{1}{2c}}\overline{\Delta} + \tp{1 - \frac{1}{2c}}}
      \leq \frac{2D}{\tp{1 + \frac{1}{2c}}(\overline{\Delta} - 1)}
      \leq \frac{4c \cdot D}{\overline{\Delta} - 1}.
  \end{align*}

  Now we verify \eqref{eq:ub-of-input}.
  It holds that
  \begin{align*}
    \frac{\lambda(\beta\gamma)^s}{\gamma^d}
    &\leq \frac{\lambda}{\gamma^d}
      \overset{(\star)}{\leq} \tp{\frac{\gamma}{\beta}}^{D/2} \frac{1}{\gamma^d}
      = \tp{\frac{\gamma}{\beta}}^{d/2 + 1/2} \frac{1}{\gamma^d}
      = \sqrt{\frac{\gamma}{\beta}}\tp{\frac{1}{\sqrt{\beta\gamma}}}^d ,
  \end{align*}
  where $(\star)$ is due to our assumption  $\lambda \leq (\gamma/\beta)^{D/2}$.
  For $\overline{\Delta} \triangleq \frac{1 + \sqrt{\beta\gamma}}{1 - \sqrt{\beta\gamma}}$, it holds that 
    $\sqrt{\beta\gamma} = \frac{\overline{\Delta} - 1}{\overline{\Delta} + 1}$,
  which implies
  \begin{align*}
    \frac{\lambda(\beta\gamma)^s}{\gamma^d}
    &\leq \sqrt{\frac{\gamma}{\beta}}\tp{\frac{1}{\sqrt{\beta\gamma}}}^d 
    = \sqrt{\frac{\gamma}{\beta}}\tp{\frac{\overline{\Delta} + 1}{\overline{\Delta} - 1}}^d 
    \overset{(\ast)}{\leq} \sqrt{\frac{\gamma}{\beta}}\tp{\frac{\overline{\Delta} + 1}{\overline{\Delta} - 1}}^{\overline{\Delta} \cdot \frac{2-\delta}{2}}
    \overset{(\star)}{\leq} \sqrt{\frac{\gamma}{\beta}} \cdot \frac{4-\delta}{\delta}
    \leq \sqrt{\frac{\gamma}{\beta}} \cdot \frac{4}{\delta},
  \end{align*}
  where 
  $(\ast)$ is due to that $d \leq \tp{1 - \frac{\delta}{2}}\overline{\Delta}$, and
  $(\star)$ is due to that the function $h(t) = \tp{\frac{t+1}{t-1}}^t$ is monotonically decreasing when $t > 1$ and $\overline{\Delta} \geq  \frac{2}{2 - \delta}$.  The monotonicity of function $h(t)$ can be verified as
  \begin{align*}
  (\ln h(t))' &= - \frac{2t}{t^2 - 1} + \ln \tp{\frac{t+1}{t-1}} \to 0 \text{ as } t \to \infty,\\
  (\ln h(t))'' &= \frac{4}{t^2 - 1} > 0.
  \end{align*}
\end{case}
Combining the two cases proves~\eqref{eq:target-max}.
\end{proof}

\begin{proof}[Proof of \Cref{lem:lambda-x}]
Recall that $\zeta_\delta(d) \triangleq d(1 - \beta\gamma) - (1 - \delta)(1 + \beta\gamma)$. By \Cref{prop:uniqueness-property}, 
\begin{align*}
  x_1(d)
  &= \frac{2(1-\delta)\gamma}{\zeta_\delta(d) + \sqrt{\zeta_\delta(d)^2 - 4(1-\delta)^2\beta\gamma}}
    \leq \frac{2(1-\delta)\gamma}{\zeta_\delta(d)}.
\end{align*}
Since $\frac{x + \gamma}{\beta x + 1}$ is increasing in $x$ when $\beta\gamma <1$, it holds that
\begin{align*}
  \frac{x_1(d) + \gamma}{\beta x_1(d) + 1}
  &\leq \frac{\frac{2(1 - \delta)\gamma}{\zeta_\delta(d)} + \gamma}{\frac{2(1 - \delta)\beta\gamma}{\zeta_\delta(d)} + 1}
    = \gamma \cdot \frac{d + (1 - \delta)}{d - (1 - \delta)},
\end{align*}
which, by \Cref{prop:uniqueness-property}, implies
\begin{align*}
  \lambda_1(d)
  &= x_1(d)\tp{\frac{x_1(d) + \gamma}{\beta x_1(d) + 1}}^d
    \leq x_1(d) \cdot \gamma^d \cdot \tp{\frac{d + (1 - \delta)}{d - (1 - \delta)}}^d
    \overset{(\star)}{\leq} x_1(d) \cdot \gamma^d \cdot \frac{2 - \delta}{\delta},
\end{align*}
where $(\star)$ holds because $f(t) = \tp{\frac{t + (1 - \delta)}{t - (1 - \delta)}}^t$ is  decreasing in $t$ when $t \geq 1$ and $d \geq 1$.
And thus,
\begin{align*}
  \lambda_1(d) \frac{(\beta\gamma)^s}{\gamma^d}
  &\leq \frac{\lambda_1(d)}{\gamma^d}
    \leq x_1(d) \cdot \frac{2-\delta}{\delta}. 
\end{align*}
Finally, we note that $f(t)$ is monotonically decreasing because
\begin{align*}
  (\ln f(t))' &= - \frac{2t(1 - \delta)}{(t + (1 - \delta))(t - (1 - \delta))} + \log\tp{\frac{t + (1 - \delta)}{t - (1 - \delta)}} \to 0 \text{ as } t \to +\infty, \\
  (\ln f(t))'' &= \frac{4(1 - \delta)^3}{(t + (1 - \delta))^2(t - (1 - \delta))^2} > 0. \qedhere
\end{align*}
\end{proof}

 \subsection{Hardcore and Ising models}\label{section-hardcore-ising-mixing}
We now prove \Cref{theorem-hardcore} and \Cref{theorem-Ising}, respectively for the hardcore and Ising models.
For these models, the constant $C(\delta)$ can be improved to $\exp(O(1/\delta))$ as the parameter $\theta$ can be chosen to be a universal constant.


\begin{proof}[Proof of \Cref{theorem-hardcore}]
Without loss of generality, we assume that $\lambda \ge \frac{1}{2\Delta}$, and for $\lambda < \frac{1}{2\Delta}$, the standard path coupling technique gives us $O(n \log n)$ mixing time upper bound and $\tp{\frac{1}{2n}}$ spectral gap lower bound.

Choose $\theta=\frac{1}{25}$. We have
  \begin{align*}
    \theta \lambda = \frac{\lambda}{25} < \frac{(\Delta-1)^{\Delta-1}}{25(\Delta-2)^\Delta} < \frac{1}{2\Delta}.
  \end{align*}
Again, by the standard path coupling method, for the hardcore model with fugacity $\theta\lambda < \frac{1}{2\Delta}$, the Glauber dynamics has a $\frac{1}{2n}$ spectral gap lower bound, which further holds up to an arbitrary feasible boundary condition $\sigma_\Lambda \in \Omega(\mu_\Lambda)$ where $\Lambda \subseteq V$.
Therefore
  \begin{align*}
    \spgap{\sgap}{\GD}\tp{\mu^{(1/25)}} \le \frac{1}{2n}.
  \end{align*}
On the other hand, by \Cref{lm:two-spin-property}, the Gibbs distribution $\mu$ is completely $\frac{144}{\delta}$-spectrally independence.
It follows from \Cref{theorem:main} that
\[
\spgap{gap}{\GD}(\mu)\ge\tp{\frac{1}{50}}^{288/\delta+7}\frac{1}{2n}=\frac{1}{C(\delta) n},
\]
for a $C(\delta)=\exp(O(1/\delta))$.
For $\lambda\ge\frac{1}{2\Delta}$, the marginal bound $b \ge \min \left\{\frac{1}{1+\lambda}, \frac{\lambda}{1+\lambda}\right\}\ge\frac{1}{2\Delta+1}$, which implies that $\mu_{\min}\ge b^n\ge\frac{1}{(2\Delta+1)^n}$. The mixing time bound follows from~\eqref{eq:mixing-time-spectral-gap}.
\end{proof}

For the Ising model with edge activity $\beta>0$, without loss of generality, we can assume $\lambda\le1$, because by symmetry, this covers the  $\lambda>1$ case by switching the roles of $\0$ and $\1$ for all vertices.

%
\begin{proof}[Proof of \Cref{theorem-Ising}]
Let $\mu$ be the Gibbs distribution of the Ising model on graph $G=(V,E)$ with edge activity $\beta>0$ and $\lambda\le 1$. 
Assume that $\beta \in \left[\frac{\Delta-2+\delta}{\Delta-\delta},\frac{\Delta-\delta}{\Delta-2+\delta}\right]$, where $\Delta=\Delta_G$ is the maximum degree of $G$.

We first verify that $\mu$ is completely $\frac{4}{\delta}$-spectrally independent.
Consider the Ising model with local fields $\+I'=(V,E,\beta,(\lambda_v)_{v\in V})$ defined on the same graph $G=(V,E)$ with the same edge activity $\beta \in \left[\frac{\Delta-2+\delta}{\Delta-\delta},\frac{\Delta-\delta}{\Delta-2+\delta}\right]$, where each vertex $v\in V$ is associated with an arbitrary local field $\lambda_v\le 1$.

It is straightforward to verify that the trivial potential function $\phi(y)=1$ is a $(\delta,2)$-potential function with respect to $\+I'$ (\Cref{def:potential}): for $h(y)=-\frac{(1-\beta^2)\mathrm{e}^y}{(\beta\mathrm{e}^y+1)(\mathrm{e}^y+\beta)}$,
  \begin{enumerate}
    \item ($\delta$-contraction) for any $1 \le d \le \Delta-1$ and $y_1,\ldots,y_d \in [-\infty,+\infty]$,
    \begin{align*}
      \sum_{i=1}^d \abs{h(y_i)} = \sum_{i=1}^d \frac{\abs{1-\beta^2}\mathrm{e}^{y_i}}{(\beta \mathrm{e}^{y_i}+1)(\mathrm{e}^{y_i}+\beta)} \le \frac{(\Delta-1)\abs{1-\beta}}{1+\beta} \le 1-\delta;
    \end{align*}
    \item (2-boundedness) similarly, for all $y\in[-\infty,+\infty]$, $\abs{h(y)} =\frac{(1-\beta^2)\mathrm{e}^y}{(\beta\mathrm{e}^y+1)(\mathrm{e}^y+\beta)}\le\frac{|1-\beta|}{1+\beta}\le \frac{1}{\Delta-1} \le \frac{2}{\Delta}$.
  \end{enumerate}
By \Cref{theorem-good-potential-imply-SI-CLV}, the Gibbs distribution of such $\+I'$ is always $\frac{4}{\delta}$-spectrally independent, which means that $\mu$ is completely $\frac{4}{\delta}$-spectrally independent.

On the other hand, the standard path coupling method gives an $O(n \log n)$ mixing time upper bound and a $\frac{1}{2n}$ spectral gap lower bound for $\lambda \le \frac{1}{500}$, which holds up to  an arbitrary feasible boundary condition $\sigma_\Lambda \in \Omega(\mu_\Lambda)$ where $\Lambda \subseteq V$.
Choose $\theta = \frac{1}{500}$. We have
  \begin{align*}
    \spgap{\sgap}{\GD}\tp{\mu^{(1/500)}} \le \frac{1}{2n}.
  \end{align*}  
It follows from \Cref{theorem:main} that
\[
\spgap{gap}{\GD}(\mu)\ge\tp{\frac{1}{1000}}^{8/\delta+7}\frac{1}{2n}=\frac{1}{C(\delta) n},
\]
for a $C(\delta)=\exp(O(1/\delta))$.
For $\beta \in \left[\frac{\Delta-2+\delta}{\Delta-\delta},\frac{\Delta-\delta}{\Delta-2+\delta}\right]$ and $\frac{1}{500}\le \lambda\le 1$, the marginal bound $b\ge \frac{1}{14000}$~\cite{chen2020optimal}, which implies that $\mu_{\min}= b^n\ge\tp{\frac{1}{14000}}^n$.
The mixing time bound follows from~\eqref{eq:mixing-time-spectral-gap}.
\end{proof}

\bibliographystyle{alpha}
\bibliography{2-spin}

\newpage
\appendix

\section{Mixing of Uniform Block Dynamics from Spectral Independence}
\label{appendix-block-mixing}
Here we give a short proof of \Cref{theorem:spectral-independent-imply-spectral-gap} using theorems and lemmas proved in~\cite{chen2020optimal,feng2021rapid}.
Let $V$ be a ground set of size $n$.
For $\eta > 0$, $\mu$ is an $\eta$-spectrally independent distribution over $\{\0, \1\}^V$.
For $2\lceil \eta \rceil\le \ell \leq n$, we want to show that the spectral gap of the uniform $\ell$-block dynamics for $\mu$ satisfies
\begin{align*}
  \spgap{gap}{}(P_\ell) &\geq \tp{\frac{\ell}{2n}}^{2\lceil \eta \rceil + 1}.
\end{align*}

For any subset $\Lambda \subset V$, any feasible partial configuration $\sigma_\Lambda \in \Omega(\mu_\Lambda)$, define the \emph{local random walk} $P_{\sigma_\Lambda}$ on $U_{\sigma_\Lambda} = \left\{(u, c) \in (V\setminus\Lambda) \times \{\0,\1\} \;|\; \mu^{\sigma_\Lambda}_u(c) > 0\right\}$ as
\begin{align*}
  \forall (u, i), (v, j) \in U_{\sigma_\Lambda}, \quad P_{\sigma_\Lambda}((u, i), (v, j)) &\triangleq \frac{\*1[u\not= v]}{\abs{V} - \abs{\Lambda} - 1} \mu^{\sigma_\Lambda, u \gets i}_v (j),
\end{align*}
where 
$\mu^{\sigma_\Lambda, u \gets i}_v(j)$ denotes the marginal distribution at $v$ induced from $\mu$ conditioned on the configuration on $\Lambda$ being fixed as $\sigma_\Lambda$ and the value of $u$ being fixed as $i$.

Then, for $\eta > 0$, we say $\mu$ is a \emph{$(\zeta_0, \cdots, \zeta_{n-2})$-local spectral expander} if for every $0 \leq k \leq n - 2$, every $\Lambda \subseteq V$ such that $\abs{\Lambda} = k$, and every feasible partial configuration $\sigma_\Lambda \in \Omega(\mu_\Lambda)$, we have
\begin{align*}
  \lambda_2(P_{\sigma_\Lambda}) \leq \zeta_k.
\end{align*}

\begin{theorem}[\text{\cite[Theorem A.9]{chen2020optimal}}] \label{theorem:local-spectral-expander-to-spectral-gap}
  Let $\zeta_0,\zeta_1,\ldots,\zeta_{n-2} \in [0,1]$ be a sequence of positive real numbers. 
  If $\mu$ is a $(\zeta_0, \cdots, \zeta_{n-2})$-local spectral expander, then for any $1\leq \ell \leq n$, it holds that
  \begin{align*}
    \spgap{gap}{}(P_\ell) &\geq \frac{\sum_{k = n - \ell}^{n-1}\Gamma_k}{\sum_{k=0}^{n-1}\Gamma_k},
  \end{align*}
  where $\Gamma_k \triangleq \prod_{j=0}^{k-1}\frac{1 - \zeta_j}{1 + \zeta_j}$ for $k > 0$ and $\Gamma_0 \triangleq  1$.
\end{theorem}

The following lemma from \cite{feng2021rapid} assumes our notion of spectral independence defined using the absolute influence matrices $\Psi_\mu^{\sigma_{\Lambda}}$ (\Cref{definition-weight-tot-inf}).

\begin{lemma}[\text{\cite[Lemma 3.6]{feng2021rapid}}] \label{lemma:SI-to-local-spectral-expander}
  For $\eta > 0$, if $\mu$ is $\eta$-spectrally independent, then there is a sequence $\zeta_0, \zeta_1, \cdots, \zeta_{n-2}\in[0,1]$, such that $\mu$ is a $(\zeta_0, \cdots, \zeta_{n-2})$-local spectral expander, where for $0 \leq k \leq n - 2$,
  \begin{align}\label{eq:local-spectral-expander-zeta-eta}
    \zeta_k = \min \left\{1 , \frac{\eta}{n - k - 1}\right\}.
  \end{align}
\end{lemma}

\begin{proof}[Proof of \Cref{theorem:spectral-independent-imply-spectral-gap}]
  Let $C \triangleq \left\lceil \eta \right \rceil$. 
  Note that $\mu$ is $C$-spectrally independent assuming that $\mu$ is $\eta$-spectrally independent.
  By \Cref{lemma:SI-to-local-spectral-expander}, for the $\zeta_k$'s given in \eqref{eq:local-spectral-expander-zeta-eta}, $\mu$ is a $(\zeta_0, \cdots, \zeta_{n-2})$-local spectral expander.
  By some calculation, one can verify that for any $1 \leq k \leq n-1$, it always holds that 
  \[
  \Gamma_k 
  = 
  \prod_{j=0}^{k-1}\frac{1 - \zeta_j}{1 + \zeta_j}
  =
  \frac{(n-k-C)(n-k-C+1)\ldots(n-k+C-1)}{(n-C)(n + 1 - C) \cdots (n - 1 + C)}.
  \]
  Therefore, by \Cref{theorem:local-spectral-expander-to-spectral-gap}, 
  \begin{align*}
    \spgap{gap}{}(P_\ell) \geq \frac{\sum_{k=n-\ell}^{n-1}\Gamma_k}{\sum_{k=0}^{n-1}\Gamma_k} \overset{(\star)}{=} \prod_{k = -C}^{C}\frac{\ell+k}{n+k} \ge \left(\frac{\ell-C}{n-C}\right)^{2C+1} \ge \left(\frac{\ell}{2n}\right)^{2C+1},  \end{align*}  
  where we use the fact that $n \geq \ell \ge 2C$, and $(\star)$ is due to the equation $\sum_{j=0}^{N-1}\binom{j}{k}=\binom{N}{k+1}$.
\end{proof}

\section{Uniqueness and Spectral Independence for Two-Spin Systems with Local Fields}\label{section-uniqueness-SI}
In this section, we give missing proofs in \Cref{sec:req-1}.
Specifically, they are the proofs of \Cref{prop:uniqueness-property} (in \Cref{proof-prop:uniqueness-property}), \Cref{theorem-good-potential-imply-SI-CLV} (in \Cref{proof-theorem-good-potential-imply-SI-CLV}), and \Cref{lm:contraction} (in \Cref{proof-lm:contraction}).
All these theorems have been proved in some variant forms previously.
But here our goal is to reestablish them for the two-spin systems with the local fields that are biased towards the easier directions as indicated in~\eqref{eq:good-direction}. 
In principle, such alternatton should only make these theorems more satisfied.
However, to formally verify this, we have to reiterate their respective existing proofs.

\subsection{Proof of \Cref{prop:uniqueness-property}}\label{proof-prop:uniqueness-property}
  When $\beta = 0$, recall that $(\beta,\gamma,\lambda)$ is $d$-unique if and only if
  \begin{align}\label{eq:hardconstraint}
    f_d(\hat{x}_d) = \frac{d(1-\beta \gamma)\hat{x}_d}{(\beta \hat{x}_d+1)(\gamma+\hat{x}_d)} = \frac{d\hat{x}_d}{\gamma + \hat{x}_d} \leq 1-\delta,
  \end{align}
  where $\hat{x}_d$ is the fixed point of function $F_d(x) = \lambda \tp{\frac{\beta x + 1}{\gamma + x}}^d = \lambda \tp{\frac{1}{\gamma + x}}^d$, namely $F_d(\hat{x}_d)=\hat{x}_d$.
  
  Note that \eqref{eq:hardconstraint} holds if and only if $\hat{x}_d \leq \frac{(1-\delta) \gamma}{d-1+\delta}$, which is equivalent to 
  \begin{align*}
    \lambda &\leq \hat{x}_d \tp{\frac{\hat{x}_d + \gamma}{\beta\hat{x}_d + 1}}^d = \hat{x}_d\tp{\hat{x}_d + \gamma}^d \leq \frac{(1-\delta) d^d \gamma^{d+1}}{(d-1+\delta)^{d+1}}.
  \end{align*}

  We assume that $\beta > 0$. Similar to previous discussion, $(\beta,\gamma,\lambda)$ is $d$-unique if and only if
  \begin{align}\label{eq:general-constraint}
    f_d(\hat{x}_d) = \frac{d(1-\beta \gamma)\hat{x}_d}{(\beta \hat{x}_d+1)(\gamma+\hat{x}_d)} \leq 1-\delta,
  \end{align}
  where $\hat{x}_d$ is the fixed point of function $F_d(x) = \lambda \tp{\frac{\beta x + 1}{\gamma + x}}^d$, namely $F_d(\hat{x}_d)=\hat{x}_d$.

  When $d < (1-\delta) \overline{\Delta}$, the following holds for any $x > 0$.
  \begin{align*}
    f_d(x) = \frac{d(1-\beta\gamma)x}{(\beta x+1)(x+\gamma)} = \frac{d(1-\beta \gamma)}{\beta x+\frac{\gamma}{x} + 1+\beta\gamma} \le \frac{d(1-\beta \gamma)}{1+\beta \gamma + 2\sqrt{\beta \gamma}} = \frac{d(1-\sqrt{\beta \gamma})}{1+\sqrt{\beta \gamma}} < 1-\delta.
  \end{align*}
  Therefore, $(\beta,\gamma,\lambda)$ is $d$-unique with gap $\delta$ if $d < (1-\delta) \overline{\Delta}$.

  If $d \ge (1-\delta) \overline{\Delta}$, it can be verified that equation $\frac{d(1-\beta \gamma)x}{(\beta x+1)(x+\gamma)} = 1-\delta$ has two positive roots
  \begin{align*}
    x_1(d) &= \frac{\zeta_\delta(d) - \sqrt{\zeta_\delta(d)^2 - 4(1 - \delta)^2\beta\gamma}}{2(1- \delta)\beta}
    \quad \mbox{and} \quad
    x_2(d) = \frac{\zeta_\delta(d) + \sqrt{\zeta_\delta(d)^2 - 4(1 - \delta)^2\beta\gamma}}{2(1- \delta)\beta},
  \end{align*}
  and \eqref{eq:general-constraint} holds if and only if $\hat{x}_d\leq x_1(d)$ or $\hat{x}_d \geq x_2(d)$.
  Note that $x\tp{\frac{x + \gamma}{\beta x + 1}}^d$ is monotone increasing in $x$ for any fixed $d$.
  So, $(\beta,\gamma,\lambda)$ is $d$-unique with gap $\delta$ if and only if
  \begin{align*}
    \lambda \in (0,\lambda_1(d)] \cup [\lambda_2(d),+\infty),
  \end{align*}
  where $\lambda_i(d) = x_i(d) \tp{\frac{x_i(d) + \gamma}{\beta x_i(d)+1}}^d$, $i \in \{1,2\}$.

  Lastly, $x_1(d)x_2(d) = \frac{\gamma}{\beta}$ by Vieta's formula. Therefore,
  \begin{align*}
    \lambda_1(d) \lambda_2(d) &= x_1(d) x_2(d)\tp{\frac{(\gamma + x_1(d))(\gamma + x_2(d))}{(\beta x_1(d) + 1)(\beta x_2(d)+1)}}^d\\
    &= x_1(d) x_2(d) \tp{\frac{\gamma^2 + \gamma(x_1(d)+x_2(d))+x_1(d)x_2(d)}{\beta^2 x_1(d)x_2(d)+\beta(x_1(d)+x_2(d))+1}}^d\\ 
    &= \frac{\gamma}{\beta} \tp{\frac{\gamma^2+\frac{\gamma}{\beta}+\gamma(x_1(d)+x_2(d))}{\gamma \beta + 1 + \beta(x_1(d)+x_2(d))}}^d
    =\tp{\frac{\gamma}{\beta}}^{d+1}.
  \end{align*}


\subsection{Proof of \Cref{theorem-good-potential-imply-SI-CLV}}\label{proof-theorem-good-potential-imply-SI-CLV}
Assuming the existence of $(\alpha,c)$-potential function $\phi$ with respect to a two-spin system $\+I=(V,E,\beta,\gamma,(\lambda_v)_{v \in V})$, we prove that its Gibbs distribution $\pi$ is $\frac{2c}{\alpha}$-spectrally independent.
By \Cref{definition-weight-tot-inf}, 
it is sufficient to bound the the spectral radius of the influence matrix $\Psi^{\sigma_\Lambda}_\pi$ for any feasible partial configuration $\sigma_\Lambda \in \Omega(\pi_\Lambda)$ specified on any subset $\Lambda \subset V$.

Let $D\in\mathbb{R}^{V\times V}$ be the diagonal matrix with $D(v,v)=\Delta_v$ for all $V$. 
By similarity, $D^{-1}\Psi^{\sigma_\Lambda}_\pi D$ has the same eigenvalues as $\Psi^{\sigma_\Lambda}_\pi$, and 
\[\rho(\Psi^{\sigma_\Lambda}_\pi) = \rho(D^{-1}\Psi^{\sigma_\Lambda}_\pi D) \leq \norm{D^{-1}\Psi^{\sigma_\Lambda}_\pi D}_{\infty}.\] 
The next lemma follows immediately.
\begin{lemma}
\label{lem:spec-ratio-to-row-sum}
Let $\Lambda \subset V$ and $\sigma_\Lambda \in \Omega(\pi_\Lambda)$. 
  Let $\eta > 0$.
  If for every $r \in V\setminus \Lambda$,
   \begin{align} \label{eq:WInf}
    \sum_{v \in V\setminus \Lambda} \Delta_v \cdot \Psi^{\sigma_\Lambda}_\pi(r, v) &\leq \eta \cdot \Delta_r,
  \end{align}
  then it holds that $\rho(\Psi^{\sigma_\Lambda}_\pi) \leq \eta$.
\end{lemma}
Fix $\Lambda \subset V$ and $\sigma_\Lambda \in \Omega(\mu_\Lambda)$.  
To bound $\rho(\Psi^{\sigma_\Lambda}_\pi)$,
we only need to bound the weighted row sum of $\Psi^{\sigma_\Lambda}_\pi$ given by~\eqref{eq:WInf} (which corresponds to weighted total influence in $\pi^{\sigma_\Lambda}$). 
This follows from a two-step argument:
\begin{itemize}
\item First, show that there is a two-spin system $\+I_T$ defined on a SAW-tree rooted at $r$ that preserves the influence from $r$ to other vertices in $\+I$ (\Cref{sec:B-1}).
\item Next, show that for any two-spin system $\+I_T$ on a a rooted tree, the $\alpha$-contraction and the $c$-boundedness of $\phi$ implies that the weighted total influence from the root to all other vertices in $\+I_T$ is always bounded by $\frac{2c}{\alpha} \cdot \Delta_r$ (\Cref{sec:B-2}).
\end{itemize}



\subsubsection{Self-avoiding walk tree} \label{sec:B-1}

\newcommand{\SAW}{\text{\tiny SAW}}
Let $\+I = (V, E, \beta, \gamma, \tp{\lambda_v}_{v \in V})$ be a  two-spin system on a connected graph $G = (V, E)$ with Gibbs distribution $\pi$.
Assume a total ordering among all vertices in $G$.
Given a vertex $r \in V$, let $T = \TSAW(G, r) = (V_\SAW, E_\SAW)$ be the \emph{self-avoiding walk (SAW) tree} rooted at $r$. The unique path between $r$ and each vertex $u$ in $T$ is either a self-avoiding walk that ends at $u$ or a self-avoiding walk with $u$ being a cycle-closing vertex.  
The vertices in $T$ enumerate all such walks in $G$ and $u$ is $v$'s parent in $T$ iff $v$'s walk extends $u$'s.
The two-spin system $\+I_T$ on $T$ is defined as follow:
\begin{enumerate}
\item 
A boundary condition is imposed to $\+I_T$ on every cycle-closing leaf, such that a $\1$ is imposed on such a leaf if the cycle is formed from a smaller vertex to a larger vertex in the total ordering (and a $\0$ is imposed if otherwise). 
Denote this cycle-closing boundary condition on $\+I_T$ as $\tau_\SAW$.
\item $\+I_T$ has the same edge interactions $\beta, \gamma$ as $\+I$, and for every free vertex $v$ in $T$, the local field of $\hat{v}$ is set to be $\lambda_v$ if $\hat{v}$ is a copy of $v$ in $T$.
\end{enumerate}
We denote the Gibbs distribution of $\+I_T$ as $\pi_\SAW$.
For each $v \in V$, we denote by $\+C_v$ the set of all free (not fixed by $\tau$) copies of $v$ in $\+I_T$.
For any partial configuration $\sigma_\Lambda \in \Omega(\pi_\Lambda)$ specified on a subset $\Lambda \subset V$, a corresponding partial configuration $\sigma_\SAW$can be constructed in $\+I_T$ by assigning $\sigma_v$ to each copy $\hat{v}$ in $\+C_v$ for any $v \in V$.
Together, $\sigma_\SAW \uplus \tau_\SAW$ forms the boundary condition on $\Lambda_\SAW$ that is imposed on $\+I_T$.
With slightly abuse of the notation, we also use $\sigma_\Lambda$ to denote this boundary condition $\sigma_\SAW \uplus \tau_\SAW$ on $\+I_T$.
We use $\Psi^{\sigma_\Lambda}_G$ and  $\Psi^{\sigma_\Lambda}_T$ to denote $\Psi^{\sigma_\Lambda}_\pi$ and $\Psi^{\sigma_\SAW \uplus \tau_\SAW}_{\pi_\SAW}$ respectively, when $T = \TSAW(G, r)$ and $\+I_T$ are generated according to the above construction.

The following signed influence is introduced in~\cite{anari2020spectral}.
\begin{definition}[pairwise signed influence]\label{definition-signed-influence-matrix}
  Let $\+I = (V, E, \beta, \gamma, \tp{\lambda_v}_{v \in V})$ be a two-spin system on $G = (V, E)$ with Gibbs distribution $\pi$.
  Let $\Lambda \subset V$ and $\sigma_\Lambda \in \Omega(\mu_\Lambda)$.
  The \emph{pairwise signed influence} is defined as follows.
  For every $u, v \in V$ such that $u \not= v$, let 
  \begin{align*}
    \ALO{\pi}{\sigma_\Lambda}(u, v)
    &\triangleq
      \begin{cases}
        \pi^{\sigma_\Lambda, u \gets \1}_v(\1) - \pi^{\sigma_\Lambda, u \gets \0}_v(\1) & \text{if }\Omega(\mu^{\sigma_\Lambda}_u) = \{\0,\1\}, \\
        0 & \text{otherwise,}
    \end{cases}
  \end{align*}
  and for every $u \in V$, let $\ALO{\pi}{\sigma_\Lambda}(u, u) = 0$.
\end{definition}
Here, for convenience, we also use $\ALO{G}{\sigma_\Lambda}$ and $\ALO{T}{\sigma_\Lambda}$ to denote $\ALO{\pi}{\sigma_\Lambda}$ and $\ALO{\pi_\SAW}{\sigma_\Lambda}$, respectively.

The following lemma is taken from \cite{chen2020rapid}, which was proved for two-spin systems with local fields.
\begin{lemma}[\text{\cite[Lemma 8]{chen2020rapid}}] \label{lem:tree-inf-preserve}
  Let $\+I = (V, E, \beta, \gamma, \tp{\lambda_v}_{v \in V})$ be a two-spin system on $G = (V, E)$ with Gibbs distribution $\pi$, where $0 \leq \beta \leq \gamma$, $\gamma > 0$, and $\lambda_v > 0$ for all $v \in V$.
  Let $r \in V$ be a vertex, and $\+I_T$ the two-spin system obtained from $\+I$ according to the SAW-tree transformation $\TSAW(G, r)$.
  Let $\Lambda \subset V$ and $\sigma_\Lambda \in \Omega(\mu_\Lambda)$.
 For any $u \in V\setminus \Lambda$ , it holds that
  \begin{align*}
    \ALO{G}{\sigma_\Lambda} (r, u) &= \sum_{\hat{u} \in \+C_u} \ALO{T}{\sigma_\Lambda} (r, \hat{u}).
  \end{align*}
\end{lemma}

Let $\Lambda \subset V$, $\sigma_\Lambda \in \Omega(\mu_\Lambda)$, and $r \in V\setminus \Lambda$. By the construction of $\+I_T$ and \Cref{lem:tree-inf-preserve}, the followings hold:
\begin{itemize}
\item for any $v \in V$, every free copy $\hat{v}$ in $\+I_T$ has the same degree and local field as $v$ in $\+I$;
\item it holds that
\begin{align*}
  \sum_{v \in V\setminus \Lambda} \Delta_{G,v} \cdot \Psi^{\sigma_\Lambda}_G(r, v)
  &= \sum_{v \in V\setminus \Lambda} \Delta_{G,v} \cdot \abs{\ALO{G}{\sigma_\Lambda}(r, v)} 
  = \sum_{v\in V\setminus \Lambda} \Delta_{G,v} \cdot \abs{\sum_{\hat{v} \in \+C_u}  \ALO{T}{\sigma_\Lambda}(r, \hat{v})} \\
  &\leq \sum_{v\in V\setminus \Lambda} \sum_{\hat{v} \in \+C_u} \Delta_{G,v} \cdot \abs{\ALO{T}{\sigma_\Lambda}(r, \hat{v})} 
  = \sum_{\hat{v} \in V_\SAW \setminus \Lambda_\SAW} \Delta_{T, \hat{v}} \cdot \Psi^{\sigma_\Lambda}_T (r, \hat{v}),
\end{align*}
where $\Delta_{G,v}$, $\Delta_{T,\hat{v}}$ denote the degree of $v$ in $G$ and $\hat{v}$ in tree $T=\TSAW(G, r)$ respectively, and $V_\SAW \setminus \Lambda_\SAW$ is the set all free vertices in $\+I_T$.
\end{itemize}
In order to get an upper bound for the weighted row sum of $\Psi^{\sigma_\Lambda}_G$ with respect to $r$, 
it is then sufficient for us to get an upper bound for the weighted row sum of $\Psi^{\sigma_\Lambda}_T$ with respect to $r$.


\subsubsection{Influence on trees} \label{sec:B-2}
We now bound the total influence on trees by assuming $(\alpha,c)$-potential function.
To work with SAW trees with boundary conditions, we slightly generalize the definition of $(\alpha, c)$-potential function in \Cref{def:potential}.
Let $\+I=(V,E,\beta,\gamma,(\lambda_v)_{v\in V})$ be a two-spin system with local fields.
Let $\phi:[-\infty,+\infty] \rightarrow [0, +\infty)$ be a function such that $\phi(y) > 0$ for any $y \in [-\infty,+\infty]$ that $|h(y)|>0$, where the function $h$ is defined in~\eqref{eq-def-function-h}.
Let  $S \subseteq V$ be a subset.
For any $\alpha\in(0,1)$ and $c>0$, we say $\phi$ is an \emph{$(\alpha,c)$-potential function with respect to $\+I$ on $S$} if it satisfies the following conditions: 
\begin{enumerate}
\item ($\alpha$-Contraction) 
For every $v\in S$ with $d_v \geq 1$ and every $(y_1,\ldots,y_{d_v})\in [-\infty, +\infty]^{d_v}$, we have
\[
 \phi(y) \sum_{i=1}^{d_v} h^{\phi}(y_i) \le 1-\alpha,
\]
where $y=H_{\lambda_v,d_v}(y_1,\ldots,y_d)$.
\item ($c$-Boundedness)
For every $u, v \in S$, every $y_u \in J_{\lambda_u,d_u}$ and $y_v \in J_{\lambda_v,d_v}$, we have
\begin{align*}
      \phi(y_v)\cdot h^\phi(y_u) \leq \frac{2c}{\Delta_{u} + \Delta_{v}}.
\end{align*}
\end{enumerate}
Recall that $h^{\phi}:[-\infty,+\infty]\to[0,+\infty)$ is defined as that for any $y\in [-\infty,+\infty]$, $h^{\phi}(y)=0$ if $h(y)=0$, and if $h(y)\neq0$,
\[
h^{\phi}(y)= 
\frac{|h(y)|}{\phi(y)}. 
\]
The only difference from  \Cref{def:potential} is that in the above definition,  the $\alpha$-Contraction and $c$-Boundedness properties are required to hold on a subset $S$. We will prove the following lemma.

\begin{lemma} \label{lem:inf-on-tree-instance}
  Let $\+I = (V, E, \beta, \gamma, \tp{\lambda_v}_{v \in V})$ be a two-spin system on a tree $T = (V, E)$ with Gibbs distribution $\pi$, 
  where $0 \leq \beta \leq \gamma$, $\gamma > 0$, and $\lambda_v > 0$ for all $v \in V$.
    Let $r\in V$, $\Lambda \subseteq V \setminus \{r\}$ and $\sigma_\Lambda \in \Omega(\pi_\Lambda)$.
    For any $\alpha \in (0, 1)$ and $c > 0$,
  if there is an $(\alpha, c)$-potential function $\phi$ with respect to $\+I$ on $V\setminus \Lambda$, then it holds that
  \begin{align*}
    \sum_{v \in V\setminus \Lambda} \Delta_v \cdot \Psi^{\sigma_\Lambda}_\pi(r, v) &\leq \Delta_r \cdot \frac{2c}{\alpha}.
  \end{align*}
\end{lemma}

Fix $\Lambda \subseteq V \setminus \{r\}$, and $\sigma_\Lambda \in \Omega(\mu_\Lambda)$. 
Without loss of generality, we assume that $\Lambda$ only contains leaves of $T$, since for any $v \in \Lambda$, due to conditional independence the descendants of such $v$ can be removed without affecting the influence from the root.
We use $T_v$ to denote the subtree of $T$ rooted at vertex $v \in V$, and $L_v(k) \subseteq V \setminus \Lambda$ for the set of all \emph{free} vertices at distance $k$ away from $v$ in the subtree $T_v$.
Moreover, we define
\begin{align*}
\Upsilon^{\sigma_\Lambda}_\pi \triangleq \Psi^{\sigma_\Lambda}_\pi + I,	
\end{align*}
where $I$ denotes the identity matrix. We will prove the following lemma.

\begin{lemma} \label{lem:tree-decay-level}
  Let $c > 0$ and $\alpha \in (0, 1)$ be two real numbers.
  If there is an $(\alpha, c)$-potential function $\phi$ with respect to $\+I$ on $V\setminus \Lambda$, then for any integer $k \geq 1$, it holds that
  \begin{align*}
    \sum_{u \in L_r(k)} \Delta_u \Psi^{\sigma_\Lambda}_\pi (r, u) &\leq \sum_{u \in L_r(k)} \Delta_u \Upsilon^{\sigma_\Lambda}_\pi (r, u) \leq 2c(1 - \alpha)^{k-1} \Delta_r.
  \end{align*}
\end{lemma}
\Cref{lem:inf-on-tree-instance} is an easy consequence of \Cref{lem:tree-decay-level}.
\begin{proof}[Proof of \Cref{lem:inf-on-tree-instance}]
Note that $\Psi^{\sigma_\Lambda}_T(r,r) = 0$. It holds that
\begin{align*}
  \sum_{u \in V_T} \Delta_{u} \cdot \Psi^{\sigma_\Lambda}_T (r, u)
  &= \sum_{k = 1}^{+\infty} \sum_{u \in L_r(k)} \Delta_u \cdot \Psi^{\sigma_\Lambda}_T (r, u)
  {\leq} \Delta_r \cdot \sum_{k = 1}^{\infty} 2c(1 - \alpha)^{k-1}
  = \frac{2c}{\alpha} \Delta_r,
\end{align*}
where the inequality is due to \Cref{lem:tree-decay-level}.
\end{proof}

Let $u \in V$ be a vertex in the rooted tree $T$.
Let $\+I_u=\+T_{T_u}$ denote the two-spin system induced by $\+I_T$ on the subtree $T_u$ rooted by $u$.
Let $\nu$ be the Gibbs distribution associated with $\+I_u$.
Let $\tau$ denote the configuration obtained by restricting $\sigma_\Lambda$ on $T_u$, formally, $\tau = \sigma_{\Lambda}(V_{T_u})$, where $V_{T_u}$ denotes all vertices in $T_u$.
We define the \emph{marginal ratio} at $u$ by
\begin{align*}
  R_u = \frac{\nu^{\tau}_u(\1)}{\nu^{\tau}_u(\0)}.
\end{align*}
It is well know~\cite{Wei06} that each $R_u$ can be calculated by the following tree recursion.
For a vertex $u \in V$ whose local field is $\lambda_u$, denote the children of $u$ as $w_1, w_2, \cdots, w_{d_u}$, where $d_u$ is the number of $u$'s children in $T_u$.
The recursion of the marginal ratio is defined explicitly as follow:
\begin{align*}
  R_u &= \begin{cases}
    0 & u \in \Lambda, \sigma_u = \0 ; \\
    +\infty & u \in \Lambda, \sigma_u = \1 ; \\
    F_{\lambda_u, d_u}(R_{w_1}, R_{w_2}, \cdots, R_{w_{d_u}}) & u\not\in \Lambda,
  \end{cases}
\end{align*}
where we recall that  the tree recursion $F_{\lambda_u, d_u}$ is defined as:
\begin{align*}
  F_{\lambda_u, d_u}\tp{R_{w_1}, R_{w_2}, \cdots, R_{w_{d_u}}} &= \lambda_u \prod_{i = 1}^{d_u} \frac{\beta R_{w_i} + 1}{R_{w_i} + \gamma}.
\end{align*}

Now, we prove \Cref{lem:tree-decay-level}.
Recall that for two vertices $u, v \in V \setminus \Lambda$, if $u\not= v$, then
\begin{align*}
  \Upsilon^{\sigma_\Lambda}_\pi(u, v) = \Psi^{\sigma_\Lambda}_\pi(u, v) = \max_{i, j \in \Omega(\pi^{\sigma_\Lambda}_u)} \DTV{\pi^{\sigma_\Lambda, u \gets i}_v}{\pi^{\sigma_\Lambda, u \gets j}_v}
\end{align*}
gives the influence from $u$ to $v$; and $\Upsilon_\pi^{\sigma_\Lambda}(w,w) = 1$ for any $w \in V \setminus \Lambda$.

The following lemma is taken from \cite{anari2020spectral}, which can be straightforwardly extended to two-spin systems with local fields.
\begin{lemma}[\text{\cite[Lemma B.2]{anari2020spectral}}]
Let $u, v, w \in V \setminus \Lambda$ be distinct vertices in tree $T$ such that $u$ is on the unique path from $v$ to $w$.
The signed influence matrix $\ALO{\pi}{\sigma_\Lambda}$ defined in \Cref{definition-signed-influence-matrix} satisfies
  \begin{align*}
   \ALO{\pi}{\sigma_\Lambda}(v, w) &= \ALO{\pi}{\sigma_\Lambda}(v, u) \cdot \ALO{\pi}{\sigma_\Lambda}(u, w).
  \end{align*}	
\end{lemma}

Note that for any $u,v \in V \setminus \Lambda$ with $u \neq v$, it holds that $\Upsilon^{\sigma_\Lambda}_\pi = \abs{ \ALO{\pi}{\sigma_\Lambda}(u, v) }$, and for any $ w \in V\setminus \Lambda$, it holds that  $\Upsilon^{\sigma_\Lambda}_\pi(w, w) = 1$.
\begin{corollary}\label{lem:inf-chain}
Let $u, v, w \in V \setminus \Lambda$ be (not necessarily distinct) vertices in tree $T$ such that $u$ is on the unique path from $v$ to $w$.
It holds that
  \begin{align*}
    \Upsilon^{\sigma_\Lambda}_\pi(v, w) &= \Upsilon^{\sigma_\Lambda}_\pi(v, u) \cdot \Upsilon^{\sigma_\Lambda}_\pi(u, w).
  \end{align*}	
\end{corollary}
%
 %
 The following lemma \cite{chen2020rapid} was proved for two-spin systems with local fields.
\begin{lemma} [\text{\cite[Lemma 16]{chen2020rapid}}] \label{lem:inf-near}
  Let $v \in V \setminus \Lambda$ and $u \in V \setminus \Lambda$  a child of $v$ in the subtree $T_v$.
  If $|\Omega(\pi^{\sigma_\Lambda}_v)| = 1$, then $\Upsilon^{\sigma_\Lambda}_\pi(v, u) = 0$; otherwise, it holds that
  \begin{align*}
    \Upsilon^{\sigma_\Lambda}_\pi(v, u) =|I^{\sigma_\Lambda}_\pi(v, u)|= \abs{h(\log R_u)},
  \end{align*}
  where  $h(y)= - \frac{(1 - \beta\gamma)\mathrm{e}^y}{(\beta \mathrm{e}^y + 1)(\mathrm{e}^y + \gamma)}$. 
\end{lemma}



Now, we are ready to prove \Cref{lem:tree-decay-level}.
\begin{proof} [Proof of \Cref{lem:tree-decay-level}]
Suppose there is a $(\alpha, c)$-potential function $\phi$ with respect to $\+I$ on $V\setminus \Lambda$.
%
%
We denote the children of $r$ by $u_1, u_2, \cdots, u_{\Delta_r}$. Then by \Cref{lem:inf-chain} and \Cref{lem:inf-near}, for any integer $k \geq 1$, 
  \begin{align}
  \label{eq-agu-0/0}
    \sum_{v \in L_r(k)} \Delta_v \cdot \Upsilon^{\sigma_\Lambda}_\pi(r, v)
    &= \sum_{i = 1}^{\Delta_r} \Upsilon^{\sigma_\Lambda}_\pi(r, u_i) \sum_{v \in L_{u_i}(k-1)} \Upsilon^{\sigma_\Lambda}_{\pi} (u_i, v)  \Delta_v \notag\\
(\text{by \Cref{lem:inf-near}})\quad    &\leq \sum_{i = 1}^{\Delta_r} \abs{h(\log R_{u_i})} \sum_{v \in L_{u_i}(k-1)} \Upsilon^{\sigma_\Lambda}_{\pi} (u_i, v)  \Delta_v \notag\\
     &\leq \sum_{i = 1}^{\Delta_r} h^{\phi}(\log R_{u_i}) \sum_{v \in L_{u_i}(k-1)} \phi(\log R_{u_i})\Upsilon^{\sigma_\Lambda}_{\pi} (u_i, v) \Delta_v,
  \end{align}
  where~\eqref{eq-agu-0/0} holds because of the fact that  $\phi(y) > 0$ for all $y \in [-\infty,+\infty]$ that $|h(y)|> 0$ and the definition of $h^{\phi}(y)$ such that $h^{\phi}(y)=0$ if $h(y)=0$ and $h^\phi(y) = \frac{|h(y)|}{\phi(y)}$ if $h(y)\neq0$.
We have for any integer $k \geq 1$,
\begin{align}
\label{eq-target-0/0}
 \sum_{v \in L_r(k)} \Delta_v \cdot \Upsilon^{\sigma_\Lambda}_\pi(r, v) \leq \Delta_r \cdot \max_{1 \leq i \leq \Delta_r}\left\{h^\phi(\log R_{u_i}) \right\} \cdot \max_{1 \leq i\leq \Delta_r} \left\{\sum_{v \in L_{u_i}(k-1)} \phi(\log R_{u_i})\Upsilon^{\sigma_\Lambda}_{\pi} (u_i, v)  \Delta_v\right\},
\end{align}
To bound the right-hand-side of~\eqref{eq-target-0/0}, we claim that for all $w \in V \setminus (\{r\} \cup \Lambda)$ and $k\ge 0$,
  \begin{align} \label{eq:aux-decay-on-tree}
   \sum_{v \in L_w(k)} \Delta_v \cdot \phi(\log R_w) \Upsilon^{\sigma_\Lambda}_\pi(w, v)
    &\leq \max_{v \in L_w(k)} \{\Delta_v \cdot \phi(\log R_v)\} \cdot \tp{1 - \alpha}^k.
  \end{align}
Combining~\eqref{eq-target-0/0} and~\eqref{eq:aux-decay-on-tree}, we have
  \begin{align*}
    \sum_{v \in L_r(k)} \Delta_v \cdot \Upsilon^{\sigma_\Lambda}_\pi(r, v)
     &\leq \Delta_r \cdot \max_{1 \leq i \leq \Delta_r}\left\{h^\phi(\log R_{u_i}) \right\} \cdot \max_{1 \leq i \leq \Delta_r} \max_{v \in L_{u_i}(k-1)} \{\Delta_v \cdot \phi(\log R_v)\} \cdot (1 - \alpha)^{k-1} \\
    &= \Delta_r \cdot \max_{1 \leq i \leq \Delta_r}\left\{h^\phi(\log R_{u_i}) \right\} \cdot \max_{v \in L_r(k)} \{\Delta_v \cdot \phi(\log R_v)\} \cdot (1 - \alpha)^{k-1} \\
 &= \Delta_r \cdot \max_{u \in L_r(1), v \in L_r(k)} \left\{\Delta_v \cdot h^\phi(\log R_{u_i})  \cdot  \phi(\log R_v)\right\} \cdot (1 - \alpha)^{k-1} \\
(\star) \quad   
&\leq \Delta_r \cdot 2c \cdot (1 - \alpha)^{k-1}.
  \end{align*}
Inequality $(\star)$ holds because of the boundedness property and the fact that $u,v \in V \setminus \Lambda$ are free variables. 
  
We now use induction on $k$ to prove \eqref{eq:aux-decay-on-tree}.
The base case is $k = 0$. We have
\begin{align*}
    \Delta_w \cdot \phi(\log R_w) \Upsilon^{\sigma_\Lambda}_\pi(w, w) = \Delta_w \cdot \phi(\log R_w),
  \end{align*}
which holds because $\Upsilon^{\sigma_\Lambda}_{\pi}(w, w) = 1$.

Now, suppose \eqref{eq:aux-decay-on-tree} holds for some integer $k - 1 \geq 0$. 
Fix $w \in V \setminus (\Lambda \cup \{r\})$. Let $w_1,w_2,\ldots,w_{d_w}$ denote the children of $w$ in $T_w$. Since $w \neq r$, we have $1\leq  d_w < \Delta$. We assume $d_w > 0$ because \eqref{eq:aux-decay-on-tree} holds trivially if $d_w = 0$.
By \Cref{lem:inf-chain} and \Cref{lem:inf-near}, we have
  \begin{align*}
    \sum_{v \in L_w(k)} \Delta_v \cdot \phi(\log R_w) \Upsilon^{\sigma_\Lambda}_\pi(w, v)
    &= \sum_{i=1}^{d_w} \phi(\log R_w)\Upsilon^{\sigma_\Lambda}_\pi(w, w_i) \sum_{v \in L_{w_i}(k-1)} \Upsilon^{\sigma_v}_\pi(w_i, v) \cdot \Delta_v \\
    &\leq \sum_{i=1}^{d_w} \phi(\log R_w)\abs{h(\log{R_{w_i}})} \sum_{v \in L_{w_i}(k-1)} \Upsilon^{\sigma_v}_\pi(w_i, v) \cdot \Delta_v.
  \end{align*}
By a similar argument in~\eqref{eq-agu-0/0}, we have 
\begin{align*}
 \sum_{v \in L_w(k)} \Delta_v \cdot \phi(\log R_w) \Upsilon^{\sigma_\Lambda}_\pi(w, v)	\leq \sum_{i=1}^{d_w}\phi(\log R_w) h^\phi(\log R_{w_i}) \sum_{v \in L_{w_i}(k-1)} \phi(\log R_{w_i})\Upsilon^{\sigma_v}_\pi(w_i, v) \cdot \Delta_v,
\end{align*}
By the induction hypothesis, it holds that
  \begin{align*}
    \sum_{v \in L_w(k)} \Delta_v \cdot \phi(\log R_w) \Upsilon^{\sigma_\Lambda}_\pi(w, v)
    &\leq \sum_{i=1}^{d_w} \phi(\log R_w) h^\phi(\log R_{w_i}) \max_{v \in L_{w_i}(k-1)}\{\Delta_v \cdot \phi(\log R_v)\} \cdot (1 - \alpha)^{k-1} \\
    &\leq \max_{v \in L_w(k)}\{\Delta_v \cdot \phi(\log R_v)\} \cdot (1 - \alpha)^{k-1} \cdot \sum_{i=1}^{d_w} \phi(\log R_w) h^\phi(\log R_{w_i})   \\
    &\leq \max_{v \in L_w(k)}\{\Delta_v \cdot \phi(\log R_v)\} \cdot (1 - \alpha)^{k},
  \end{align*}
 where the last inequality holds because $w \in V\setminus \Lambda$ and $d_w \geq 1$, and $\phi$ satisfies the $\alpha$-Contraction property on $V\setminus \Lambda$.
\end{proof}

\begin{proof}[Proof of \Cref{theorem-good-potential-imply-SI-CLV}]
Recall that $\+I = (V, E, \beta, \gamma, \tp{\lambda_v}_{v \in V})$ is a two-spin system with Gibbs distribution $\pi$ where $0 \leq \beta \leq \gamma$, $\gamma > 0$, $\lambda_v > 0$ for all $v \in V$, and $\phi$ is an $(\alpha, c)$-function with respect to $\+I$.
For any $r \in V$, $\Lambda \subseteq V \setminus \{r\}$, and $\sigma_\Lambda \in \Omega(\pi_\Lambda)$, let $\+I_T$ denote the two-spin system on the SAW tree $\TSAW(G, r)$ whose Gibbs distribution is $\pi_\SAW$.
Recall that we use $\Psi^{\sigma_\Lambda}_T$ to denote $\Psi^{\sigma_\Lambda}_{\pi_\SAW}$ and $\Psi^{\sigma_\Lambda}_G$ to denote $\Psi^{\sigma_\Lambda}_\pi$ respectively.
Then, by \Cref{lem:tree-inf-preserve}, we have
\begin{align} \label{eq:aux-B3-1}
  \sum_{v \in V\setminus \Lambda} \Delta_{G,v} \cdot \Psi^{\sigma_\Lambda}_G(r, v) &\leq \sum_{\hat{v} \in V_\SAW \setminus \Lambda_\SAW} \Delta_{T,\hat{v}} \cdot \Psi^{\sigma_\Lambda}_T (r, \hat{v}),
\end{align} 
By the SAW-tree construction, for any $v \in V$, if $\hat{v}$ is a free copy of $v$ in $\+I_T$, then $v$ and $\hat{v}$ have the same local field and degree. It is straightforward to verify that the function $\phi$ is also an $(\alpha, c)$-function with respect to $\+I_T$ on $V_\SAW \setminus \Lambda_\SAW$.
\Cref{lem:inf-on-tree-instance} can be applied to $\+I_T$, and it holds that
\begin{align}\label{eq:aux-B3-2}
  \sum_{\hat{v} \in V_\SAW \setminus \Lambda_\SAW} \Delta_{T, \hat{v}} \cdot \Psi^{\sigma_\Lambda}_T(r, \hat{v}) &\leq \Delta_r \cdot \frac{2c}{\alpha}.
\end{align}
Finally, \Cref{theorem-good-potential-imply-SI-CLV} is proved by combining \eqref{eq:aux-B3-1} and \eqref{eq:aux-B3-2} with \Cref{lem:spec-ratio-to-row-sum}.
\end{proof}

\subsection{Proof of \Cref{lm:contraction}}\label{proof-lm:contraction}
Recall the definition of the tree recursion for marginal ratios $F_{\lambda,d}(x_1,\ldots,x_d)$ in~\eqref{eq:tree-recursion-marginal-ratio}. For $\vec{x}=(x_1,\ldots,x_d)\in[0,+\infty]^d$, let
\begin{align*}
\alpha_d(x_1,\ldots,x_d)
\triangleq 
\sqrt{\frac{(1-\beta \gamma)F_{\lambda,d}(\vec{x})}{\tp{\beta F_{\lambda,d}(\vec{x}) + 1} \tp{F_{\lambda,d}(\vec{x})+ \gamma}}} 
\sum_{i=1}^d \sqrt{\frac{(1-\beta \gamma)x_i}{(\beta x_i+1)(x_i+\gamma)}}.
\end{align*}

Recall the definition of the tree recursion for log-marginal-ratios $H_{\lambda,d}(y_1,\ldots,y_d)$ in~\eqref{eq:tree-recursion-log-marginal-ratio}, function $h(y)$ in~\eqref{eq-def-function-h} and $h^{\phi}(y)$~\eqref{eq:definition-h-phi}.
For any integer $d \ge 1$, any $y_1,y_2,\ldots,y_d \in [-\infty,+\infty]^d$ and $y = H_{\lambda,d}(y_1,y_2,\ldots,y_d)$, it holds that
\begin{align}\label{eq:log-to-normal}
 \phi(y)\sum_{i=1}^{d} h^{\phi}\tp{y_i}  = \sum_{i=1}^{d} \sqrt{|h(y)h(y_i)|}
 =\alpha_d(e^{y_1},\ldots,e^{y_d}).
\end{align}
%
For $x\in[0,+\infty]$, let $\alpha_d(x)$ denote the symmetric version of $\alpha_d(x_1,\ldots,x_d)$, specifically,
\[
\alpha_d(x)\triangleq \alpha_d(\underbrace{x,\ldots,x}_{d})
=
\sqrt{f_d(F_d(x))} 
\sqrt{f_d(x)},
\]
where $F_d(x)\triangleq\lambda\tp{\frac{\beta x+1}{x+\gamma}}^d$, defined in~\eqref{eq:tree-recursion}, is the symmetric version of the tree recursion for marginal ratios $F_{\lambda,d}(\vec{x})$; and $f_d(d)\triangleq \frac{d(1-\beta\gamma)x}{(\beta x+1)(x+\gamma)}$ is defined in \eqref{eq-def-fd}.
The following symmetrization was well known.
\begin{lemma}[\text{\cite[Lemma 13]{LLY13} }]\label{lm:lly-13}
  Let $d \ge 1$ be an integer, and let $\beta,\gamma,\lambda$ be real numbers satisfying that $0\le \beta\le\gamma$, $\gamma> 0$, $\lambda\ge 0$ and $\beta\gamma<1$.
  For any $x_1,x_2,\ldots,x_d \in [0,+\infty)$, there exists $\overline{x} \in [0,+\infty)$ such that
  \begin{align*}
    \alpha_d(x_1,x_2,\ldots,x_d) \le \alpha_d(\overline{x}).
  \end{align*}
\end{lemma}

Furthermore, it was known that $\alpha_d(x)$ is bounded by  $\sqrt{f_d(\hat{x}_d)}$ when $(\beta,\gamma,\lambda)$ is $d$-unique.
\begin{lemma}[\text{\cite{LLY13}}]\label{lm:lly-14}
  Let $d \ge 1$ be an integer, and let $\beta,\gamma,\lambda$ be real numbers satisfying that $0\le \beta\le\gamma$, $\gamma> 0$, $\lambda\ge 0$ and $\beta\gamma<1$.
  If $(\beta,\gamma,\lambda)$ is $d$-unique with gap $\delta$, then for any $x\in[0,+\infty)$, 
  \[
  \alpha_d(x)=\sqrt{f_d(F_d(x))}\sqrt{f_d(x)}\le \sqrt{f_d(\hat{x}_d)}\le\sqrt{1-\delta},
  \]
  where $\hat{x}_d$ denotes the unique positive fixed point of $F_d(x)=\lambda\tp{\frac{\beta x+1}{x+\gamma}}^d$ and $f_d(\hat{x}_d)=|F'_d(\hat{x}_d)|=\frac{d(1-\beta\gamma)\hat{x}}{(\beta\hat{x}+1)(\hat{x}+\gamma)}$.
\end{lemma}

\begin{remark}
In~\cite{LLY13}, the lemma was proved by assuming the up-to-$\Delta$ uniqueness (but the conclusion was also stronger which held for all $1\le d<\Delta$).
\Cref{lm:lly-14} can be proved by going through that proof.
\end{remark}

\begin{proof}
By taking derivative
\begin{align}\label{eq:decay-factor-gradient}
\alpha_d'(x)=c(x)\cdot
\left(\frac{\gamma-\beta x^2}{d(1-\beta\gamma)x}-\frac{\gamma-\beta F_d(x)^2}{(\beta F_d(x)+1)(F_d(x)+\gamma)}\right),
\end{align}
where $c(x)>0$ is a function that is always positive for $x\in[0,+\infty)$.
Note that $\frac{\gamma-\beta x^2}{d(1-\beta\gamma)x}$ is decreasing in $x$ and ranges from $+\infty$ to $-\infty$ and $F_d(x)$ is decreasing in $x$ and has a bounded range $[\lambda\beta^d,\lambda\gamma^{-d}]$. 
Then the equation
\begin{align}\label{eq:x_d}
\frac{\gamma-\beta x^2}{d(1-\beta\gamma)x}=\frac{\gamma-\beta F_d(x)^2}{(\beta F_d(x)+1)(F_d(x)+\gamma)}.
\end{align}
has a unique positive solution $x_d\in(0,+\infty)$, at which $\alpha_d(x)$ achieves its maximum.
Substituting $(\beta F_d(x_d)+1)(F_d(x_d)+\gamma)=\frac{d(1-\beta\gamma)x_d(\gamma-\beta F_d(x_d)^2)}{\gamma-\beta x_d^2}$ in $\alpha_d(x_d)$, we have
\begin{align*}
\alpha_d(x_d)
= \sqrt{d(1-\beta\gamma)\cdot\frac{(\gamma-\beta x_d^2)}{(\beta{x}_d+1)({x}_d+\gamma)}\cdot\frac{F_d(x_d)}{\left(\gamma-\beta F_d(x_d)^2\right)}}
\triangleq
\tilde{\alpha}_d(x_d).
\end{align*}
We then claim that if $(\beta,\gamma,\lambda)$ is $d$-unique, then $\tilde{\alpha}_d(x_d)\le\tilde{\alpha}_d(\hat{x}_d)$ for the fixed point $\hat{x}_d=F_d(\hat{x}_d)$.
To see that this is sufficient to prove the lemma, note that $\alpha_d(x)\le\alpha_d(x_d)=\tilde{\alpha}_d(x_d)\le\tilde{\alpha}_d(\hat{x}_d)$, and  by substituting $F_d(\hat{x}_d)=\hat{x}_d$, we have 
\[
\tilde{\alpha}_d(\hat{x}_d)
=
\sqrt{\frac{d(1-\beta\gamma)\hat{x}_d}{(\beta\hat{x}_d+1)(\hat{x}_d+\gamma)}}=\sqrt{f_d(\hat{x}_d)}.
\]
Next, we prove this claim. 
It is sufficient to show that $\tilde{\alpha}_d(x)$ is decreasing in $x\in[\hat{x}_d,x_d]$ if $\hat{x}_d\le x_d$ and $\tilde{\alpha}_d(x)$ is increasing in $x\in[x_d,\hat{x}_d]$ if $\hat{x}_d> x_d$.
\begin{itemize}
\item Case 1: $\hat{x}_d\le x_d$. 
In this case, according to \eqref{eq:decay-factor-gradient}, $\alpha'(\hat{x}_d)\ge 0$.
Note that
\begin{align*}
\alpha'(\hat{x}_d)
=
c(\hat{x}_d)
(\gamma-\beta\hat{x}_d^2)
\left(\frac{1}{d(1-\beta\gamma)\hat{x}_d}-\frac{1}{(\beta\hat{x}_d+1)(\hat{x}_d+\gamma)}\right),
\end{align*}
where $c(\hat{x}_d)>0$.
Due to the $d$-uniqueness of $(\beta,\gamma,\lambda)$, we have $|f_d'(\hat{x}_d)|=\frac{d(1-\beta\gamma)\hat{x}_d}{(\beta\hat{x}_d+1)(\hat{x}_d+\gamma)}<1$, that is, $\frac{1}{d(1-\beta\gamma)\hat{x}_d}-\frac{1}{(\beta\hat{x}_d+1)(\hat{x}_d+\gamma)}>0$. 
Therefore $\alpha'(\hat{x}_d)\ge 0$ means that $\gamma-\beta\hat{x}_d^2\ge 0$. 
Since $F_d(x)$ is monotonically decreasing in $x$ and $\hat{x}_d$ is its fixed point, we have
\[
\gamma-\beta F_d({x}_d)^2
\ge
\gamma-\beta F_d(\hat{x}_d)^2
=
\gamma-\beta\hat{x}_d^2
\ge
0.
\] 
Since $x_d$ satisfies \eqref{eq:x_d}, $\gamma-\beta{x}_d^2$ and $\gamma-\beta F_d({x}_d)^2$ must be simultaneously positive or negative, thus it also holds that $\gamma-\beta{x}_d^2\ge 0$. Then 
both $\frac{(\gamma-\beta x^2)}{(\beta{x}+1)({x}+\gamma)}$ and $\frac{F_d(x)}{\left(\gamma-\beta F_d(x)^2\right)}$ are positive and monotonically decreasing in $x\in[\hat{x}_d,x_d]$. Therefore, $\tilde{\alpha}_d(x_d)\le \tilde{\alpha}_d(\hat{x}_d)$.

\item Case 2: $\hat{x}_d> x_d$. By symmetry, we have $\gamma-\beta F_d(\hat{x}_d)^2=\gamma-\beta\hat{x}_d^2<0$, $\gamma-\beta F_d({x}_d)^2<0$, and $\gamma-\beta{x}_d^2<0$. Thus both $\frac{(\gamma-\beta x^2)}{(\beta{x}+1)({x}+\gamma)}$ and $\frac{F_d(x)}{\left(\gamma-\beta F_d(x)^2\right)}$ are negative and monotonically decreasing in $x\in[x_d,\hat{x}_d]$, hence their product is positive and increasing in $x\in[x_d,\hat{x}_d]$.  Therefore, $\tilde{\alpha}_d(x_d)\le \tilde{\alpha}_d(\hat{x}_d)$.
\end{itemize}
\end{proof}

\begin{proof}[Proof of \Cref{lm:contraction}]
We now prove \Cref{lm:contraction}. Given $y_1,y_2,\ldots,y_d \in [-\infty, +\infty]$, define 
\begin{align*}
x_i = \mathrm{e}^{y_i}.	
\end{align*}
If all $y_i \in [-\infty, +\infty)$,  then all  $x_i \in \mathbb{R}_{\geq 0}$. By~\eqref{eq:log-to-normal}, \Cref{lm:lly-13} and \Cref{lm:lly-14}, it holds that
\begin{align}
\label{eq-reals}
 \phi(y)\sum_{i=1}^{d} h^{\phi}\tp{y_i}= \alpha_d(e^{y_1},\ldots,e^{y_i}) =\alpha_d(x_1,\ldots,x_d)  
 \leq \sqrt{1-\delta}.	
\end{align}
The lemma follows.

Suppose there is a subset $S \subseteq [d]$ such that $y_i = +\infty$ for all $i \in S$ and $y_j \in [-\infty, +\infty)$ for all $j \in [d] \setminus S$.
In this case, $x_i = +\infty$ for all $i \in S$, and $x_j \in \mathbb{R}_{\geq 0}$ for all $j \in [d] \setminus S$. 
Define a function
\begin{align*}
g(x) =  \alpha_d(z_1,z_2,\ldots,z_d),
\end{align*}
where $z_i = x$ for all $i \in S$, and $z_j = x_j=\mathrm{e}^{y_j} \in [0,+\infty)$ for all $j \in [d] \setminus S$. By~\eqref{eq:log-to-normal} and definition of $\alpha_d(x_1,\ldots,x_d)$, it is straightforward to verify $\lim_{x \to \infty}g(x)$ exists and it holds that 
\begin{align*}
\alpha_d(e^{y_1},\ldots,e^{y_d})  =\alpha_d(x_1,\ldots,x_d)  = \lim_{x \to \infty}g(x).
\end{align*}
By~\eqref{eq-reals}, for any $x \in [0,+\infty)$, we have $g(x) \leq \sqrt{1-\delta}$. Hence 
\begin{align*}
 \phi(y)\sum_{i=1}^{d} h^{\phi}\tp{y_i} &= \lim_{x \to \infty}g(x) \leq \sqrt{1-\delta}.\qedhere
\end{align*}
\end{proof}
\end{document}